\documentclass[11pt, a4paper]{article}

\usepackage{eurosym}
\usepackage{url}
\usepackage{a4wide}
\usepackage{graphicx}
\usepackage{natbib}
\usepackage{enumitem}
\usepackage{hyperref}
\usepackage{setspace}
\usepackage{multirow, booktabs}
\usepackage{amsmath, amsthm, amssymb, amsfonts}
\usepackage{bm, bbm}
\usepackage[T1]{fontenc}
\usepackage{color}

\allowdisplaybreaks
\DeclareMathAlphabet\mathbfcal{OMS}{cmsy}{b}{n}
\newcommand{\mbf}{\mathbf}
\newcommand{\mc}{\mathcal}
\newcommand{\bmx}{\begin{bmatrix}}
\newcommand{\emx}{\end{bmatrix}}

\renewcommand{\l}{\left}
\renewcommand{\r}{\right}
\def\wh{\widehat}
\def\wt{\widetilde}

\newcommand{\E}[0]{\mathsf{E}}

\newcommand{\p}{\mathsf{P}}
\newcommand{\R}{\mathbb{R}}
\newcommand{\Z}{\mathbb{Z}}

\newcommand{\nn}{\nonumber}
\newcommand{\vech}[0]{\mathsf{Vech}}

\theoremstyle{definition}
\newtheorem{thm}{Theorem}
\theoremstyle{definition}

\theoremstyle{definition}
\newtheorem{lem}[thm]{Lemma}
\theoremstyle{definition}
\newtheorem{prop}[thm]{Proposition}
\theoremstyle{definition}
\newtheorem{assum}{Assumption}
\theoremstyle{remark}

\theoremstyle{definition}
\newtheorem{defn}{Definition}
\theoremstyle{definition}
\newtheorem{ex}{Example}
\theoremstyle{definition}

\graphicspath{{./figs/}}

\begin{document}

\title{Moving sum procedure for multiple change point detection in large factor models}
\author{Matteo Barigozzi$^1$ and Haeran Cho$^2$ and Lorenzo Trapani$^3$}
\maketitle

\begin{abstract}
This paper proposes a moving sum methodology for detecting multiple change points in high-dimensional time series under a factor model, where changes are attributed to those in loadings as well as emergence or disappearance of factors. We establish the asymptotic null distribution of the proposed test for family-wise error control, and show the consistency of the procedure for multiple change point estimation. Simulation studies and an application to a large dataset of volatilities demonstrate the competitive performance of the proposed method.
\end{abstract}

\footnotetext[1]{%
Department of Economics, Universit\`a di Bologna. 
Email: \url{matteo.barigozzi@unibo.it}.}

\footnotetext[2]{%
School of Mathematics, University of Bristol. Email: %
\url{haeran.cho@bristol.ac.uk}.}

\footnotetext[3]{%
Department of Economics and Management, Universita' di Pavia and School of Business, University of Leicester. Email: %
\url{lorenzo.l.trapani@gmail.com}.}


\noindent\textit{Keywords:} data segmentation, MOSUM, factor model

\noindent\textit{MSC subject classification:} 62M10, 62H25 

\section{Introduction}

Factor models are, arguably, one of the most frequently employed tools to model and carry out inference when large-dimensional, vector-valued time series are available. 
Whilst a comprehensive review of the literature goes well above and beyond the scope of this paper, it is worth noting that factor models have a long history, dating back at least to the seminal paper by \citet{spearman}. 
Since being popularised by the contribution by \citet{chamberlainrothschild83}, and since the development of the asymptotic theory to analyse large dimensional factor models \citep{bai03}, their usage has become \textit{de rigueur} in social sciences and economics where they have been applied to diverse fields such as business cycle analysis, asset pricing and economic monitoring and forecasting (see, \textit{inter alia}, the review by \citet{SW11} for a comprehensive list of references).

As the time dimension increases, it is inevitable that any model (including factor models) may undergo changes in its structure which, in turn, may affect the properties of estimation techniques and hamper predictive ability. 
There is a huge literature on both ex-post and online detection of change points in various contexts, and we refer to the recent book by \citet{horvath2024change} for a review focussed in particular on ex-post detection, and to the paper by \citet{aue2024state} for a review of the literature on online (sequential) detection. 

Where factor models are concerned, on the one hand, the analysis in \citet{SW09} and \citet{bates2013} shows that, when changes are ``sufficiently
small'' and do not involve a change in the number of common factors, the presence of change points is inconsequential for inference on the factor spaces. 
On the other hand, in several applications, there is no guarantee that the factor structure only undergoes small changes and that the number of common factors is time-invariant. 
Hence, it is not surprising that the literature has developed several methods to test, retrospectively, whether there is a change point or not. 
\citet{breitung2011}, \citet{corradi2014}, \citet{han2014}, \citet{yamamoto2015}, \citet{baltagi2017identification}, \citet{cheng2016} and \citet{baiduan} propose several tests to check if there is a break, based on the idea that a change point in the loadings is observationally equivalent to a change in the covariance matrix of the common factors.
Thus assuming homoscedastic factors, a change in the covariance matrix of the estimated common factors
can only be due to a break in the loadings. 
See also 
\citet{barigozzi2020sequential} and \citet{he2021online} for alternative formulation of change point tests in an online context. 
Whilst the literature on change point testing is well-established, there is relatively little work on
detection and estimation of (possibly) multiple change points in the factor structure with some exceptions \citep{barigozzi2018simultaneous, li2023detection, baiduan}.

In this paper, we make an advance on the change point literature by proposing a method to estimate the number and locations of (possibly) multiple change points under a factor model. 
Specifically, we propose a procedure based on the so-called
MOving SUMs (MOSUM) process which, since the seminal contribution by \citet{huvskova2001permutation}, has been shown to have desirable properties in multiple change point estimation under weak assumptions on the distribution of the data.
Versatility of the MOSUM procedure has been demonstrated in the context of detecting changes in the mean (\citealp{huvskova2001permutation}; \citealp{kirch}), trend \citep{kim2024moving} and the drift of stochastic processes \citep{kirch2021moving}; see also \cite{kirch2024data} for a general framework for multivariate time series segmentation based on an estimating function.
Some recent contributions extend the use of this methodology to high-dimensional time series (\citealp{cho2022high}; \citealp{cho2023high}), demonstrating its scalability and suggesting that it is worth exploring the performance of a MOSUM procedure in the context of large factor models.

We build on that, as discussed above, a change in the
loadings is observationally equivalent to a change in the covariance matrix of the common factors, and construct a MOSUM statistic based on the (moving) partial sums of the outer products of the estimated factors.
In this respect, our statistic is akin to the one based on the (maximally selected) sequence of likelihood ratio tests proposed e.g.\ in \citet{duan2023quasi}, save for the fact that it is based on using \textit{moving}, as opposed to \textit{cumulative}, partial sums. 
Upon testing for changes in the factor model using the MOSUM process, if the null hypothesis of no change is rejected, we perform multiple change point detection by estimating both the total number and the locations of the breaks.
We establish the asymptotic null distribution of the MOSUM process as well as the consistency of the proposed change point detection procedure in estimating the total number and locations of the change points, with the accompanying rate of estimation for individual breaks.
In doing so, we derive a Strong Invariance Principle for the MOSUM process based on the outer products of the estimated factors, which may be of independent interest.


The remainder of the paper is organised as follows. We describe the MOSUM procedure in Section~\ref{mosum}. Assumptions, the limiting distribution under the null, and the consistency under the alternative are reported in Section~\ref{theory}. 
Sections~\ref{simulations} and~\ref{sec:real} provide findings from simulation studies and an application to financial data. 
Section~\ref{sec:conc} concludes the paper, and all the proofs of theoretical results are given in Appendix.
An implementation of the proposed method is available at \url{https://github.com/haeran-cho/mosum.fts}.

\paragraph{Notations.}

For a random variable $X$, we have $|X|_{\nu} = ( \mathsf{E}(|X|^{\nu}))^{1/\nu}$. 
We denote: weak convergence on the space $\mathcal{D}[0, 1]$ with $\overset{\mathcal{D}}{\rightarrow}$; $\overset{\mathcal{D}}{=}$ is equality in distribution; $\rightarrow$ is the ordinary limit; $\Gamma(\cdot)$ is the Gamma function; $\log(x)$ is the natural logarithm of $x > 0$; $\mathbb{I}_{\mathcal{A}}$ is an indicator function satisfying $\mathbb{I}_{\mathcal{A}} = 1$ if the event $\mathcal{A}$ is true and $\mathbb{I}_{\mathcal{A}} = 0$ otherwise.
By $\mathbf{I}$, $\mathbf{O}$ and $\mathbf{0}$, we denote an
identity matrix, a matrix of zeros and a vector of zeros whose dimensions depend on the context. 
For a matrix $\mathbf{A} \in \R^{m \times n}$, we denote by $\mathbf{A}^{\top}$ its transpose and $\Vert \mathbf{A}\Vert $ its Euclidean norm, with $\Lambda_{\max}(\mathbf{A})$ and $\Lambda_{\min}(\mathbf{A})$ denoting its largest and the
smallest eigenvalues in modulus. 
When $m = n$, we denote by $\mathsf{Vech}(\mathbf{A})$ the
vector of length $m(m+1)/2$ that stacks the elements on and below the main
diagonal of $\mathbf{A}$, and $\mathbf{L}_{m}$ the $m(m+1)/2\times m^{2}$-matrix satisfying $\mathsf{Vech}(\mathbf{A})=\mathbf{L}_{m}\mathsf{Vec}(\mathbf{A})$, where $\mathsf{Vec}(\cdot)$ is an operator that stacks the
columns of the matrix into a vector. Conversely, we define $\mathbf{K}_{m}\in \mathbb{R}^{m^{2}\times m(m+1)/2}$ satisfying $\mathsf{Vec}(\mathbf{A})=\mathbf{K}_{m}\mathsf{Vech}(\mathbf{A})$, see Appendix~A.12 of \cite{lutkepohl2005new}. 
Finally, given two sequences $\{a_{n}\}$ and $\{b_{n}\}$, we
write $a_{n}=O(b_{n})$ if, for some finite positive constant $C$ there exists $N\in \mathbb{N}_{0}=\mathbb{N}\cup \{0\}$ such that $|a_{n}||b_{n}|^{-1}\leq C$ for all $n\geq N$.

\section{MOSUM procedure for data segmentation}
\label{mosum}

\subsection{Model}

Let $\{\mathbf{X}_{t},\,1\leq t\leq T\}$ denote an $N$-dimensional time
series that admits the following factor model representation with $R$ change points, as
\begin{equation}
\mathbf{X}_{t}=\left\{ 
\begin{array}{ll}
\bm\Lambda_{0}\mathbf{f}_{t}+\mathbf{e}_{t} & \text{for \ }k_{0}+1=1\leq
t\leq k_{1}, \\ 
\bm\Lambda_{1}\mathbf{f}_{t}+\mathbf{e}_{t} & \text{for \ }k_{1}+1\leq
t\leq k_{2}, \\ 
\vdots &  \\ 
\bm\Lambda_{R}\mathbf{f}_{t}+\mathbf{e}_{t} & \text{for \ }k_{R}+1\leq
t\leq k_{R+1}=T. \label{eq:model:one}
\end{array}%
\right.
\end{equation}%
Here, the matrix of loadings $\bm\Lambda_{j}=[\bm\lambda_{j,1},\ldots ,\bm%
\lambda_{j,N}]^{\top }\in \mathbb{R}^{N\times r}$ has column rank $r_{j}\leq r$ fixed for all $N$, and ``loads'' the vector of random factors $\mathbf{f}_{t}=(f_{1,t},\ldots ,f_{r,t})^{\top}$ onto the cross-sections of 
$\mathbf{X}_{t}$, and $\mathbf{e}_{t}=(e_{1,t},\ldots ,e_{N,t})^{\top}$
denotes the idiosyncratic component. 
The model in~\eqref{eq:model:one} allows for the changes due to emergence or disappearance of factor(s) by permitting the ranks of $\bm\Lambda_j$ to vary,
as well as rotational changes in the loading matrices.
At the same time, the model is not identifiable in that the changes in the loadings are
(observationally) equivalent to changes in the second-order properties
of the common factor $\mathbf{f}_{t}$ -- a fact, as mentioned in Introduction,
frequently explored in the relevant literature for developing change point tests; 
see e.g.\ \citet{baiduan}. 
We further illustrate this point in the following example.

\begin{ex}
\label{ex:one} Consider the factor model in~\eqref{eq:model:one} with $R = 1$ and the
single change point at $k_{1}=k^{\ast}$. For any $\bm\Lambda_{j}\in 
\mathbb{R}^{N\times r_{j}}$, $j=0,1$, we can find $\bm\Lambda \in \mathbb{R}%
^{N\times r}$ of full column rank with $r\geq \max (r_{0},r_{1})$ such that $\bm\Lambda_{j}=\bm\Lambda \mathbf{A}_{j}$ for $\mathbf{A}_{j}\in \mathbb{R}^{r\times r_j}$ of rank $r_{j}$. Then, we can re-write~\eqref{eq:model:one} as 
\begin{align*}
& 
\begin{bmatrix}
\mathbf{X}_{0:k^{\ast }} \\ 
\mathbf{X}_{k^{\ast }:T}%
\end{bmatrix}%
=%
\begin{bmatrix}
\mathbf{F}_{0:k^{\ast }}\bm\Lambda_{0}^{\top } \\ 
\mathbf{F}_{k^{\ast }:T}\bm\Lambda_{1}^{\top }%
\end{bmatrix}%
+\mathbf{E}, \text{ \ where \ } 
\mbf E^\top = \begin{bmatrix}
\mbf e_1, \ldots, \mbf e_T
\end{bmatrix}, \\
& \underset{(b-a)\times N}{\mathbf{X}_{a:b}}=%
\begin{bmatrix}
\mathbf{X}_{a+1}^{\top } \\ 
\vdots \\ 
\mathbf{X}_{b}^{\top }%
\end{bmatrix}
\text{ \ and \ }
\underset{(b-a)\times r}{\mathbf{F}_{a:b}}=%
\begin{bmatrix}
\mathbf{f}_{a+1}^{\top } \\ 
\vdots \\ 
\mathbf{f}_{b}^{\top }%
\end{bmatrix}
\text{ \ for all \ }0\leq a<b\leq T.
\end{align*}
It follows that 
\begin{equation*}
\begin{bmatrix}
\mathbf{X}_{0:k^{\ast }} \\ 
\mathbf{X}_{k^{\ast }:T}%
\end{bmatrix}%
=%
\begin{bmatrix}
\mathbf{F}_{0:k^{\ast }}\bm\Lambda_{0}^{\top } \\ 
\mathbf{F}_{k^{\ast }:T}\bm\Lambda_{1}^{\top }%
\end{bmatrix}%
+\mathbf{E}=%
\begin{bmatrix}
\mathbf{F}_{0:k^{\ast }}\mathbf{A}_{0}^{\top } \\ 
\mathbf{F}_{k^{\ast }:T}\mathbf{A}_{1}^{\top }%
\end{bmatrix}%
\bm\Lambda^{\top }+\mathbf{E} =: \mathbf{G}\bm\Lambda^{\top }+\mathbf{E},
\end{equation*}%
which is an observationally equivalent representation with a time-invariant
loading matrix $\bm\Lambda$ and pseudo factors of dimension $r$ contained in $\mathbf{G}$.
\end{ex}

We extend the observation made in Example~\ref{ex:one} to the multiple
change point setting and re-write the model in~\eqref{eq:model:one} as 
\begin{align}  \label{eq:model:two}
\mathbf{X}_t = \bm\Lambda \sum_{j = 0}^R \mathbf{A}_j \mathbf{f}_t \cdot 
\mathbb{I}_{\{ k_j < t \le k_{j + 1} \}} + \mathbf{e}_t =: \bm\Lambda \mathbf{%
g}_t + \mathbf{e}_t
\end{align}
with $\bm\Lambda_j = \bm\Lambda \mbf A_j$ for all $0 \le j \le R$.
Under the homoscedasticity of $\{\mathbf{f}_t\}$, we have 
\begin{align}
\mathsf{Cov}(\mathbf{g}_{k_j}) = \mathbf{A}_{j - 1} \mathsf{Cov}(\mathbf{f}_0) 
\mathbf{A}_{j - 1}^\top \ne \mathbf{A}_j \mathsf{Cov}(\mathbf{f}_0) \mathbf{A%
}_j^\top = \mathsf{Cov}(\mathbf{g}_{k_j + 1})  \notag
\end{align}
for $1 \le j \le R$. Then, the problem of detecting the multiple
change points in the loadings under~\eqref{eq:model:two}, becomes that of
detecting change points in the covariance of the pseudo factors~$\{\mathbf{g}_t\}$.

\subsection{Methodology}
\label{sec:method}

Let us suppose that the number of pseudo factors $r$ is known. Then, we can
estimate the pseudo factors $\mathbf{g}_{t}$ as, up to an invertible transformation, $\sqrt{T}$ times the $r$ leading eigenvectors
of the $T\times T$ matrix $(NT)^{-1}\mathbf{X}\mathbf{X}^{\top}$  with $\mathbf{X}=\mathbf{X}_{0:T}$. Denoting such estimator
by $\widehat{\mathbf{g}}_{t}$, we propose to test for a change under the model in~\eqref{eq:model:one} and, if any, estimate the multiple change points, by scanning the data
using the following MOSUM statistic: 
\begin{align}
\mathcal{T}_{N,T,\gamma }(k)& =\left\vert \mathbf{M}_{N,T,\gamma }^{\top }(k)%
\mathbf{V}_{k}^{-1}\mathbf{M}_{N,T,\gamma }(k)\right\vert^{1/2}\text{ \ for
\ }\gamma \leq k\leq T-\gamma ,\text{ \ with}  \label{mos} \\
\mathbf{M}_{N,T,\gamma }(k)& =\frac{1}{\sqrt{2\gamma }}\mathsf{Vech}\left(
\sum_{t=k+1}^{k+\gamma }\widehat{\mathbf{g}}_{t}\widehat{\mathbf{g}}%
_{t}^{\top }-\sum_{t=k-\gamma +1}^{k}\widehat{\mathbf{g}}_{t}\widehat{\mathbf{g}}_{t}^{\top }\right) ,  \notag
\end{align}
where $\gamma \ge 1$ is a pre-selected bandwidth.
The matrix $\mathbf{V}_{k}$ denotes the long-run
covariance matrix of $\mathbf{M}_{N,T,\gamma }(k)$ which, in the absence of
any change point, satisfies $\mathbf{V}_k \equiv \mathbf{V}$ for all $k$, with $\mathbf{V}$ defined explicitly in Theorem~\ref{thm:null} below.

For testing the null hypothesis of no change point, $\mathcal{H}_0: \, R = 0$, 
we compare the maximum of the MOSUM statistics against some threshold, say $D_{T, \gamma}$,
and reject $\mathcal{H}_0$ if 
\begin{align*}
\mathcal{T}_{N, T, \gamma} = \max_{\gamma \le k \le T - \gamma} \mathcal{T}%
_{N, T, \gamma}(k) > D_{T, \gamma}.
\end{align*}
In Theorem~\ref{thm:null} below, we derive the
asymptotic null distribution of $\mathcal{T}_{N, T, \gamma}$, which enables selecting $D_{T, \gamma}$ as its upper $\alpha$-quantile at a given significance level $\alpha \in (0, 1)$.

Beyond testing for any change, we propose to detect and locate the
multiple change points by adopting an approach proposed by \cite{kirch} in
the univariate mean change point setting. Simply put, we select every local
maximiser of $\mathcal{T}_{N,T,\gamma }(k)$ over a sufficiently large enough
interval at which $\mathcal{T}_{N,T,\gamma }(k)$ exceeds the threshold.
Specifically, for some fixed $\eta \in (0,1]$, we set as a change point estimator every $\widehat{k}$
that simultaneously satisfies 
\begin{equation}
\widehat{k}={\arg \max }_{k:\,|\widehat{k}-k|\leq \eta \gamma}\,\mathcal{T}%
_{N,T,\gamma }(k)\text{ \ and \ }\mathcal{T}_{N,T,\gamma }(\widehat{k}%
)>D_{T,\gamma }. \label{eq:eta}
\end{equation}%
Denoting such estimators by $\widehat{k}%
_{j},\,1\leq j\leq \widehat{R}$, their total number $\widehat{R}$ is the
estimator of the total number of change points $R$.

For the implementation of the
MOSUM procedure, we require an estimator of the long-run covariance matrix $\mathbf{V}_{k}$.
While we allow it to be location-dependent to account for the heteroscedasticity in the presence of change points, estimating a long-run covariance matrix of multivariate time series is well-known to be highly challenging.
In the current setting, this is augmented by that the computation of $\mc T_{N, T, \gamma}(k)$ calls for the inverse of $\mathbf{V}_k$, which may bring further numerical instabilities.
Therefore, we opt to use the following HAC-type estimator in place of $\mathbf{V}_k$, 
\begin{align}
\widehat{\mathbf{V}}& =\widehat{\bm\Gamma }(0)+\sum_{\ell =1}^{m}\left( 1-%
\frac{\ell }{m+1}\right) \left( \widehat{\bm\Gamma }(\ell )+\widehat{\bm%
\Gamma }^{\top }(\ell )\right) ,\text{ \ where}  \label{v-hat} \\
\widehat{\bm\Gamma }(\ell )& =\frac{1}{T}\sum_{t=\ell +1}^{T}\mathsf{Vech}%
\left( \widehat{\mathbf{g}}_{t}\widehat{\mathbf{g}}_{t}^{\top }-\mathbf{I}_{r}\right) \mathsf{Vech}\left( \widehat{\mathbf{g}}_{t-\ell }\widehat{\mathbf{g}}_{t-\ell }^{\top }-\mathbf{I}_{r}\right)^{\top}  \notag
\end{align}
with some bandwidth $m \ge 1$.
We later show that the matrix $\widehat{\mathbf{V}}$
provides a consistent estimator of the long-run covariance matrix of $\mathbf{M}_{N,T,\gamma }(k)$ under $\mathcal{H}_{0}$ (see Proposition~\ref{cor:V}).
The Bartlett kernel in~\eqref{v-hat} is only one of
the many possible choices, and we refer to \citet{baiduan} for a comprehensive analysis of various kernel-based estimators of $\mathbf{V}$.
For consistency of multiple change point detection, we only require that a positive definite matrix is used in place of $\mbf V_k$ (Theorem~\ref{thm:consistency}).
By default, we propose to adopt $\wh{\mbf V}$ (or its diagonal entries) which is shown to work well; see Section~\ref{simulations} for further discussions.

\section{Theoretical properties}
\label{theory}

\subsection{Assumptions}

We begin by providing a definition of $\mathcal{L}_{\nu}$\textit{-decomposable Bernoulli shifts}.

\begin{defn}
\label{bernoulli} {\it The $d$-dimensional sequence $\left\{ \mbf m_{t},-\infty
<t<\infty \right\} $ forms an $L_{\nu}$-decomposable Bernoulli shift if and
only if it holds that $\mbf m_{t}=h(\bm\eta_{t},\bm\eta_{t-1},\ldots )$, where (i) $h:\, \mathbb{S}^{\infty }\rightarrow \mathbb{R}^{d}$ is a non random measurable function; (ii) $\{\bm\eta_{t}\}_{t\in \mathbb{Z}}$ is an i.i.d.\ sequence with values in a
measurable space $\mathbb{S}$; (iii) $\mathsf{E}(\mbf m_{t})=0$ and $| \Vert \mbf m_{t} \Vert |_{\nu }<\infty $; and
(iv) $| \Vert \mbf m_{t}- \mbf m_{t,l}^{\ast } \Vert |_{\nu }\leq c_{0}l^{-a}$ for some $c_{0}>0$ and 
$a>0$, where $\mbf m_{t,l}^{\ast }=h(\bm\eta_{t},\ldots ,\bm\eta_{t-l+1},\bm\eta
_{t-l,t,l}^{\ast },\bm\eta_{t-l-1,t,l}^{\ast },\ldots )$, with $\{\bm\eta
_{s,t,l}^{\ast },\,-\infty <s,l,t<\infty \}$ that are i.i.d.\ copies of $\bm\eta_{0}$ which are independent of $\{\bm\eta_{t}\}_{t\in \mathbb{Z}}$.}
\end{defn}

Since the seminal works by \citet{wu2005} and \citet{berkeshormann} (see
also \citealp{hormann2009}), decomposable Bernoulli shifts have proven a very convenient way to model and study dependent time series, mainly due to
their generality and since it is much easier to verify whether a
sequence forms a decomposable Bernoulli shift than, e.g.\ verifying mixing
conditions. Virtually all of the most commonly employed models in econometrics
and statistics can be shown to generate decomposable Bernoulli shifts, such
as ARMA and (G)ARCH processes, non-linear time series models (e.g.\ random
coefficient autoregressive models and threshold models), Volterra series and
data generated by dynamical systems; see \citet{berkeshormann}, \citet{aue09}
and \citet{linliu}.

We establish the theoretical properties of the proposed MOSUM procedure
under the following assumptions.

\begin{assum}
\label{factors}

\begin{enumerate}[label = (\roman*)]

\item \label{assum:factors:one} \textit{
There exists some fixed $\epsilon \in (0, 1)$ such that $\{ \mathbf{f}_{t} \}_{t \in \mathbb{%
Z}}$ is an $\mathcal{L}_\nu$-decomposable Bernoulli shift with $\nu \ge 8\rho + \epsilon$ for $\rho = 1$ or $\rho = 2$, and $a > 2$.}

\item \label{assum:factors:two} \textit{$\bm\Sigma_F = \mathsf{E}(\mathbf{f}%
_t \mathbf{f}_t^\top) \in \mathbb{R}^{r \times r}$ is positive definite. }

\item \label{assum:factors:three} \textit{Denoting the long-run variance
matrix of $\mathcal{F}_{t}=\mathsf{Vech}(\mathbf{f}_{t}\mathbf{f}_{t}^{\top
})$ by 
\begin{equation}
\mathbf{D}=\lim_{T\rightarrow \infty }\frac{1}{T}\mathsf{E}\left\{ \left[
\sum_{t=1}^{T}(\mathcal{F}_{t}-\mathsf{E}(\mathcal{F}_{t}))\right] \; \left[ \sum_{t=1}^{T}\left(
\mathcal{F}_{t}-\mathsf{E}(\mathcal{F}_{t})\right) \right]^{\top}\right\} , \label{eq:D}
\end{equation}%
we suppose that $\mathbf{D}$ is invertible.}
\end{enumerate}
\end{assum}

\begin{assum}
\label{loadings} \textit{There exists some $c_0 \in (0, \infty)$ such that: }

\begin{enumerate}[label = (\roman*)]

\item \label{assum:loadings:one} \textit{$\bm\lambda_{i}$ is deterministic
with $\Vert \bm\lambda_{i} \Vert \leq c_{0}$ for all $1 \leq i \leq N$ and $N \in \mathbb N$.}

\item \label{assum:loadings:two} \textit{$\Vert N^{-1}\bm\Lambda^{\top }\bm%
\Lambda -\bm\Sigma_{\Lambda }\Vert \leq c_{0}N^{-1/2}$ for all $N\in\mathbb N$, where $\bm\Sigma
_{\Lambda }\in \mathbb{R}^{r\times r}$ is positive definite. }
\end{enumerate}
\end{assum}

\begin{assum}
\label{idiosyncratic} \textit{There exist some $\epsilon \in (0, 1)$ and $c_{0} \in
(0,\infty )$ such that the following holds for all $N, T \in \mathbb{N}$ and $\rho = 1$ or $\rho = 2$:}

\begin{enumerate}[label = (\roman*)]

\item \label{assum:idiosyncratic:one} \textit{$\mathsf{E}(e_{it}) = 0$ and $\mathsf{E}( \vert e_{i, t} \vert^{8 + \epsilon} ) < \infty$ for all $1 \le i \le N$ and $1 \le t \le T$. }

\item \label{assum:idiosyncratic:two}  \textit{Letting $\gamma_{s, t} =
N^{-1} \sum_{i = 1}^{N} \mathsf{E}(e_{i, s} e_{i, t})$, it holds that $\sum_{t = 1}^{T} \vert \gamma_{s, t} \vert \leq c_0$ for all $1 \le s \le T$. }

\item \label{assum:idiosyncratic:three} \textit{$\mathsf{E}( \vert \sum_{i =
1}^{N} (e_{i,t} e_{i,s} - \gamma_{s, t}) \vert^{4 + \epsilon}) \leq c_{0}
N^{2 + \epsilon / 2}$ for all $1 \le s, t \le T$. }

\item \label{assum:idiosyncratic:four}\textit{$\mathsf{E}(\Vert
\sum_{i=1}^{N}\bm\lambda_{i}e_{i,t}\Vert^{8+\epsilon })\leq
c_{0}N^{4+\epsilon /2}$ for all $1 \le t \le T$. }

\item \label{assum:idiosyncratic:five} \textit{$\sum_{j=1}^{N}\left\vert \mathsf{E}%
\mathit{(e_{i,t}e_{j,t})}\right\vert \leq c_{0}$ for all $1 \le i \le N$ and $1 \le t \le T$.}

\item \label{assum:idiosyncratic:six} \textit{
$\E( \vert \sum_{t = 1}^T (e_{i, t}e_{j, t} - \E(e_{i, t} e_{j, t})) \vert^{4 \rho + \epsilon}) \le c_0 T^{2\rho + \epsilon/2}$ for all $1 \le i, j \le N$.}
\end{enumerate}
\end{assum}

\begin{assum}
\label{depFE} \textit{There exist some $\epsilon \in (0, 1)$ and $c_{0} \in
(0,\infty )$ such that the following holds for all $N, T \in \mathbb{N}$ and $\rho = 1$ or $\rho = 2$:}

\begin{enumerate}[label = (\roman*)]

\item \label{assum:depFE:one} \textit{$\mathsf{E} (\Vert \sum_{t = a + 1}^{b} \mathbf{g}_t \sum_{i = 1}^{N} (e_{i, t} e_{i, s} - \gamma_{s, t} )
\Vert^{2\rho + \epsilon}) \leq c_{0} (N(b - a))^{\rho + \epsilon/2}$ for all $1 \le s \le T$ and $0 \le a < b \le T$. }

\item \label{assum:depFE:two} \textit{$\mathsf{E} (\vert \sum_{i = 1}^{N}
\sum_{t = a + 1}^{b} \bm\lambda_i^\top \mathbf{g}_t e_{i, t} \vert^{4\rho + \epsilon}) \leq c_0 (N (b - a))^{2\rho + \epsilon/2}$ for all $0 \le a < b \le T$. }

\item \label{assum:depFE:three} \textit{$\mathsf{E}( \vert \sum_{t = 1}^T \bm\lambda_i^\top \mathbf{g}_t e_{j, t} \vert^{4\rho + \epsilon}) \le c_0 T^{2\rho + \epsilon/2}$ for all $1 \le i, j \le N$.}
\end{enumerate}
\end{assum}

Assumptions~\ref{factors}--\ref{depFE} are closely related to those adopted in the factor model literature, see also \cite{bai03}.
Assumption~\ref{loadings} is the same as Assumptions~B in \citet{bai03}.
In Assumption~\ref{factors}~\ref{assum:factors:one}, we strengthen the moment condition typically employed in the
literature on $\mathbf{f}_{t}$, switching from the existence of the $4$-th
moment to that of the $8$-th or higher moment. 
Also, we assume a specific form of
dependence for $\{\mbf f_{t},-\infty <t<\infty \}$ which, as mentioned above, accommodates a wide range of time series models.
This is required to derive a Strong Invariance Principle (SIP) for $\mathsf{Vech}(\mbf f_t \mbf f_t^\top)$, see Lemma~\ref{sip-1} in Appendix~\ref{sec:prem:proofs}. 
Assumption~\ref{idiosyncratic} allows for weak temporal and cross-sectional dependence in the idiosyncratic component, with similarities between~\ref{assum:idiosyncratic:two} and Assumption~E1 in \citet{bai03}, \ref{assum:idiosyncratic:three} and C5, and
\ref{assum:idiosyncratic:five} and E2;
part~\ref{assum:idiosyncratic:four} strengthens their F3 and also Assumption~6~(ii) of \cite{baiduan};
part~\ref{assum:idiosyncratic:six} can be derived under more primitive conditions on $e_{i, t}$.
Assumption~\ref{depFE}~\ref{assum:depFE:one}--\ref{assum:depFE:two} extend Assumption~F1 and~F2 of \cite{bai03}, to account for the scanning for multiple change points performed by the MOSUM procedure.

Generally, the strengthened conditions found in Assumptions~\ref{factors}~\ref{assum:factors:one} and~\ref{idiosyncratic} on the moments of $\mathbf{f}_t$ and $e_{i, t}$, are required as we go a step further from the typical factor modelling literature that focuses on establishing the consistency of the estimated factors, to control the partial sums involved in the MOSUM process. 
We note that $\rho = 1$ in Assumptions~\ref{factors}, \ref{idiosyncratic} and~\ref{depFE}, is sufficient 
for deriving the asymptotic null distribution of the MOSUM test statistic (Theorem~\ref{thm:null})
as well as the detection consistency of the MOSUM procedure (Theorem~\ref{thm:consistency}~\ref{thm:consistency:one}), 
while $\rho = 2$ is required for establishing the rate of estimation for the change points 
(Theorem~\ref{thm:consistency}~\ref{thm:consistency:two}).

\begin{assum}
\label{assum:cps}

\begin{enumerate}[label = (\roman*)]

\item \label{assum:cp:one} \textit{There exist $\tau_j, \, 1 \le j \le R$, satisfying $0 < \tau_1 < \ldots < \tau_R < 1$,
such that $k_j = \lfloor \tau_j T \rfloor$. }

\item \label{assum:cp:three} \textit{$\Vert \mathbf{A}_j \Vert \le c_0 \in
(0, \infty)$ for all $0 \le j \le R$.}

\item \label{assum:cp:two} \textit{Denoting by $\bm\Sigma_G = T^{-1} \sum_{t
= 1}^T \mathsf{E}(\mathbf{g}_t \mathbf{g}_t^\top)$, the eigenvalues of $\bm%
\Sigma_G \bm\Sigma_\Lambda$ are positive and distinct. }
\end{enumerate}
\end{assum}

\begin{assum}
\label{assum:nt} \textit{There exists some $\epsilon_{\circ }\in (0,\infty
) $ such that 
\begin{equation*}
\lim_{\min (N,T)\rightarrow \infty }\frac{T^{1/2+\epsilon_{\circ }}}{N}=0.
\end{equation*}%
}
\end{assum}

When $R = 0$, Assumption~\ref{assum:cps} only requires that $\bm\Sigma_F \bm\Sigma_\Lambda$ has distinct eigenvalues, paralleling the commonly found condition such as Assumption~G of \cite{bai03}.
When $R \geq 1$, part~\ref{assum:cp:one} assumes that the change points are linearly spaced. 
The positive definiteness imposed on $\bm\Sigma_{G}$ in part~\ref{assum:cp:two}, together with part~\ref{assum:cp:one} and Assumption~\ref{factors}~\ref{assum:factors:two}, implies
that any local factors are pervasive over segment(s) where they are present.
Finally, Assumption~\ref{assum:nt} is also found in \cite{baiduan}, and arises from that we construct the MOSUM process based on an estimate of the latent factors.

\subsection{Asymptotic null distribution}

We present the limiting distribution of the maximally selected MOSUM
statistic in~\eqref{mos}. We write, for simplicity, $d = r(r + 1)/2$ and
denote by $\beta = \log(N)/\log(T)$; under Assumption~\ref{assum:nt}, we
have $1/2 + \epsilon_\circ < \beta$ with some $\epsilon_\circ > 0$. Define
also 
\begin{align}  \label{eq:zeta}
\zeta = \max\l \{ \frac{2}{\nu}, 1 - \min(1, \beta) \r\},
\end{align}
where $\nu$ is defined in Assumption \ref{factors}. Then, we always have $\zeta \in (0, 1/2)$.

\begin{thm}
\label{thm:null} \textit{Suppose that Assumption~\ref{factors}--\ref{assum:nt} hold 
with $\rho = 1$ for Assumptions~\ref{factors}, \ref{idiosyncratic} and~\ref{depFE}, and let the
bandwidth $\gamma $ satisfy 
\begin{equation}
{\frac{T^{2\zeta }\log (T/\gamma )}{\gamma }\rightarrow 0\text{ \ and \ }%
\frac{\gamma }{T}\rightarrow 0.} \label{b-mosum}
\end{equation}%
Let us define {$\mathbf{V}=\mathbf{L}_{r}(\mathbf{H}_{0}^{\top }\otimes 
\mathbf{H}_{0}^{\top })\mathbf{K}_{r}\mathbf{D}\mathbf{K}_{r}^{\top }(%
\mathbf{H}_{0}\otimes \mathbf{H}_{0})\mathbf{L}_{r}^{\top}$}, with $\mathbf{%
D}$ in~\eqref{eq:D}, and
\begin{equation}
\mathbf{H}_{0}=\quad \text{\upshape p}\!\!\!\!\!\!\!\!\!\!\!\lim_{\min (N,T)\rightarrow \infty }\frac{1}{NT}(\bm%
\Lambda^{\top }\bm\Lambda )(\mathbf{G}^{\top }\widehat{\mathbf{G}})\bm\Phi
_{NT}^{-1} \label{h-tilde}
\end{equation}%
where we denote with $\bm\Phi_{NT}\in \mathbb{R}^{r\times r}$ the diagonal matrix containing
the $r$ largest eigenvalues of $(NT)^{-1}\mathbf{X}\mathbf{X}^{\top}$ on
its diagonal. }
\begin{enumerate}[label = (\alph*)]

\item \label{thm:null:one} \textit{Under $\mathcal{H}_{0}:\,R=0$, for all $x\in \mathbb{R}$, we have 
\begin{multline}
\lim_{\min (N,T)\rightarrow \infty }\mathsf{P}\left( a\left( \frac{T}{\gamma 
}\right) \max_{\gamma \leq k\leq T-\gamma }\left\vert \mathbf{M}_{N,T,\gamma
}^{\top }(k)\mathbf{V}^{-1}\mathbf{M}_{N,T,\gamma }(k)\right\vert^{1/2} - b_d\left( \frac{T}{\gamma }\right) \leq x\right) \\
=\exp \left( -2\exp (-x)\right) ,  \label{eq:null}
\end{multline}%
where $a(x)=\sqrt{2\log (x)}$ and $b_d(x)=2\log (x)+d\log \log (x)/2+\log
(1/2)-\log (\Gamma (d/2))$.}

\item \label{thm:null:two} \textit{The assertion in~\ref{thm:null:one}
continues to hold if $\mathbf{V}$ is replaced by a positive definite matrix $\widehat{\mathbf{V}}$ satisfying
\begin{equation}
\left\Vert \widehat{\mathbf{V}}-{\mathbf{V}}\right\Vert =o_{P}\left( \log
^{-1}(T/\gamma )\right) . \label{cov-stronger}
\end{equation}%
}
\end{enumerate}
\end{thm}

The limiting law in Theorem~\ref{thm:null} is analogous to those derived in
\citet{huvskova2001permutation} and \citet{kirch}, modulo the fact that here, we deal with $d$-variate vectors and therefore the function $b_d(\cdot)$ depends
on $d$. In contrast with a ``standard''
multivariate time series application, however, in our result, the cross-sectional dimension $N$ also plays a role through the definition of $\zeta$ in~\eqref{eq:zeta}, which enters in the condition~\eqref{b-mosum} made on the bandwidth~$\gamma$.
As a by-product, we establish the SIP of the process $\mathsf{Vech}( \widehat{\mathbf{g}}_{t} \widehat{\mathbf{g}}_{t}^{\top})$ after an appropriate centering (see Lemma~\ref{bound} in Appendix~\ref{sec:prem:proofs}), which proves crucial in deriving the asymptotic null distribution in~\eqref{eq:null}. 

Based on this limiting law of the maximally
selected MOSUM process, we reject $\mathcal{H}_{0}:\,R=0$ at the significance level $\alpha \in
(0,1)$, if 
\begin{equation}
\max_{\gamma \leq k\leq T-\gamma }\left\vert \mathbf{M}_{N,T,\gamma }^{\top
}(k)\mbf V^{-1}\mathbf{M}_{N,T,\gamma }(k)\right\vert^{1/2}>%
\widetilde{D}_{T,\gamma }(\alpha ):=\frac{b_d(T/\gamma )-\log \log \left( 
\frac{1}{\sqrt{1-\alpha }}\right) }{a(T/\gamma )}.\label{eq:thresholdDtilde}
\end{equation}

The condition in~\eqref{b-mosum} places both upper and lower bounds on the
bandwidth $\gamma$. Specifically, $\gamma $ is required
to grow sufficiently faster than $T^{2\zeta}$ while satisfying $T^{-1}\gamma \rightarrow 0$,
and the former restriction calls for larger $\gamma$ when $\mbf f_t$ has fewer moments or when $N$ is small. 
We note that $\beta = \log(N)/\log(T)$ is known
and does not need to be estimated.
Conversely, $\nu$ is in general not
known. We may select $\gamma$ satisfying~\eqref{b-mosum} by plugging in an
estimate of $\nu$, say $\widehat{\nu}$; alternatively, one can decide a
value of $\nu$, say $\nu^{\ast}$, and test whether $\left\vert \mathbf{g}%
_{t}\right\vert_{\nu^{\ast }}<\infty $.
In both cases, one difficulty is
that $\mathbf{g}_{t}$ is not observable; however, deriving $\nu $ from the
data $\mathbf{X}_{t}$ yields a lower bound.

Theorem~\ref{thm:null}~\ref{thm:null:two} shows that
when the unknown $\mathbf{V}$ is replaced by its estimator, the
asymptotic null distribution continues to hold provided that~%
\eqref{cov-stronger} is met. 
The following Proposition~\ref{cor:V}
guarantees that this requirement is met by the estimator $\widehat{\mathbf{V}}$ proposed in~\eqref{v-hat}, strengthening the observation made in %
\citet{baiduan} that $\Vert \widehat{\mathbf{V}}-\mathbf{V}\Vert =o_{P}(1)$.

\begin{prop}
\label{cor:V} \textit{Suppose that Assumption~\ref{factors}--\ref{assum:nt} hold with $\rho = 1$ for Assumptions~\ref{factors}, \ref{idiosyncratic} and~\ref{depFE}. Also, let the bandwidth $m$ satisfy
\begin{align}
\frac{\log(T/\gamma)}{m} \rightarrow 0 \text{ \ and \ } \frac{m \log(T/\gamma)}{%
\sqrt{\min(N, T)}} \rightarrow 0.  \label{b-mosum-restrict}
\end{align}
Then, as $\min(N, T) \rightarrow \infty$, the estimator $\widehat{\mathbf{V}}$ in~%
\eqref{v-hat} satisfies the condition in~\eqref{cov-stronger}.}
\end{prop}

\subsection{Consistency in multiple change point estimation}
\label{segment}

To establish the consistency of the MOSUM procedure in multiple change point
detection, we make the following assumption on the size of changes. 

\begin{assum}
\label{kirch}

\begin{enumerate}[label = (\roman*)]

\item \label{assum:kirch:three} \textit{$\min_{0 \le j \le R} \Delta_j \ge 2
\gamma$, where $\Delta_j = k_{j + 1} - k_j$. }

\item \label{assum:kirch:two} \textit{
At each $1 \le j \le R$, let $\bm\delta_j = \mathbf{A}_j \bm\Sigma_F \mathbf{A}_j^\top - \mathbf{A}%
_{j - 1} \bm\Sigma_F \mathbf{A}_{j - 1}^\top$ and $d_j = \Vert \bm\delta_{j}\Vert$. Then for $\omega
_{T}^{(1)}\rightarrow \infty $ arbitrarily slowly, it holds that 
\begin{equation*}
\frac{\omega_{T}^{(1)}\sqrt{\log (T/\gamma )}}{\min_{1 \le j \le R} d_j \sqrt{\gamma}}=o(1).
\end{equation*}%
}
\end{enumerate}
\end{assum}

Assumption~\ref{kirch}~\ref{assum:kirch:three} ensures that there exists at most a single change point over each moving window, and is implied jointly by Assumption~\ref%
{assum:cps}~\ref{assum:cp:one} and the condition~\eqref{b-mosum} on $\gamma$.
Condition~\ref{assum:kirch:two} permits local changes with $d_j \rightarrow 0$, at a sufficiently slow rate, which is the case e.g.\ when the neighbouring loading matrices are rotations of one another with $\Vert \mbf A_j - \mbf A_{j - 1} \Vert \to 0$.

\begin{thm}
\label{thm:consistency} \textit{Suppose that Assumption~\ref{factors}--\ref{assum:nt} hold with $\rho = 1$ for Assumptions~\ref{factors}, \ref{idiosyncratic} and~\ref{depFE}.
Additionally, let Assumption~\ref{kirch} hold and the bandwidth $\gamma$ satisfy~\eqref{b-mosum}, and suppose that some positive definite matrix $\widetilde{\mathbf{V}}$ is used in place of $\mathbf{V}_{k}$ in $\mathcal{T}_{N,T,\gamma }(k)$, see~\eqref{mos}. 
\begin{enumerate}[label = (\alph*)]
\item \label{thm:consistency:one} For any $\alpha, \eta \in (0,1)$,
there exists some sequence $\omega_{T}^{(1)} \rightarrow \infty$ arbitrarily slowly,
such that the MOSUM
procedure with $D_{T, \gamma} = \widetilde{D}_{T,\gamma }(\alpha )\cdot \omega
_{T}^{(1)}$ as the threshold, returns $\{\widehat{k}_{j},\,1\leq j\leq 
\widehat{R}:\,\widehat{k}_{1}<\ldots <\widehat{k}_{\widehat{R}}\}$ which
satisfies
\begin{equation*}
\mathsf{P}\left( \widehat{R}=R; \, \max_{1\leq j\leq R} \vert \widehat{k}_{j}-k_{j} \vert \leq \eta \gamma \right) \rightarrow 1 \text{ \ as \ } \min(N, T) \rightarrow \infty.
\end{equation*}%
\item \label{thm:consistency:two} Further, if Assumptions~\ref{factors}, \ref{idiosyncratic} and~\ref{depFE} hold with $\rho = 2$, there exists some sequence $\omega_{T}^{(2)}\rightarrow \infty$ arbitrarily slowly, such that
\begin{align*}
\mathsf{P}\left( \widehat{R}=R;\,\max_{1\leq j\leq R}d_j^{2} \vert \widehat{k}_{j}-k_{j} \vert \leq \omega_{T}^{(2)}\right) \rightarrow 1 \text{ \ as \ } \min(N, T) \rightarrow \infty.
\end{align*}
\end{enumerate}
}
\end{thm}

Theorem~\ref{thm:consistency} shows that the MOSUM procedure consistently estimates the total number and the locations of the change points.
Here, we adopt a fixed, positive definite matrix $\widetilde{\mathbf{V}}$ in place of $\mathbf{V}_k$ which, without being a consistent estimator of the latter at some~$k$, still leads to consistency in multiple change point detection.
This flexibility in the choice of $\widetilde{\mathbf{V}}$ is particularly favourable since, as noted earlier in Section~\ref{sec:method}, the estimation of (time-varying) long-run covariance matrix for multivariate time series is challenging. 
While the asymptotic null distribution in Theorem~\ref{thm:null} allows for testing the null hypothesis of no change point with the family-wise error controlled, strengthening of the threshold is necessary for consistently detecting the number of change points, see e.g.\ \cite{kirch} and \cite{baiduan} where they set $\alpha = \alpha_T \to 0$ at a suitable rate.
Instead, we introduce an additional multiplicative factor of $\omega_T^{(1)} \to \infty$ in the threshold $D_{T, \gamma}$; see Section~\ref{sec:tuning} where we discuss the choice of the threshold. 
Under a stronger moment assumption, we obtain the rate of estimation which is inversely proportional to the squared size of change as $\vert \widehat{k}_j - k_j \vert = O_P(d_j^{-2})$,
This indicates that dominant changes are located with better accuracy, such as those accompanied by a change in the dimension of the factor space due to the introduction or disappearance of factors(s).

\section{Numerical experiments}
\label{simulations}

\subsection{Tuning parameter selection}

\label{sec:tuning}

Empirical performance of the MOSUM procedure depends on the choice of tuning parameters.

Inspecting the proof of Theorem~\ref{thm:null}, we observe that the requirement on $\nu$ in Assumption~\ref{factors}~\ref{assum:factors:one} 
may be weakened if the $r$ largest eigenvalues of $(NT)^{-1} \mbf X\mbf X^\top$ are bounded away from zero deterministically.
Inspired by this and the condition in~\eqref{b-mosum}, we propose to
select the bandwidth as $\gamma = \lfloor T^{2\zeta} \cdot \log^\varrho(T)
\rfloor$ with $\zeta = \max(2/5, 1 - \log(N)/\log(T))$.
Thus-selected
bandwidth with $\varrho = 1.1$ works reasonably well in our simulation
studies where datasets of dimensions $N \le 500$ and $T \le 1000$ are
considered. For the real data application in Section~\ref{sec:real} with 
$T \ge 4000$, we set~$\varrho = 0.5$. 

We set the threshold $D_{T, \gamma}$ as described in Theorem~\ref{thm:consistency}, namely $D_{T, \gamma} = \widetilde{D}_{T, \gamma}(\alpha)
\cdot \omega_T^{(1)}$ where $\widetilde{D}_{T, \gamma}(\alpha)$ is given by \eqref{eq:thresholdDtilde}
according to the asymptotic null distribution in Theorem~\ref{thm:null}. 
As for $\omega_T^{(1)}$, we have considered $\log^\kappa(T/\gamma)$ with $\kappa \in \{0, 0.1, 0.2, 0.3\}$, and observed that the choice of $\kappa = 0.2$ returned stably good performance in all experiments, see Appendix~\ref{sec:kappa} for full detail. 
Compared to the choice of $\omega_T^{(1)}$, the selection of the fixed significance level $\alpha \in (0, 1)$ has relatively little influence and in all our studies, we set $\alpha = 0.05$.
Finally, for the detection rule in~\eqref{eq:eta}, we set $\eta = 0.6$.

For estimating the number of pseudo factors $r$, we apply the approach proposed by \cite{ABC10} in combination with the three information criteria of \cite{baing02}: Addressing the arbitrariness in the choice of the penalty, it looks for a stable estimate of the factor number as the minimiser of the information criterion over sub-samples of varying dimensions and sample sizes. 
We take the median of the estimates from the three information criteria if they do not agree. 
On simulated datasets, we find that this approach consistently identifies the correct number of factors over $90\%$ of the realisations. 

Finally, in place of $\mathbf{V}_k$ in~\eqref{mos}, we plug in the estimator $\widehat{\mathbf{V}}$ in~\eqref{v-hat} with bandwidth $m = \lfloor T^{1/4} \rfloor$.
Due to the presence of change points, the number of pseudo factors $r$ can
be large in which case inverting the $d \times d$-matrix $\widehat{\mathbf{V}%
}$ with $d = r (r + 1)/2$, may bring numerical instability, see \cite{kirch2015detection} for the alternative approaches to handling similar difficulties in multivariate time series segmentation.
Therefore, we explore two approaches, one performing the standardisation using the full matrix $\widehat{\mathbf{V}}$, and the other using its diagonal entries only, respectively referred to as MOSUM-full and MOSUM-diagonal. 
We remark that MOSUM-diagonal meets the requirement in Theorem~\ref{thm:consistency} provided that all diagonal entries of $\widehat{\mathbf{V}}$ are positive.
Our numerical experiments indicate that MOSUM-diagonal is to be preferred between the two, see Section~\ref{sec:sim:res} for further discussions.


\subsection{Settings}

We consider the following data generating processes considered in \cite{li2023detection} and 
\cite{duan2023quasi}.

\begin{enumerate}[label = (M\arabic*)]
\item \label{m:one} 
Adopted from \cite{li2023detection}, we fix $T=400$, $N=200$ and $r_{0}=5$, and introduce $R=2$ change points at $(k_{1},k_{2})=(133,267)$ as follows: $X_{it}=\chi_{it}+%
\sqrt{0.5}e_{it}$, where 
\begin{align*}
\chi_{it}& =\bm\Lambda_{j}\mathbf{f}_{t}\text{ \ with \ }\mathbf{f}_{t}\sim \mathcal{N}_{r_{0}}(\mathbf{0},\bm\Sigma_{j})\text{ \ for \ }
k_{j-1}+1\leq t\leq k_{j},\text{ \ and} \\
\bm\Sigma_{0}& =[\sigma_{0,ij}]=\mathbf{D}\bm\Sigma_{F}\mathbf{D}\text{ \
with \ }\mathbf{D}=\text{diag}(d_{ii},\,1\leq i\leq r_{0}),\,d_{ii}\sim_{%
\text{\upshape iid}}\mathcal{U}[0.5,1.5], \\
\bm\Sigma_{1}=\bm\Sigma_{2}& =[\sigma_{1,ij}]\text{ \ with \ }\sigma
_{1,ij}=\sigma_{1,ji}=\left\{ 
\begin{array}{ll}
0.9\sqrt{\sigma_{0,11}\sigma_{0,22}} & \text{for \ }(i,j)=(1,2), \\ 
1.3^{2}\sigma_{0,55} & \text{for \ }(i,j)=(5,5), \\ 
0.5^{|i-5|}\sqrt{\sigma_{0,ii}\sigma_{0,55}} & \text{for \ }1\leq i\leq 4,
\\ 
\sigma_{0,ij} & \text{otherwise,}
\end{array}%
\right.
\end{align*}
with $\bm\Sigma_{F}=[0.5^{|i-j|},\,1\leq i,j\leq r_{0}]$.
The loadings are generated as $\bm\Lambda_{0}=\bm\Lambda_{1}=[\lambda
_{0,ij},\,1\leq i\leq p,\,1\leq j\leq r_{0}]$ with $\lambda_{0,ij}\sim_{%
\text{\upshape iid}}\mathcal{U}[-1,1]$, and $\bm\Lambda_{2}=[\lambda
_{2,ij},\,1\leq i\leq p,\,1\leq j\leq r_{0}]$ with $\lambda_{2,ij}\sim_{%
\text{\upshape iid}}\mathcal{U}[-1,1]$ for $j\leq 2$, while $\lambda
_{2,ij}=\lambda_{0,ij}$ for $j\geq 3$. Within each segment, the number of
factors remains constant at $r_{0}=5$ while the overall factor number is $r=r_{0}+2$ due to the increase of factor space after $k_{2}$. The
idiosyncratic component is generated as independent Gaussian random vectors
whose covariance undergoes changes at $t=100,200$ and $300$ (with the proportion of changes set at $0.1$) which are not to be detected by the proposed MOSUM method.

\item \label{m:two} 
We join together three single change point scenarios from \cite{duan2023quasi} to form a multiple change point one: Setting $R = 3$ and $r_0 = 3$, we generate 
\begin{align*}
\mathbf{f}_t &= \rho_f \mathbf{f}_{t - 1} + \bm\varepsilon_{f, t}, \, \bm\varepsilon_{f, t} \sim_{\text{\upshape iid}} \mathcal{N}_{r_0}(\mathbf{0}, 
\mathbf{I}_{r_0}), \\
\mathbf{e}_t &= \rho_e \mathbf{e}_{t - 1} + \bm\varepsilon_{e, t}, \, \bm\varepsilon_{e, t} \sim \mathcal{N}_p(\mathbf{0}, \bm\Sigma_e) \text{ \ with \ } \bm\Sigma_e = [(0.3)^{\vert i - j \vert}, \, 1 \le i, j \le p],
\end{align*}
and $\bm\Lambda_0 = [\lambda_{0, ij}, \, 1 \le i \le p, \, 1 \le j \le r_0]$ with $\lambda_{0, ij} \sim_{\text{\upshape iid}} \mathcal{N}(0, 1/r_0)$. 
The change points are introduced at $k_j = Tj/4, \, 1 \le j \le 3$, at each of which the loading matrix undergoes a shift to $\bm\Lambda_j = \bm\Lambda_0 \mathbf{C}_j$, where 
\begin{align*}
\mathbf{C}_1 = 
\begin{bmatrix}
0.5 & 0 & 0 \\ 
c_{1, 21} & 1 & 0 \\ 
c_{1, 31} & c_{1, 32} & 1.5%
\end{bmatrix}%
, \quad \mathbf{C}_2 = 
\begin{bmatrix}
1 & 0 & 0 \\ 
0 & 1 & 0 \\ 
0 & 0 & 0%
\end{bmatrix}
\text{ \ and \ } \mathbf{C}_3 = [c_{3, ij}, \, 1 \le i, j \le r_0],
\end{align*}
with $c_{1, ij} \sim_{\text{\upshape iid}} \mathcal{N}(0, 1)$ and $c_{3, ij}
\sim_{\text{\upshape iid}} \mathcal{N}(0, 1/r_0)$. The factor number varies as $(r_0, \ldots, r_3) = (3, 3, 2, 3)$, while the number of pseudo factors
increases from $3$ to $6$ due to the change at~$k_3$. We vary $T \in \{400, 600, 800, 1000\}$ and $N \in \{100, 200, 500\}$ as well as $(\rho_f, \rho_e) = \{(0, 0), (0.7, 0.3)\}$.

\item \label{m:zero} 
Additionally, we consider the ``null'' model with $R = 0$ by generating the
data from the model corresponding to the first segment of~\ref{m:two} for each
setting.

\end{enumerate}

\subsection{Results}
\label{sec:sim:res}

\begin{figure}[h!t!b!]
\centering
\includegraphics[width = 1\linewidth]{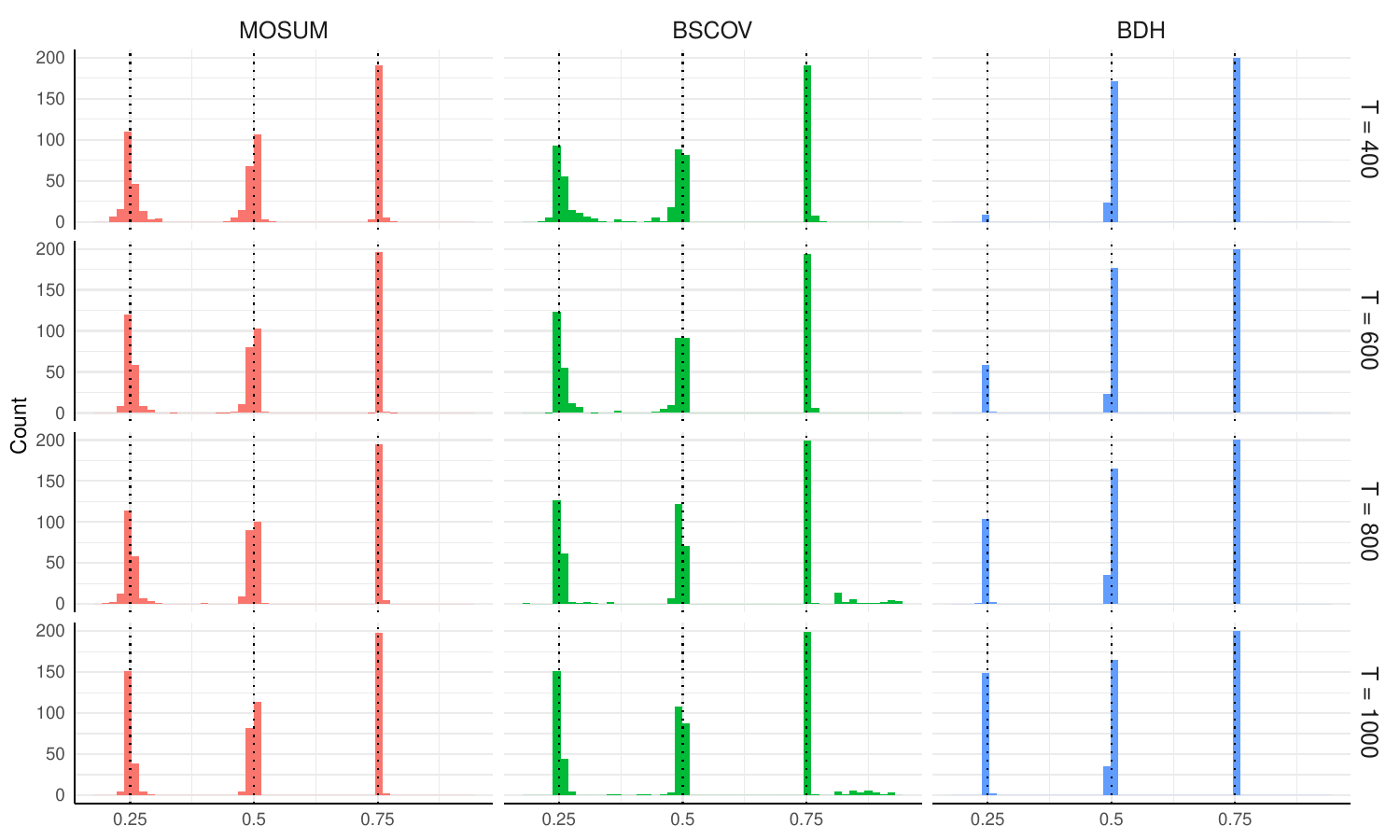}
\caption{\ref{m:two} Histogram of the change point estimators returned by MOSUM-diagonal, BSCOV and BDH when $N = 100$, $(\rho_f, \rho_e) = (0, 0)$ and varying $T \in \{400, 600, 800, 1000\}$ (top to bottom). The scaled locations of the true change points, $k_j/T$, at $(1/4, 1/2, 3/4)$ are marked by vertical dotted lines.}
\label{fig:m3:p100depFALSE}
\end{figure}

\begin{figure}[h!t!b!]
\centering
\includegraphics[width = 1\linewidth]{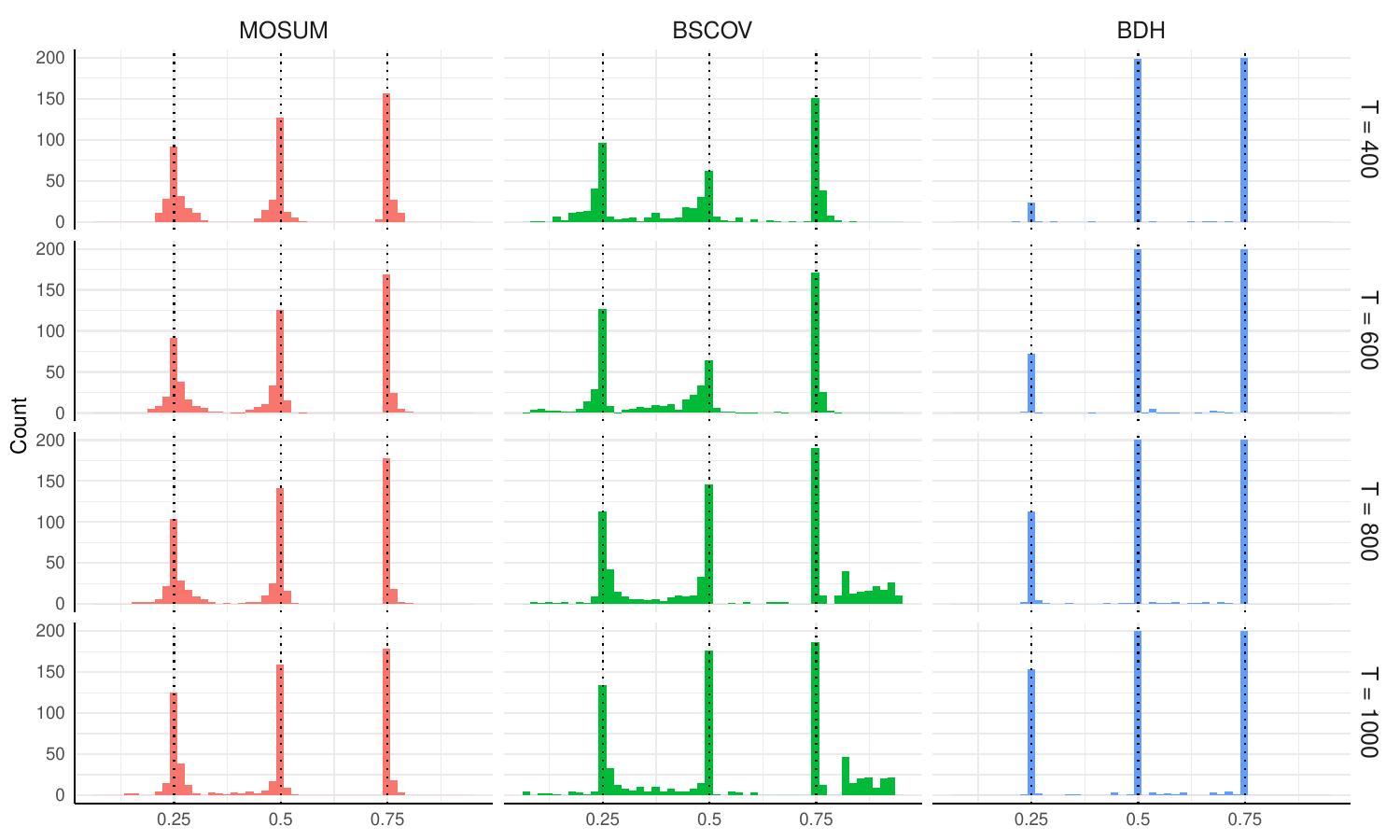}
\caption{\ref{m:two} Histogram of the change point estimators returned by MOSUM-diagonal, BSCOV and BDH with $p = 100$, $(\rho_f, \rho_e) = (0.7, 0.3)$ and varying $T \in \{400, 600, 800, 1000\}$ (top to bottom). The scaled locations of the true change points, $k_j/T$, at $(1/4, 1/2, 3/4)$ are marked by vertical dotted lines.}
\label{fig:m3:p100depTRUE}
\end{figure}

For each setting, we generate $200$ realisations and report the distribution
of $\widehat{R} - R$, and the accuracy of
change point estimators measured by 
\begin{equation*}
\frac{1}{200}\sum_{i=1}^{200}\mathbb{I}\left[ \min_{1\leq \ell \leq \widehat{K}^{(i)}}|\widehat{k}_{\ell }^{(i)}-k_{j}|\leq \log (T)\right]
\end{equation*}%
for each $1 \le j \le R$, as proposed by \cite{li2023detection}, where $\widehat{k}_{\ell
}^{(i)},\,1\leq \ell \leq \widehat{R}^{(i)}$, refer to the change point
estimators from the $i$-th realisation. We apply the MOSUM procedure with the tuning parameters chosen as described in Section~\ref{sec:tuning}, and consider the two choices of the standardisation matrix (MOSUM-full and MOSUM-diagonal). Additionally, we include the two competitors: 
\begin{enumerate}[label = (\roman*)]
\item Proposed by \cite{li2023detection}, BSCOV scans for changes in the covariance of $\{\mathbf{g}_{t}\}$ under~\eqref{eq:model:two} via an extension of the binary segmentation, produces a path of solutions and
selects the final change point model by minimising an information criterion. 
\item Proposed by \cite{baiduan}, BDH recursively applies the likelihood ratio test via binary segmentation to detect the multiple change points under~\eqref{eq:model:two}.
\end{enumerate}
Both methods are applied with default tuning parameters and in-built factor number estimators that are based on the information criterion proposed by \cite{baing02}; for BDH, we set the proportion of the data trimmed off at each recursion to be $0.1$.
Tables~\ref{tab:m:one}--\ref{tab:m:three} report the summary of the results over $200$ realisations, and Figures~\ref{fig:m3:p100depFALSE}--\ref{fig:m3:p100depTRUE} plot the histograms of change point estimators returned by the proposed MOSUM procedure, BSCOV and BDH under~\ref{m:two} when $N = 100$; see also Appendix~\ref{app:sim} for the additional results.

Overall we observe that the MOSUM procedure (`MOSUM') demonstrates competitive performance across all scenarios, both in detection and estimation.
BSCOV tends to return spurious estimators under~\ref{m:two} and~\ref{m:zero} when serial dependence is present with $(\rho_f, \rho_e) = (0.7, 0.3)$.
On the other hand, BDH suffers from lack of detection power against those changes that transform the loading matrix while do not alter the number of local factors, such as $k_1$ under~\ref{m:one} and~\ref{m:two}, particularly for smaller $T$.
For such a change point, MOSUM is able to detect its presence although with less accuracy.
Generally, $k_j$'s which are associated with changes in
the number of pseudo factors (such as $k_2$ under~\ref{m:one} and $k_3$ under~\ref{m:two})
are estimated with higher accuracy, which agrees with the observations made in \cite{duan2023quasi} and also below Theorem~\ref{thm:consistency}.

Between MOSUM-full and MOSUM-diagonal, the latter demonstrates better accuracy in estimating both the total number and locations of change points when $R \ge 1$. This is explained by that $d = r (r + 1)/2$ is as large as $d = 28$ under~\ref{m:one} and $d = 21$ under~\ref{m:two}. Retaining the diagonal elements of $\widehat{\mathbf{V}}$ only, effectively performs standardisation without suffering from the numerical instability inherent in inverting a large matrix.
When $R = 0$, MOSUM-full performs marginally better as here, the factor number is kept at $r = 3$ under~\ref{m:zero}, leading to $d = 6$.
Based on these observations, in practice, we recommend the use of MOSUM-diagonal when the number of (pseudo) factors is moderately large. 

\begin{table}[h!t!]
\caption{\ref{m:one} with $R = 2$: Summary of change point estimators returned by
MOSUM, BSCOV and BDH. The results for BSCOV have been taken from \protect\cite{li2023detection}.}
\label{tab:m:one}
\centering
{\small 
\begin{tabular}{rrccccccc}
\toprule 
&	&	\multicolumn{5}{c}{$\wh{R} - R$} & \multicolumn{2}{c}{Accuracy}		\\	
Method &	LRV &	$-2 \le$ &	$-1$ &	$0$ &	$1$ &	$\ge 2$ &	$j = 1$ &	$j = 2$	\\	\cmidrule(lr){1-2} \cmidrule(lr){3-7} \cmidrule(lr){8-9}
MOSUM &	Diagonal &	0 &	0 &	0.985 &	0.015 &	0 &	0.7 &	0.925	\\	
&	Full &	0 &	0.005 &	0.94 &	0.055 &	0 &	0.595 &	0.8	\\	
BSCOV &	— &	0 &	0.03 &	0.97 &	0 &	0 &	0.64 &	0.95	\\	
BDH &	— &	0 &	0.965 &	0.015 &	0.005 &	0.015 &	0 &	1	\\	\bottomrule
\end{tabular}
}
\end{table}

\begin{table}[h!t!p!]
\caption{\ref{m:two} with $(\protect\rho_f, \protect\rho_e) = (0, 0)$ and $R = 3$:
Summary of change point estimators returned by MOSUM and BSCOV over $200$ realisations.}
\label{tab:m:two}
\centering
{\footnotesize 
\begin{tabular}{rrrrcccccccc}
\toprule 
&	&	&	&	\multicolumn{5}{c}{$\wh{R} - R$} &			\multicolumn{3}{c}{Accuracy} \\		
$n$ &	$p$ &	Method &	LRV &	$-2 \le$ &	$-1$ &	$0$ &	$1$ &	$\ge 2$ &	$j = 1$ &	$j = 2$ &	$j = 3$	\\	\cmidrule(lr){1-4} \cmidrule(lr){5-9} \cmidrule(lr){10-12}	
$400$ &	$100$ &	MOSUM &	Diagonal &	0 &	0.01 &	0.99 &	0 &	0 &	0.82 &	0.88 &	0.985	\\		
&	&	&	Full &	0 &	0.025 &	0.975 &	0 &	0 &	0.75 &	0.895 &	0.93	\\		
&	&	BSCOV &	— &	0.02 &	0 &	0.98 &	0 &	0 &	0.725 &	0.85 &	0.98	\\		
&	&	BDH &	— &	0.02 &	0.935 &	0.045 &	0 &	0 &	0.045 &	0.98 &	1	\\		
&	$200$ &	MOSUM &	Diagonal &	0 &	0 &	1 &	0 &	0 &	0.795 &	0.875 &	0.97	\\		
&	&	&	Full &	0 &	0.035 &	0.965 &	0 &	0 &	0.73 &	0.86 &	0.93	\\		
&	&	BSCOV &	— &	0.01 &	0.005 &	0.985 &	0 &	0 &	0.79 &	0.855 &	0.99	\\		
&	&	BDH &	— &	0 &	0.93 &	0.07 &	0 &	0 &	0.07 &	1 &	1	\\		
&	$500$ &	MOSUM &	Diagonal &	0 &	0.02 &	0.98 &	0 &	0 &	0.8 &	0.88 &	0.975	\\		
&	&	&	Full &	0 &	0.03 &	0.97 &	0 &	0 &	0.74 &	0.895 &	0.95	\\		
&	&	BSCOV &	— &	0.01 &	0 &	0.98 &	0.01 &	0 &	0.8 &	0.89 &	0.98	\\		
&	&	BDH &	— &	0 &	0.885 &	0.115 &	0 &	0 &	0.115 &	1 &	1	\\	\cmidrule(lr){1-4} \cmidrule(lr){5-9} \cmidrule(lr){10-12}	
$600$ &	$100$ &	MOSUM &	Diagonal &	0 &	0 &	1 &	0 &	0 &	0.83 &	0.88 &	0.99	\\		
&	&	&	Full &	0 &	0 &	0.995 &	0.005 &	0 &	0.75 &	0.87 &	0.945	\\		
&	&	BSCOV &	— &	0 &	0 &	0.99 &	0.01 &	0 &	0.81 &	0.895 &	0.995	\\		
&	&	BDH &	— &	0 &	0.7 &	0.3 &	0 &	0 &	0.3 &	1 &	1	\\		
&	$200$ &	MOSUM &	Diagonal &	0 &	0 &	1 &	0 &	0 &	0.81 &	0.89 &	0.98	\\		
&	&	&	Full &	0 &	0.005 &	0.99 &	0.005 &	0 &	0.715 &	0.905 &	0.95	\\		
&	&	BSCOV &	— &	0 &	0 &	0.995 &	0.005 &	0 &	0.875 &	0.945 &	0.995	\\		
&	&	BDH &	— &	0 &	0.65 &	0.35 &	0 &	0 &	0.35 &	1 &	1	\\		
&	$500$ &	MOSUM &	Diagonal &	0 &	0 &	1 &	0 &	0 &	0.83 &	0.92 &	0.99	\\		
&	&	&	Full &	0 &	0 &	1 &	0 &	0 &	0.755 &	0.885 &	0.945	\\		
&	&	BSCOV &	— &	0.005 &	0 &	0.995 &	0 &	0 &	0.865 &	0.95 &	0.995	\\		
&	&	BDH &	— &	0 &	0.615 &	0.385 &	0 &	0 &	0.385 &	1 &	1	\\	\cmidrule(lr){1-4} \cmidrule(lr){5-9} \cmidrule(lr){10-12}	
$800$ &	$100$ &	MOSUM &	Diagonal &	0 &	0 &	0.995 &	0.005 &	0 &	0.785 &	0.905 &	0.975	\\		
&	&	&	Full &	0 &	0 &	0.985 &	0.01 &	0.005 &	0.715 &	0.88 &	0.925	\\		
&	&	BSCOV &	— &	0 &	0.015 &	0.795 &	0.185 &	0.005 &	0.885 &	0.91 &	0.995	\\		
&	&	BDH &	— &	0 &	0.46 &	0.54 &	0 &	0 &	0.535 &	1 &	1	\\		
&	$200$ &	MOSUM &	Diagonal &	0 &	0 &	1 &	0 &	0 &	0.805 &	0.92 &	0.965	\\		
&	&	&	Full &	0 &	0 &	0.985 &	0.015 &	0 &	0.705 &	0.875 &	0.9	\\		
&	&	BSCOV &	— &	0 &	0 &	0.78 &	0.215 &	0.005 &	0.89 &	0.965 &	0.995	\\		
&	&	BDH &	— &	0 &	0.345 &	0.655 &	0 &	0 &	0.655 &	1 &	1	\\		
&	$500$ &	MOSUM &	Diagonal &	0 &	0 &	1 &	0 &	0 &	0.775 &	0.89 &	0.975	\\		
&	&	&	Full &	0 &	0 &	0.98 &	0.02 &	0 &	0.765 &	0.92 &	0.92	\\		
&	&	BSCOV &	— &	0 &	0 &	0.8 &	0.2 &	0 &	0.86 &	0.95 &	0.995	\\		
&	&	BDH &	— &	0 &	0.375 &	0.625 &	0 &	0 &	0.62 &	1 &	1	\\	\cmidrule(lr){1-4} \cmidrule(lr){5-9} \cmidrule(lr){10-12}	
$1000$ &	$100$ &	MOSUM &	Diagonal &	0 &	0 &	1 &	0 &	0 &	0.865 &	0.9 &	0.97	\\		
&	&	&	Full &	0 &	0 &	0.94 &	0.06 &	0 &	0.715 &	0.88 &	0.925	\\		
&	&	BSCOV &	— &	0 &	0.005 &	0.835 &	0.16 &	0 &	0.87 &	0.925 &	0.99	\\		
&	&	BDH &	— &	0 &	0.245 &	0.755 &	0 &	0 &	0.75 &	1 &	1	\\		
&	$200$ &	MOSUM &	Diagonal &	0 &	0 &	0.985 &	0.015 &	0 &	0.83 &	0.935 &	0.98	\\		
&	&	&	Full &	0 &	0 &	0.925 &	0.075 &	0 &	0.735 &	0.88 &	0.92	\\		
&	&	BSCOV &	— &	0 &	0 &	0.86 &	0.13 &	0.01 &	0.91 &	0.955 &	0.99	\\		
&	&	BDH &	— &	0 &	0.185 &	0.815 &	0 &	0 &	0.81 &	1 &	1	\\		
&	$500$ &	MOSUM &	Diagonal &	0 &	0 &	0.995 &	0.005 &	0 &	0.805 &	0.895 &	0.98	\\		
&	&	&	Full &	0 &	0 &	0.925 &	0.075 &	0 &	0.73 &	0.905 &	0.935	\\		
&	&	BSCOV &	— &	0 &	0.005 &	0.83 &	0.16 &	0.005 &	0.875 &	0.965 &	0.99	\\		
&	&	BDH &	— &	0 &	0.205 &	0.795 &	0 &	0 &	0.795 &	1 &	1	\\	\bottomrule	
\end{tabular}
}
\end{table}

\begin{table}[h!t!p!]
\caption{\ref{m:two} with $(\protect\rho_f, \protect\rho_e) = (0.7, 0.3)$ and $R = 3$: Summary of change point estimators returned by MOSUM and BSCOV over $200$ realisations.}
\label{tab:m:two:dep}
\centering
{\footnotesize 
\begin{tabular}{rrrrcccccccc}
\toprule 
&	&	&	&	\multicolumn{5}{c}{$\wh{R} - R$} &					\multicolumn{3}{c}{Accuracy}			\\		
$n$ &	$p$ &	Method &	LRV &	$-2 \le$ &	$-1$ &	$0$ &	$1$ &	$\ge 2$ &	$j = 1$ &	$j = 2$ &	$j = 3$	\\	\cmidrule(lr){1-4} \cmidrule(lr){5-9} \cmidrule(lr){10-12}	
$400$ &	$100$ &	MOSUM &	Diagonal &	0.005 &	0.08 &	0.915 &	0 &	0 &	0.595 &	0.765 &	0.905	\\		
&	&	&	Full &	0 &	0.06 &	0.935 &	0.005 &	0 &	0.49 &	0.74 &	0.85	\\		
&	&	BSCOV &	— &	0.005 &	0.12 &	0.84 &	0.035 &	0 &	0.6 &	0.375 &	0.92	\\		
&	&	BDH &	— &	0.005 &	0.845 &	0.14 &	0.01 &	0 &	0.12 &	0.995 &	1	\\		
&	$200$ &	MOSUM &	Diagonal &	0.005 &	0.11 &	0.885 &	0 &	0 &	0.605 &	0.73 &	0.93	\\		
&	&	&	Full &	0 &	0.125 &	0.875 &	0 &	0 &	0.44 &	0.67 &	0.88	\\		
&	&	BSCOV &	— &	0 &	0.145 &	0.825 &	0.03 &	0 &	0.625 &	0.32 &	0.9	\\		
&	&	BDH &	— &	0 &	0.875 &	0.12 &	0.005 &	0 &	0.09 &	1 &	1	\\		
&	$500$ &	MOSUM &	Diagonal &	0 &	0.115 &	0.885 &	0 &	0 &	0.535 &	0.71 &	0.95	\\		
&	&	&	Full &	0.005 &	0.115 &	0.88 &	0 &	0 &	0.485 &	0.68 &	0.855	\\		
&	&	BSCOV &	— &	0 &	0.11 &	0.85 &	0.04 &	0 &	0.605 &	0.32 &	0.915	\\		
&	&	BDH &	— &	0 &	0.815 &	0.185 &	0 &	0 &	0.145 &	1 &	1	\\	\cmidrule(lr){1-4} \cmidrule(lr){5-9} \cmidrule(lr){10-12}	
$600$ &	$100$ &	MOSUM &	Diagonal &	0 &	0.025 &	0.97 &	0.005 &	0 &	0.565 &	0.69 &	0.905	\\		
&	&	&	Full &	0 &	0.045 &	0.925 &	0.03 &	0 &	0.48 &	0.705 &	0.85	\\		
&	&	BSCOV &	— &	0.005 &	0.095 &	0.75 &	0.14 &	0.01 &	0.695 &	0.355 &	0.95	\\		
&	&	BDH &	— &	0 &	0.575 &	0.39 &	0.035 &	0 &	0.375 &	1 &	1	\\		
&	$200$ &	MOSUM &	Diagonal &	0 &	0.03 &	0.96 &	0.01 &	0 &	0.535 &	0.75 &	0.935	\\		
&	&	&	Full &	0 &	0.025 &	0.945 &	0.03 &	0 &	0.535 &	0.68 &	0.875	\\		
&	&	BSCOV &	— &	0 &	0.1 &	0.785 &	0.11 &	0.005 &	0.665 &	0.38 &	0.945	\\		
&	&	BDH &	— &	0 &	0.54 &	0.405 &	0.045 &	0.01 &	0.38 &	1 &	1	\\		
&	$500$ &	MOSUM &	Diagonal &	0 &	0.005 &	0.995 &	0 &	0 &	0.57 &	0.77 &	0.935	\\		
&	&	&	Full &	0 &	0.025 &	0.945 &	0.03 &	0 &	0.495 &	0.685 &	0.865	\\		
&	&	BSCOV &	— &	0.02 &	0 &	0.825 &	0.145 &	0.01 &	0.6 &	0.735 &	0.96	\\		
&	&	BDH &	— &	0 &	0.53 &	0.4 &	0.065 &	0.005 &	0.415 &	1 &	1	\\	\cmidrule(lr){1-4} \cmidrule(lr){5-9} \cmidrule(lr){10-12}	
$800$ &	$100$ &	MOSUM &	Diagonal &	0 &	0.01 &	0.96 &	0.03 &	0 &	0.54 &	0.69 &	0.92	\\		
&	&	&	Full &	0 &	0.005 &	0.885 &	0.11 &	0 &	0.49 &	0.685 &	0.85	\\		
&	&	BSCOV &	— &	0.01 &	0 &	0.115 &	0.625 &	0.25 &	0.61 &	0.705 &	0.97	\\		
&	&	BDH &	— &	0 &	0.345 &	0.605 &	0.05 &	0 &	0.58 &	1 &	1	\\		
&	$200$ &	MOSUM &	Diagonal &	0 &	0 &	0.955 &	0.045 &	0 &	0.575 &	0.785 &	0.925	\\		
&	&	&	Full &	0 &	0.03 &	0.85 &	0.12 &	0 &	0.455 &	0.69 &	0.83	\\		
&	&	BSCOV &	— &	0.01 &	0 &	0.08 &	0.595 &	0.315 &	0.625 &	0.805 &	0.975	\\		
&	&	BDH &	— &	0 &	0.275 &	0.655 &	0.065 &	0.005 &	0.685 &	1 &	1	\\		
&	$500$ &	MOSUM &	Diagonal &	0 &	0.01 &	0.95 &	0.04 &	0 &	0.6 &	0.725 &	0.915	\\		
&	&	&	Full &	0 &	0.01 &	0.835 &	0.155 &	0 &	0.495 &	0.63 &	0.8	\\		
&	&	BSCOV &	— &	0.015 &	0 &	0.085 &	0.58 &	0.32 &	0.655 &	0.775 &	0.97	\\		
&	&	BDH &	— &	0 &	0.295 &	0.63 &	0.065 &	0.01 &	0.64 &	1 &	1	\\	\cmidrule(lr){1-4} \cmidrule(lr){5-9} \cmidrule(lr){10-12}	
$1000$ &	$100$ &	MOSUM &	Diagonal &	0 &	0 &	0.945 &	0.055 &	0 &	0.585 &	0.765 &	0.91	\\		
&	&	&	Full &	0 &	0.005 &	0.805 &	0.18 &	0.01 &	0.49 &	0.725 &	0.81	\\		
&	&	BSCOV &	— &	0.005 &	0.005 &	0.195 &	0.5 &	0.295 &	0.64 &	0.825 &	0.95	\\		
&	&	BDH &	— &	0 &	0.19 &	0.715 &	0.09 &	0.005 &	0.76 &	1 &	1	\\		
&	$200$ &	MOSUM &	Diagonal &	0 &	0 &	0.945 &	0.05 &	0.005 &	0.545 &	0.77 &	0.92	\\		
&	&	&	Full &	0 &	0.01 &	0.81 &	0.17 &	0.01 &	0.47 &	0.74 &	0.84	\\		
&	&	BSCOV &	— &	0 &	0.02 &	0.21 &	0.405 &	0.365 &	0.64 &	0.825 &	0.945	\\		
&	&	BDH &	— &	0 &	0.135 &	0.765 &	0.08 &	0.02 &	0.825 &	1 &	1	\\		
&	$500$ &	MOSUM &	Diagonal &	0 &	0 &	0.93 &	0.07 &	0 &	0.605 &	0.73 &	0.915	\\		
&	&	&	Full &	0 &	0 &	0.78 &	0.21 &	0.01 &	0.515 &	0.68 &	0.84	\\		
&	&	BSCOV &	— &	0 &	0.01 &	0.205 &	0.45 &	0.335 &	0.67 &	0.82 &	0.96	\\		
&	&	BDH &	— &	0 &	0.16 &	0.755 &	0.07 &	0.015 &	0.795 &	1 &	1	\\	\bottomrule	
\end{tabular}
}
\end{table}

\begin{table}[h!t!p!]
\caption{\ref{m:zero} with $R = 0$: Distribution of $\widehat R - R$ returned by MOSUM and BSCOV over $200$ realisations for $(\rho_f, \rho_e) \in \{ (0, 0), (0.7, 0.3)\}$.}
\label{tab:m:three}
\centering
{\footnotesize
\begin{tabular}{rrrrcccccc}
\toprule 
&	&	&	&	\multicolumn{3}{c}{$(\rho_f, \rho_e) = (0, 0)$}  &			\multicolumn{3}{c}{$(\rho_f, \rho_e) = (0.7, 0.3)$} 			\\ \cmidrule(lr){5-7} \cmidrule(lr){8-10}
$n$ &	$p$ &	Method &	LRV &	$0$ &	$1$ &	$\ge 2$ &	$0$ &	$1$ &	$\ge 2$ 	\\	\cmidrule(lr){1-4} \cmidrule(lr){5-7} \cmidrule(lr){8-10}	
$400$ &	$100$ &	MOSUM &	Diagonal &	0.985 &	0.015 &	0 &	0.895 &	0.08 &	0.025	\\		
&	&	&	Full &	0.99 &	0.01 &	0 &	0.965 &	0.025 &	0.01	\\		
&	&	BSCOV &	— &	1 &	0 &	0 &	0.75 &	0.215 &	0.035	\\		
&	&	BDH &	— &	1 &	0 &	0 &	1 &	0 &	0	\\		
&	$200$ &	MOSUM &	Diagonal &	0.995 &	0.005 &	0 &	0.88 &	0.11 &	0.01	\\		
&	&	&	Full &	0.995 &	0.005 &	0 &	0.955 &	0.04 &	0.005	\\		
&	&	BSCOV &	— &	0.995 &	0.005 &	0 &	0.705 &	0.25 &	0.045	\\		
&	&	BDH &	— &	1 &	0 &	0 &	1 &	0 &	0	\\		
&	$500$ &	MOSUM &	Diagonal &	0.99 &	0.01 &	0 &	0.88 &	0.11 &	0.01	\\		
&	&	&	Full &	0.99 &	0.01 &	0 &	0.95 &	0.045 &	0.005	\\		
&	&	BSCOV &	— &	0.995 &	0.005 &	0 &	0.75 &	0.22 &	0.03	\\		
&	&	BDH &	— &	1 &	0 &	0 &	1 &	0 &	0	\\	\cmidrule(lr){1-4} \cmidrule(lr){5-7} \cmidrule(lr){8-10}	
$600$ &	$100$ &	MOSUM &	Diagonal &	0.99 &	0.01 &	0 &	0.855 &	0.125 &	0.02	\\		
&	&	&	Full &	0.99 &	0.01 &	0 &	0.905 &	0.075 &	0.02	\\		
&	&	BSCOV &	— &	1 &	0 &	0 &	0.82 &	0.155 &	0.025	\\		
&	&	BDH &	— &	1 &	0 &	0 &	1 &	0 &	0	\\		
&	$200$ &	MOSUM &	Diagonal &	0.99 &	0.01 &	0 &	0.845 &	0.14 &	0.015	\\		
&	&	&	Full &	0.995 &	0.005 &	0 &	0.88 &	0.105 &	0.015	\\		
&	&	BSCOV &	— &	1 &	0 &	0 &	0.77 &	0.18 &	0.05	\\		
&	&	BDH &	— &	1 &	0 &	0 &	1 &	0 &	0	\\		
&	$500$ &	MOSUM &	Diagonal &	0.985 &	0.015 &	0 &	0.855 &	0.11 &	0.035	\\		
&	&	&	Full &	0.995 &	0.005 &	0 &	0.875 &	0.1 &	0.025	\\		
&	&	BSCOV &	— &	1 &	0 &	0 &	0.85 &	0.135 &	0.015	\\		
&	&	BDH &	— &	1 &	0 &	0 &	1 &	0 &	0	\\	\cmidrule(lr){1-4} \cmidrule(lr){5-7} \cmidrule(lr){8-10}	
$800$ &	$100$ &	MOSUM &	Diagonal &	1 &	0 &	0 &	0.905 &	0.085 &	0.01	\\		
&	&	&	Full &	1 &	0 &	0 &	0.95 &	0.045 &	0.005	\\		
&	&	BSCOV &	— &	1 &	0 &	0 &	0.835 &	0.15 &	0.015	\\		
&	&	BDH &	— &	1 &	0 &	0 &	1 &	0 &	0	\\		
&	$200$ &	MOSUM &	Diagonal &	1 &	0 &	0 &	0.94 &	0.05 &	0.01	\\		
&	&	&	Full &	0.995 &	0.005 &	0 &	0.97 &	0.02 &	0.01	\\		
&	&	BSCOV &	— &	1 &	0 &	0 &	0.875 &	0.115 &	0.01	\\		
&	&	BDH &	— &	1 &	0 &	0 &	1 &	0 &	0	\\		
&	$500$ &	MOSUM &	Diagonal &	1 &	0 &	0 &	0.92 &	0.075 &	0.005	\\		
&	&	&	Full &	1 &	0 &	0 &	0.94 &	0.055 &	0.005	\\		
&	&	BSCOV &	— &	1 &	0 &	0 &	0.86 &	0.12 &	0.02	\\		
&	&	BDH &	— &	1 &	0 &	0 &	1 &	0 &	0	\\	\cmidrule(lr){1-4} \cmidrule(lr){5-7} \cmidrule(lr){8-10}	
$1000$ &	$100$ &	MOSUM &	Diagonal &	1 &	0 &	0 &	0.895 &	0.1 &	0.005	\\		
&	&	&	Full &	1 &	0 &	0 &	0.945 &	0.05 &	0.005	\\		
&	&	BSCOV &	— &	1 &	0 &	0 &	0.915 &	0.085 &	0	\\		
&	&	BDH &	— &	1 &	0 &	0 &	1 &	0 &	0	\\		
&	$200$ &	MOSUM &	Diagonal &	1 &	0 &	0 &	0.9 &	0.09 &	0.01	\\		
&	&	&	Full &	1 &	0 &	0 &	0.955 &	0.045 &	0	\\		
&	&	BSCOV &	— &	1 &	0 &	0 &	0.895 &	0.09 &	0.015	\\		
&	&	BDH &	— &	1 &	0 &	0 &	1 &	0 &	0	\\		
&	$500$ &	MOSUM &	Diagonal &	1 &	0 &	0 &	0.91 &	0.085 &	0.005	\\		
&	&	&	Full &	1 &	0 &	0 &	0.95 &	0.05 &	0	\\		
&	&	BSCOV &	— &	1 &	0 &	0 &	0.91 &	0.085 &	0.005	\\		
&	&	BDH &	— &	1 &	0 &	0 &	1 &	0 &	0	\\	\bottomrule	
\end{tabular}
}
\end{table}

\section{Real data application}
\label{sec:real}

We consider daily stock prices from $72$ US blue chip companies across industry sectors
between January 3, 2005 and February 16, 2022,
retrieved from the Wharton Research Data Services.
We measure the volatility as the daily high-low range as $\sigma_{it}^2 = 0.361 (p_{it}^{\text{high}} - p_{it}^{\text{low}})^2$
where $p_{it}^{\text{high}}$ (resp.\ $p_{it}^{\text{low}}$) denotes the maximum (resp.\ minimum) price of stock $i$ on day $t$, and set $X_{it} = \log(\sigma_{it}^2)$, see, e.g.\ \citet{diebold2014network}.

The sub-sampling-based factor number estimator discussed in Section~\ref{sec:tuning} returns $r = 7$ as the factor number.
However, the panel data is unbalanced with the dimension $N = 72$ being considerably smaller than the sample size $T = 4312$, a situation that does not favour the sub-sampling approach as pointed out by \cite{onatski2024}.
The information criteria of \citet{baing02} return $r = 5$, while the approach based on inspecting the ratio of successive eigenvalues \citep{ahnhorenstein13} returns $r = 1$, which implies that any change point we detect would solely be attributed to heteroscedasticity of the single factor. 
While the former is known to detect weakly pervasive factors \citep{baing23}, the latter tends to recover only strongly pervasive ones. 

In the presence of some uncertainty in the number of factors, a situation commonly faced in real data analysis, we choose to apply the proposed MOSUM procedure with varying $r \in \{1, \ldots, 7\}$ and inspect its outputs.
We set other tuning parameters as described in Section~\ref{sec:tuning} and adopt the standardisation based on the diagonal entries of $\widehat{\mbf V}$ only (referred to as `MOSUM-diagonal' in Section~\ref{simulations}); this gives a bandwidth $\gamma = 227$ corresponding almost to one trading year.
Figure~\ref{fig:real} illustrates the series of standardised MOSUM statistics $D_{T, \gamma}^{-1} \cdot \mc T_{N, T, \gamma}(k), \, \gamma \le k \le T - \gamma$, obtained for different values of $r$, where the standardisation is applied to ensure that the MOSUM series derived from $\textsf{Vech}(\wh{\mbf g}_t \wh{\mbf g}_t^\top)$ of different dimensions are comparable. 
See also Figure~\ref{fig:real:all:7} and Table~\ref{tab:real} for the visualisation and the list of change point estimators returned by MOSUM and those competitors considered in the simulation studies, namely BSCOV \citep{li2023detection} and BDH \citep{baiduan}.

It is noteworthy that the sets of change point estimators output by the MOSUM procedure are \textit{nested} as the number of pseudo factors increases. 
That is, denoting by $\widehat{\mathcal{K}}(r)$ the set of change point estimators with $r$ as the factor number, we have $\widehat{\mathcal{K}}(r) \subset \widehat{\mathcal{K}}(r')$ for any $r < r'$, when accommodating the possible bias in the change point estimators (up to $3$ months).
Specifically, with $r \le 3$, we detect two prominent changes in mid-2008 and 2009 which, being associated with the Great Financial Crisis in 2007--2009, are detected invariably with all $r$. 
With $r = 4$, we additionally detect 2012-10-05 as a change point, which is subsequently detected for all $r \ge 5$.
With $r \ge 5$, 2020-02-24 and 2021-01-19 emerge as change points, which are accounted for by the stock market crash in February 2020 and the ensuing recession due to the COVID-19 pandemic.
With $r \in \{6, 7\}$, MOSUM outputs almost identical sets of change point estimators.
These results offer an interpretation as to how different factors `encode' different structural changes, and demonstrate that the MOSUM procedure is insensitive to the specified number of factors within certain ranges (i.e.\ $\{1, 2, 3\}$, $\{6, 7\}$).

Similarly to MOSUM, BSCOV also returns nested sets of change point estimators as $r$ increases, and many of its estimators overlap with those returned by MOSUM.
When $r = 7$, the estimators returned by MOSUM form a subset of those returned by BSCOV, and some changes detected solely by BSCOV do not appear as estimators detected with $r < 7$ by any method.
BDH estimates the number of factors by the information criterion of \cite{baing02} at each iteration of the binary segmentation algorithm and as such, when applied with a fixed number of factors, its output lacks the nested property.

\begin{figure}[h!t!]
\centering
\includegraphics[width = .9\textwidth]{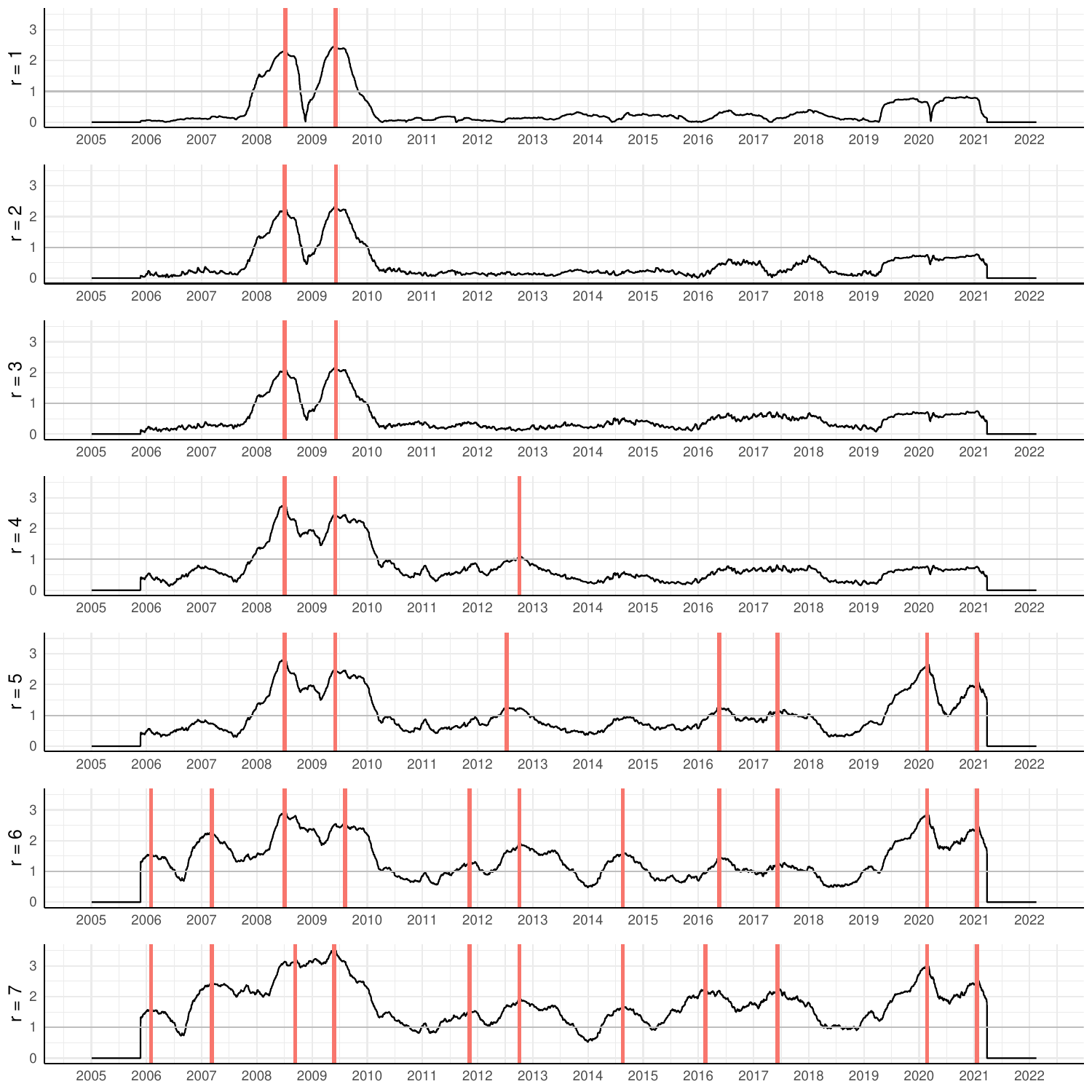}
\caption{US blue chip data: Standardised MOSUM statistics obtained with $r \in \{1, \ldots, 7\}$ (top to bottom). Vertical lines denote the change point estimators returned by MOSUM and the horizontal line is at $y = 1$.}
\label{fig:real}
\end{figure} 

\begin{figure}[h!t!]
\centering
\includegraphics[width = 1\textwidth]{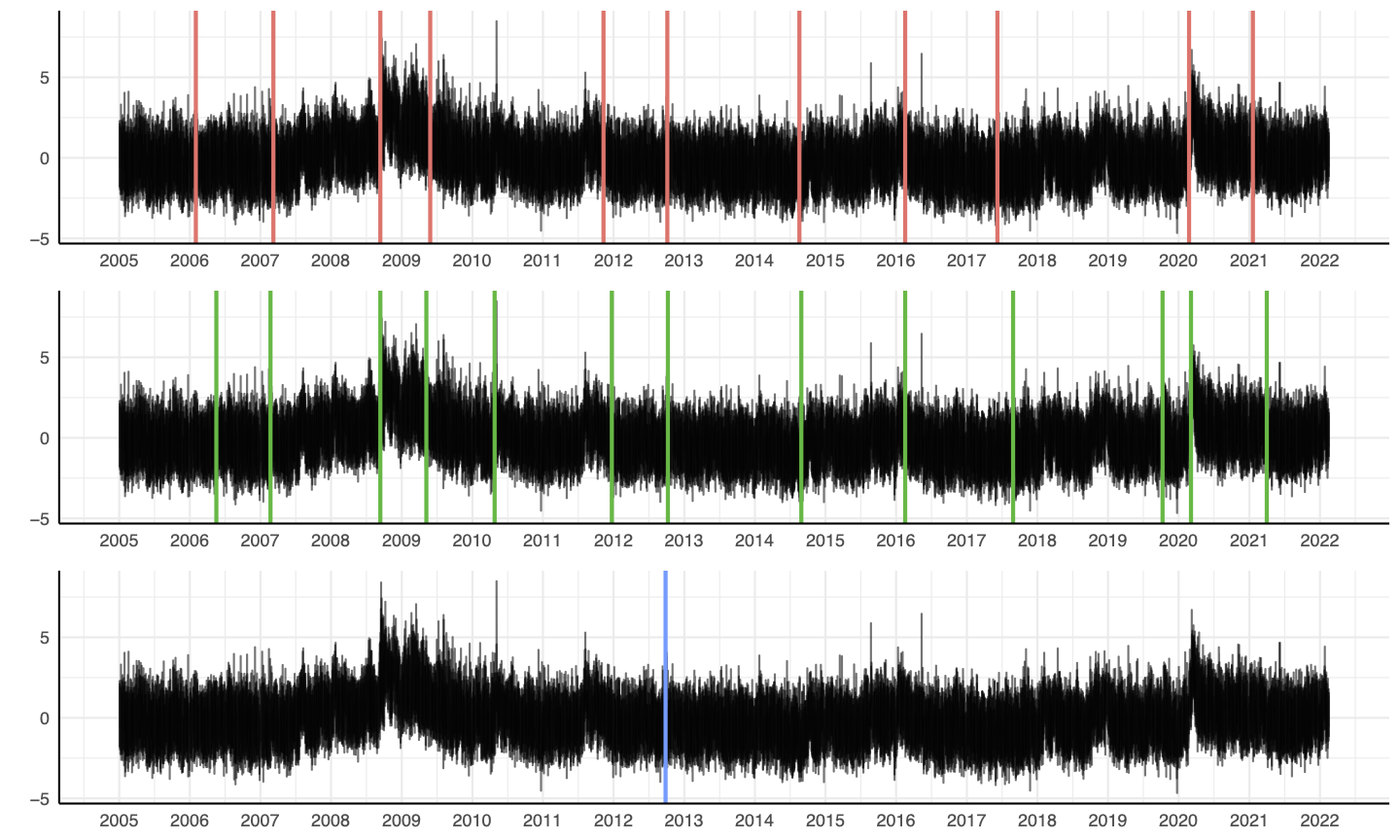}
\caption{US blue chip data: Volatilities from the $72$ companies. Vertical lines denote the change point estimators detected by MOSUM, BSCOV and BDH (top to bottom), with $r = 7$.}
\label{fig:real:all:7}
\end{figure} 

\begin{table}[h!t!b!]
\caption{US blue chip data: List of change point estimators obtained with $r \in \{1, \ldots, 7\}$ by MOSUM, BSCOV \citep{li2023detection} and BDH \citep{baiduan}.}
\label{tab:real}
\resizebox{\textwidth}{!}{  
\begin{tabular}{c l l l}
\toprule
$r$ &	MOSUM 	&	BSCOV 	&	BDH 	\\	\cmidrule(lr){1-1} \cmidrule(lr){2-2} \cmidrule(lr){3-3} \cmidrule(lr){4-4}
$1$ &	2008-07-07, 2009-06-05	&	2008-09-12, 2009-03-25	&	—	\\	\cmidrule(lr){1-1} \cmidrule(lr){2-2} \cmidrule(lr){3-3} \cmidrule(lr){4-4}
$2$ &	2008-07-07, 2009-06-05	&	2008-09-12, 2009-03-25	&	—	\\	\cmidrule(lr){1-1} \cmidrule(lr){2-2} \cmidrule(lr){3-3} \cmidrule(lr){4-4}
$3$ &	2008-07-07, 2009-06-05	&	2008-09-12, 2009-03-25	&	2008-06-27, 2029-08-07	\\	\cmidrule(lr){1-1} \cmidrule(lr){2-2} \cmidrule(lr){3-3} \cmidrule(lr){4-4}
$4$ &	2008-07-03, 2009-06-04, 2012-10-05	&	2008-09-12, 2009-05-08, 2012-10-05	&	2008-07-03, 2009-07-23, 2012-06-21,	\\	
&		&		&	2014-09-29, 2020-02-24	\\	\cmidrule(lr){1-1} \cmidrule(lr){2-2} \cmidrule(lr){3-3} \cmidrule(lr){4-4}
$5$ &	2008-07-03, 2009-06-04, 2012-07-13,	&	2008-09-12, 2009-05-08, 2012-10-08,	&	2007-10-31, 2012-10-05	\\	
&	2016-05-20, 2017-06-08, 2020-02-24,	&	2014-10-24, 2020-03-05	&		\\	
&	2021-01-19	&		&		\\	\cmidrule(lr){1-1} \cmidrule(lr){2-2} \cmidrule(lr){3-3} \cmidrule(lr){4-4}
$6$&	2006-02-01, 2007-03-09, 2008-07-03,	&	2007-01-11, 2008-09-12, 2009-05-08,	&	2012-10-05	\\	
&	2009-08-07, 2011-11-10, 2012-10-05,	&	2011-12-23, 2012-10-08, 2014-08-29,	&		\\	
&	2014-08-19, 2016-05-20, 2017-06-08,	&	2020-03-05	&		\\	
&	2020-02-24, 2021-01-19	&		&		\\	\cmidrule(lr){1-1} \cmidrule(lr){2-2} \cmidrule(lr){3-3} \cmidrule(lr){4-4}
$7$ &	2006-02-01, 2007-03-09, 2008-09-12,	&	2006-05-18, 2007-02-22, 2008-09-12,	&	2012-09-26	\\	
&	2009-05-28, 2011-11-10, 2012-10-05,	&	2009-05-08, 2010-04-26, 2011-12-23,	&		\\	
&	2014-08-19, 2016-02-17, 2017-06-08, 	&	2012-10-08, 2014-08-29, 2016-02-17, 	&		\\	
&	2020-02-24, 2021-01-19	&	2017-08-28, 2019-10-10, 2020-03-05, 2021-04-01	&		\\	\bottomrule
\end{tabular}}
\end{table}
\section{Conclusions}
\label{sec:conc}

This paper proposes a MOSUM procedure for change point analysis under a static factor model that is popularly adopted in econometrics and statistics.
In addition to deriving the asymptotic null distribution of the maximally selected MOSUM statistic, 
we establish the consistency of the procedure in multiple change point estimation with the accompanying rate of estimation, 
contributing to the relatively scarce literature on multiple change point detection in factor models.
On a range of simulated datasets and in a real data application, we demonstrate the competitiveness of the proposal empirically.
At the same time, the success of the proposed single-scale MOSUM procedure hinges on the availability of the bandwidth $\gamma$ that fulfils a set of assumptions, which may not exist in the presence of multiscale change points \citep{cho2021data}. 
One natural avenue for an extension is to apply the MOSUM procedure with a range of bandwidths, which we leave for future research. 

\bibliographystyle{apalike}
\bibliography{LTbiblio}

\clearpage
\appendix

\numberwithin{equation}{section} \numberwithin{figure}{section} %
\numberwithin{table}{section} \numberwithin{thm}{section}

\section{Proofs}

Throughout, we write $C_{NT} = \sqrt{\min(N, T)}$, and denote by $c_i \in (0, \infty), \, i \ge 0$, some fixed constants, and by $\epsilon \in (0, 1)$ a small constant which may vary from one occasion to another.

\subsection{Preliminary lemmas}
\label{sec:prem:proofs}

\label{lemmas}

The following quantities are used extensively in \citet{bai03} and also in
our proofs: 
\begin{align*}
\gamma_{s,t}=\mathsf{E}\left( \frac{1}{N}\sum_{i=1}^{N}e_{i,t}e_{i,s}%
\right) ,& \qquad \zeta_{s,t}=\frac{1}{N}\sum_{i=1}^{N}e_{i,t}e_{i,s}-%
\gamma_{s,t}, \\
\eta_{s,t}=\frac{1}{N}\sum_{i=1}^{N}\mathbf{g}_{s}^{\top }\bm\lambda
_{i}e_{i,t},& \qquad \xi_{s,t}=\frac{1}{N}\sum_{i=1}^{N}\mathbf{g}%
_{t}^{\top }\bm\lambda_{i}e_{i,s}.
\end{align*}
Then, it holds that 
\begin{align}
\widehat{\mathbf{g}}_{t}-\mathbf{H}^{\top }\mathbf{g}_{t}& =\bm\Phi
_{NT}^{-1}\left( \frac{1}{T}\sum_{s=1}^{T}\widehat{\mathbf{g}}_{s}\gamma
_{s,t}+\frac{1}{T}\sum_{s=1}^{T}\widehat{\mathbf{g}}_{s}\zeta_{s,t}+\frac{1%
}{T}\sum_{s=1}^{T}\widehat{\mathbf{g}}_{s}\eta_{s,t}+\frac{1}{T}%
\sum_{s=1}^{T}\widehat{\mathbf{g}}_{s}\xi_{s,t}\right) ,  \label{eq:bai} \\
\text{where \ }\mathbf{H}& =\frac{1}{NT}(\bm\Lambda^{\top }\bm\Lambda )(%
\mathbf{G}^{\top }\widehat{\mathbf{G}})\bm\Phi_{NT}^{-1},  \label{eq:H}
\end{align}
with $\bm\Phi_{NT}$ denoting the $r\times r$-diagonal matrix with the $r$
largest eigenvalues of $(NT)^{-1}\mathbf{X}\mathbf{X}^{\top}$ on its diagonal.

\begin{lem}
\label{lem:gg} Under Assumptions~\ref{factors} (with $\rho = 1$) and~\ref{assum:cps}, we have 
\begin{align*}
\E\l( \l\Vert \frac{1}{T} \mathbf{G}^\top \mathbf{G }- \bm\Sigma_G \r\Vert^2 \r) \le c_0 T^{-1}.
\end{align*}
\end{lem}

\begin{proof}
Firstly, we show that for any $j = 0, \ldots, R$,
\begin{align}
\label{eq:ff}
\E\l( \l\Vert \frac{1}{k_{j + 1} - k_j} \sum_{t = k_j + 1}^{k_{j + 1}} \mbf f_t \mbf f_t^\top - \bm\Sigma_F \r\Vert^2 \r) \le c_0 T^{-1}.
\end{align}
We begin by showing that Assumption~\ref{factors}~\ref{assum:factors:one} entails that $\{ \mbf f_t\mbf f_t^\top - \bm\Sigma_F \}$ is an $\mc L_{\phi}$-decomposable Bernoulli shift with some $\phi > 2$. 
Without loss of generality, let $r = d = 1$ for simplicity.
Noting that $f_t^{2} = h^2(\eta_t, \eta_{t - 1}, \ldots)$, consider the construction
\begin{align*}
\widetilde{f}_{t, \ell}^{2} = h^{2}(\eta_t, \ldots, \eta_{t - \ell}, \eta_{t - \ell -1}^{\prime}, \eta_{t - \ell - 2}^\prime, \ldots),
\end{align*}
where $\{\eta_t^\prime\}_{t \in \Z}$ is a sequence of i.i.d\ copies of $\eta_0$ independent of $\{\eta_t\}_{t \in \Z}$, such that $\eta_t \overset{\mathcal{D}}{=} \eta^\prime_t$. 
Then, $f_{t}^{2} - \widetilde{f}_{t, \ell}^{2} = ( f_{t} + \wt{f}_{t, \ell} ) ( f_{t} - \wt{f}_{t, \ell} )$ whence, using the Cauchy-Schwartz inequality and Minkowski's inequality,
\begin{align*}
\l\vert f_{t}^{2} - \widetilde{f}_{t, \ell}^{2} \r\vert_{\phi} \leq
\l\vert f_{t} + \wt{f}_{t, \ell} \r\vert_{2\phi} \l\vert f_{t} - \wt{f}_{t, \ell} \r\vert_{2\phi} \leq 
2 \l\vert f_{t} \r\vert_{2\phi} \l\vert f_{t} -\wt f_{t, \ell} \r\vert_{2\phi}.
\end{align*}
By Assumption~\ref{factors}~\ref{assum:factors:one}, we have $\vert f_{t} \vert_{2\phi} < \infty$ provided that $\phi \le 4$, and also that $\vert f_{t} - \wt f_{t, \ell} \vert_{2\phi} \leq c_{0} \ell^{-a}$ with some $a > 2$. 
Hence, $\{ f_{t}^{2} - \E(f_t^2) \}$ is an $\mc L_{\phi}$-decomposable Bernoulli shift with some $a > 2$ such that
\begin{align*}
\l\vert f_{t}^{2} - \wt{f}_{t, \ell}^2 \r\vert_{\phi} \leq c_1 \ell^{-a}.
\end{align*}
Then, \eqref{eq:ff} follows from Lemma~S2.1 of \cite{aue2014}.
This, combined with Assumption~\ref{assum:cps} and the following observations,
\begin{align*}
\l\Vert \frac{1}{T} \sum_{t = 1}^T \mbf g_t \mbf g_t^\top - \bm\Sigma_G \r\Vert
\le \sum_{j = 0}^R (\tau_{j + 1} - \tau_j) \l\Vert \mbf A_j \l( \frac{1}{k_{j + 1} - k_j} \sum_{t = k_j + 1}^{k_{j + 1}} \mbf f_t \mbf f_t^\top - \bm\Sigma_F \r) \mbf A_j^\top \r\Vert
\\
\le \sum_{j = 0}^R (\tau_{j + 1} - \tau_j) \Vert \mbf A_j \Vert^2 \l\Vert \frac{1}{k_{j + 1} - k_j} \sum_{t = k_j + 1}^{k_{j + 1}} \mbf f_t \mbf f_t^\top - \bm\Sigma_F \r\Vert,
\end{align*}
concludes the proof.
\end{proof}

\begin{lem}
\label{lem:Phi} Let Assumptions~\ref{factors}--\ref{assum:cps} hold, and denote
by $\bm\Phi \in \mathbb{R}^{r\times r}$ the diagonal matrix containing the
eigenvalues of $\bm\Sigma_{\Lambda }^{1/2}\bm\Sigma_{G}\bm\Sigma_{\Lambda
}^{1/2}$ on its diagonal. Then, we have 
\begin{equation*}
\E\l( \Vert \bm\Phi_{NT}-\bm\Phi \Vert^{4 \rho + \epsilon} \r) \le c_0 C_{NT}^{- 4\rho - \epsilon},
\end{equation*}
where $\rho$ is as in Assumptions~\ref{factors}, \ref{idiosyncratic} and~\ref{depFE}.
\end{lem}

\begin{proof}
The proof follows from standard arguments, which we briefly summarise. 
We note that the leading $r$ eigenvalues of $(NT)^{-1}\mathbf{X}\mathbf{X}^\top$ are
identical to those of $(NT)^{-1}\mathbf{X}^\top\mathbf{X}$, and
\[
\mathbf{X}^\top\mathbf{X}= \bm\Lambda \mathbf{G}^\top \mathbf{G} \bm\Lambda^\top + \mathbf{E}^\top \mathbf{E} + 
\bm\Lambda \mathbf{G}^\top \mathbf{E} + \mathbf{E}^\top \mathbf{G} \bm\Lambda^\top.
\]%
Let $\Lambda_{j}\left( \mathbf{A}\right) $ denote the $j$-th
eigenvalue, sorted in descending order, of a matrix $\mathbf{A}$. 
Then,
\begin{align*}
\E\l( \Vert \bm\Phi_{NT}-\bm\Phi \Vert^{4\rho + \epsilon} \r) &\le \E\l[ \l( \sum_{j = 1}^r \left\vert \Lambda_{j}\left( \frac{1}{NT} \mathbf{X}^\top\mathbf{X} \right)
-\Lambda_{j}(\bm\Phi) \right\vert^2 \r)^{2\rho + \epsilon/2} \r]
\\
&\le r^{2\rho + \epsilon/2} \E\l( \sum_{j = 1}^r \left\vert \Lambda_{j}\left( \frac{1}{NT} \mathbf{X}^\top\mathbf{X} \right)
-\Lambda_{j}(\bm\Phi) \right\vert^{4\rho + \epsilon} \r).
\end{align*}
By Weyl's inequality, we have%
\begin{eqnarray*}
&&\left\vert \Lambda_{j}\left( \frac{1}{NT} \mathbf{X}^\top\mathbf{X} \right)
-\Lambda_{j}\left( \frac{1}{NT} \bm\Lambda \mathbf{G}^\top \mathbf{G} \bm\Lambda^\top \right) \right\vert^{4\rho + \epsilon}  \\
&=&\left\vert \Lambda_{j}\left( \frac{1}{NT} \mathbf{X}^\top\mathbf{X} \right)
-\Lambda_{j}\left( \frac{1}{NT} \mathbf{G}^\top \mathbf{G} \bm\Lambda^\top \bm\Lambda \right) \right\vert^{4\rho + \epsilon}  \\
&\leq & 2\l( \Lambda_{\max }\left( \frac{1}{NT} \mathbf{E}^\top \mathbf{E} \right) \r)^{4\rho + \epsilon}
+ 2\l( \Lambda_{\max }\left( \frac{1}{NT}\left( \bm\Lambda \mathbf{G}^\top \mathbf{E}
+ \mathbf{E}^\top \mathbf{G} \bm\Lambda^\top \right) \right) \r)^{4\rho + \epsilon}.
\end{eqnarray*}%
Further%
\[
\Lambda_{\max }\left( \frac{1}{NT} \mathbf{E}^\top \mathbf{E} \right) \leq
\Lambda_{\max }\left( \frac{1}{NT}\mathsf{E}\left(  \mathbf{E}^\top \mathbf{E} \right) \right) +
\left\Vert \frac{1}{NT}\left( \mathbf{E}^\top \mathbf{E} -\mathsf{E}\left(\mathbf{E}^\top \mathbf{E} \right) \right) \right\Vert.
\]%
Using Assumption~\ref{idiosyncratic}~\ref{assum:idiosyncratic:five},
\[
\Lambda_{\max }\left( \frac{1}{NT}\mathsf{E}\left(  \mathbf{E}^\top \mathbf{E} \right) \right) \leq \frac{1}{NT}\max_{1\leq i\leq
N}\sum_{j=1}^{N}\sum_{s=1}^{T}\left\vert \mathsf{E}\left(
e_{i,t}e_{j,t}\right) \right\vert \leq c_{0}N^{-1}.
\]%
Similarly, on account of Assumption~\ref{idiosyncratic}~\ref{assum:idiosyncratic:six},
\begin{eqnarray*}
&&\mathsf{E}\left( \left\Vert \frac{1}{NT}\left( \mathbf{E}^\top \mathbf{E} -\mathsf{E}%
\left( \mathbf{E}^\top \mathbf{E} \right) \right) \right\Vert_{F}^{4\rho + \epsilon} \right) \\
&=& 
\E\l\{ \l[ \frac{1}{N^2} \sum_{i, j = 1}^N \l\vert \frac{1}{T} \sum_{t = 1}^T \l(e_{i, t} e_{j, t} - \E(e_{i, t} e_{j, t}) \r) \r\vert^2 \r]^{2\rho + \epsilon/2} \r\}
\\
&\le& \E\l[ \frac{1}{N^2} \sum_{i, j = 1}^N \l\vert \frac{1}{T} \sum_{t = 1}^T \l(e_{i, t} e_{j, t} - \E(e_{i, t} e_{j, t}) \r) \r\vert^{4\rho + \epsilon} \r]
\le c_0 T^{-2\rho - \epsilon/2},
\end{eqnarray*}%
and therefore%
\[
\l( \Lambda_{\max }\left( \frac{1}{NT} \mathbf{E}^\top \mathbf{E} \right) \r)^{4\rho + \epsilon} \le c_0 C_{N, T}^{-4\rho - \epsilon}.
\]%
Besides, by Assumption~\ref{depFE}~\ref{assum:depFE:three},
\begin{eqnarray*}
&&\mathsf{E}\left( \left\Vert \frac{1}{NT} \bm\Lambda \mathbf{G}^\top \mathbf{E} \right\Vert_{F}^{4\rho + \epsilon} \right) \\
&=& \E\l\{ \l[ \frac{1}{N^2} \sum_{i, j = 1}^N \l( \frac{1}{T} \sum_{t = 1}^T \bm\lambda_i^\top \mathbf{g}_t e_{j, t} \r)^2 \r]^{2\rho + \epsilon/2} \r\}
\\
&\le& \frac{1}{N^2} \sum_{i, j = 1}^N \E\l[ \l( \frac{1}{T} \sum_{t = 1}^T \bm\lambda_i^\top \mathbf{g}_t e_{j, t} \r)^{4\rho + \epsilon} \r] \le c_0 T^{-2\rho- \epsilon/2}
\end{eqnarray*}%
and $\E(\Vert (NT)^{-1} \mathbf{E}^\top \mathbf{G} \bm\Lambda^\top \Vert_F^{4\rho + \epsilon})$ is similarly bounded.
Finally, denoting by $\mbf B = \bm\Sigma_\Lambda^{1/2} \bm\Sigma_G \bm\Sigma_\Lambda^{1/2}$ and
$\mbf B_{NT} = (NT)^{-1} (\bm\Lambda^\top\bm\Lambda)^{1/2} ( \mbf G^\top \mbf G ) (\bm\Lambda^\top\bm\Lambda)^{1/2}$,
we have $\Vert \mbf B_{NT} - \mbf B \Vert = O_P(C_{NT}^{-1})$ by Assumption~\ref{loadings}~\ref{assum:loadings:two} and Lemma~\ref{lem:gg}.
The desired result now follows from putting all the bounds together.
\end{proof}

\citet{bai03} proves that $T^{-1} \sum_{t = 1}^{T} \Vert \widehat{\mathbf{g}}%
_{t} - \mathbf{H}^\top \mathbf{g}_{t} \Vert^2$ and $T^{-1} \Vert \sum_{t =
1}^{T} \mathbf{g}_t (\widehat{\mathbf{g}}_{t} - \mathbf{H}^\top \mathbf{g}%
_{t})^\top \Vert$ are bounded in probability. We report two results of
independent interest by deriving the upper bounds on the two terms in $\mathcal{L}_{\delta}$-norm for some $\delta \geq 1$.

\begin{lem}
\label{moment-1} Suppose that Assumptions~\ref{factors} and~%
\ref{idiosyncratic} hold with $\rho = 1$. Then it follows that 
\begin{equation*}
\mathsf{E}\left( \left\Vert \sum_{t=1}^{T}\left( \widehat{\mathbf{g}}_{t}-%
\mathbf{H}^{\top }\mathbf{g}_{t}\right) \left( \widehat{\mathbf{g}}_{t}-%
\mathbf{H}^{\top }\mathbf{g}_{t}\right)^{\top}\right\Vert^{\delta
}\right) \leq c_{0}\left( TC_{NT}^{-2}\right)^{\delta }
\end{equation*}%
for all $1\leq \delta \leq 2 + \epsilon$, with $\mathbf{H}$ defined in~%
\eqref{eq:H}.
\end{lem}

\begin{proof}

To simplify the notation, we set $r = d = 1$ and omit the matrix $\mbf H$. 
By convexity,
\begin{align*}
\E \l( \l\vert \frac{1}{T}\sum_{t = 1}^{T} \l( \widehat{g}_{t} - g_{t} \r)^{2} \r\vert^{\delta} \r) \leq \frac{1}{T} \sum_{t = 1}^{T} \E\l( \l\vert \widehat{g}_{t} - g_t \r\vert^{2\delta} \r).
\end{align*}
Using~\eqref{eq:bai},
\begin{align}
\frac{1}{T} \sum_{t = 1}^{T} \E\l( \l\vert \widehat{g}_{t} - g_t \r\vert^{2\delta} \r)
& \leq 
T^{-2\delta - 1} \sum_{t = 1}^{T} \E\l( \l\vert \sum_{s = 1}^{T} \widehat{g}_{s} \gamma_{s, t} \r\vert^{2\delta} \r)
+ T^{-2\delta - 1} \sum_{t = 1}^{T} \E\l( \l\vert \sum_{s = 1}^{T} \widehat{g}_{s} \zeta_{s, t} \r\vert^{2\delta} \r) 
\nn \\
& + T^{-2\delta -  1} \sum_{t = 1}^{T} \E\l( \l\vert \sum_{s = 1}^{T} \widehat{g}_{s} \eta_{s, t} \r\vert^{2\delta} \r) 
+ T^{-2\delta - 1} \sum_{t = 1}^{T} \E\l( \l\vert \sum_{s = 1}^{T}\widehat{g}_{s} \xi_{s, t} \r\vert^{2\delta}\r)
\nn \\
&=: T_1 + T_2 + T_3 + T_4.   
\label{bai03-1}
\end{align}

We now study each of these terms. 
By construction, $\sum_{s = 1}^{T} \widehat{g}_{s}^{2} = T$. 
Also by Assumption~\ref{idiosyncratic}~\ref{assum:idiosyncratic:two} and Lemma~1~(i) in \citet{baing02}, it follows that $\sum_{s = 1}^{T}\gamma_{s, t}^{2} \leq c_{0}$.
Therefore,
\begin{align}
T_1 
\leq 
\E\l( \sum_{t = 1}^{T} \l\vert \sum_{s = 1}^{T}
\widehat{g}_{s}^{2} \r\vert^{\delta} 
\l\vert \sum_{s = 1}^{T} \gamma_{s, t}^2 \r\vert^{\delta} \r) \le c_0 T^{-\delta}.
\label{a-1}
\end{align}
Next, from Assumption~\ref{idiosyncratic}~\ref{assum:idiosyncratic:three},
\begin{align*}
T^{2\delta + 1} \cdot T_2  
&\leq 
\sum_{t = 1}^{T} \E\l( \l\vert \l( \sum_{s = 1}^{T} \widehat{g}_{s}^{2} \r)^{1/2} \l( \sum_{s = 1}^{T} \zeta_{s, t}^{2} \r)^{1/2} \r\vert^{2\delta} \r) 
\\
&\leq T^{\delta} \sum_{t = 1}^{T} \E\l( \l\vert \sum_{s = 1}^{T}\zeta_{s, t}^{2} \r\vert^{\delta} \r) 
\leq T^{2\delta - 1} \sum_{t = 1}^{T} \sum_{s = 1}^{T} \E\l( \l\vert \zeta_{s, t} \r\vert^{2\delta} \r)
\leq T^{2\delta + 1} N^{-\delta},
\end{align*}
Therefore it follows that
\begin{align}
T_2
\le c_0 N^{-\delta}. 
\label{a-2}
\end{align}
Similarly, by the Cauchy-Schwartz inequality,
\begin{align*}
T^{2\delta + 1} \cdot T_3
&= \sum_{t = 1}^{T} \E\l( \l\vert \sum_{s = 1}^{T}\widehat{g}_{s} g_{s} \frac{1}{N} \sum_{i = 1}^{N}\lambda_{i} e_{i,t} \r\vert^{2\delta} \r) 
\\
&\leq \sum_{t = 1}^{T} \E\l( \l\vert \frac{1}{N}\sum_{i = 1}^{N}\lambda_{i} e_{i, t} \r\vert^{2\delta} \l\vert \sum_{s = 1}^{T}\widehat{g}_{s}^{2} \r\vert^{\delta} \l\vert \sum_{s = 1}^{T} g_{s}^{2} \r\vert^{\delta} \r)
\\
&= T^{\delta} \sum_{t = 1}^{T} \E\l( \l\vert \frac{1}{N} \sum_{i = 1}^{N}\lambda_{i} e_{i,t}\r\vert^{2\delta} \l\vert \sum_{s = 1}^{T}g_{s}^{2} \r\vert^{\delta} \r) 
\\
&\leq T^{\delta} \sum_{t = 1}^{T} \l[ \E \l( \l\vert \frac{1}{N} \sum_{i = 1}^{N} \lambda_{i}e_{i, t} \r\vert^{4\delta} \r) \r]^{1/2}
\l[ \E \l( \l\vert \sum_{s = 1}^{T} g_{s}^{2}\r\vert^{2\delta} \r) \r]^{1/2}.
\end{align*}
By Assumption~\ref{idiosyncratic}~\ref{assum:idiosyncratic:four},
\begin{align*}
\E\l( \l\vert \frac{1}{N} \sum_{i = 1}^{N}\lambda_{i} e_{i,t} \r\vert^{4\delta} \r) \leq c_{0}N^{-2\delta}.
\end{align*}
Also, from Assumption~\ref{factors}~\ref{assum:factors:one},
\begin{align*}
\E\l( \l\vert \sum_{s = 1}^{T} g_{s}^{2}\r\vert^{2\delta} \r) \leq T^{2\delta - 1} \sum_{s = 1}^{T} \E\l( \l\vert g_{s}\r\vert^{4\delta} \r) \leq
c_{0} T^{2\delta}.
\end{align*}
Putting all together, we have
\begin{align}
T_3
\le c_0 N^{-\delta}.  
\label{a-3}
\end{align}
Finally, we consider
\begin{align*}
T^{2\delta + 1} \cdot T_4
&\leq \sum_{t = 1}^{T} \E\l( \l\vert \sum_{s = 1}^{T}\widehat{g}_{s}^{2} \r\vert^{\delta} \l\vert \sum_{s = 1}^{T} \xi_{s, t}^{2} \r\vert^{\delta} \r) 
\\
&= T^{\delta} \sum_{t = 1}^{T} \E\l( \l\vert \sum_{s = 1}^{T}\xi_{s, t}^{2} \r\vert^{\delta}\r) 
\leq T^{2\delta -1}\sum_{t = 1}^{T} \sum_{s = 1}^{T} \E\l( \l\vert \xi_{s, t} \r\vert^{2\delta}\r).
\end{align*}
By Assumptions~\ref{idiosyncratic}~\ref{assum:idiosyncratic:four} and~\ref{factors}~\ref{assum:factors:one},
\begin{align*}
\E\l( \l\vert \xi_{s, t} \r\vert^{2\delta}\r) = \E\l( \l\vert \frac{1}{N} \sum_{i = 1}^{N} g_{t} \lambda_{i} e_{i,s} \r\vert^{2\delta} \r) 
\leq \l[ \E \l(\l\vert \frac{1}{N} \sum_{i = 1}^{N} \lambda_{i} e_{i,s} \r\vert^{4\delta}\r) \r]^{1/2}\l[ \E \l( \l\vert g_{t} \r\vert^{4\delta}\r) \r]^{1/2} \leq c_{0} N^{-\delta},
\end{align*}
Therefore it follows that
\begin{align}
T_4 
\le c_0 N^{-\delta}.
\label{a-4}
\end{align}
Putting together~\eqref{a-1}--\eqref{a-4} into~\eqref{bai03-1}, the desired result follows.
\end{proof}

\begin{lem}
\label{moment-2} Suppose that Assumptions~\ref{factors}--\ref{assum:cps} hold. Then it
holds that 
\begin{equation*}
\mathsf{E}\left( \left\Vert \sum_{t=1}^{T}\mathbf{g}_{t}\left( \widehat{\mathbf{g}}_{t}-\mathbf{H}^{\top }\mathbf{g}_{t}\right)^{\top}\right\Vert
^{\delta }\right) \leq c_{0}\left( TC_{NT}^{-2}\right)^{\delta },
\end{equation*}%
for all $1 \leq \delta \leq \rho + \epsilon$, with $\rho$ is as in Assumptions~\ref{factors}, \ref{idiosyncratic} and~\ref{depFE}.
\end{lem}

\begin{proof}
Throughout, we frequently use that for $1 \le \delta \le 4 + \epsilon$, 
\begin{align}
\label{eq:g:delta}
\E\l( \l\Vert \sum_{s = 1}^{T} \mbf g_{s} \mbf g_s^\top \r\Vert^{\delta} \r) \le T^{\delta - 1} \sum_{s = 1}^T \E\l( \Vert \mbf g_s \Vert^{2\delta} \r) \le c_0 T^\delta,
\end{align}
from Assumption~\ref{factors}~\ref{assum:factors:one}.
Thanks to~\eqref{eq:bai}, we have
\begin{align}
\E\l(\l\Vert \sum_{t = 1}^{T} \mbf g_{t}\l( \widehat{\mbf g}_{t} - \mbf H^\top \mbf g_{t} \r)^\top \r\Vert^{\delta}  \r)
&\leq 
\E\l( \l\Vert \frac{1}{T} \sum_{s, t = 1}^{T} \widehat{\mbf g}_{s} \mbf g_{t}^\top \gamma_{s, t} \r\Vert^{\delta} \r) +
\E\l( \l\Vert \frac{1}{T} \sum_{s, t = 1}^{T} \widehat{\mbf g}_{s} \mbf g_{t}^\top \zeta_{s, t} \r\Vert^{\delta} \r) 
\nn \\
& + \E\l( \l\Vert \frac{1}{T} \sum_{s, t = 1}^{T}\widehat{\mbf g}_{s} \mbf g_{t}^\top \eta_{s, t} \r\Vert^{\delta} \r) 
+ \E\l( \l\Vert \frac{1}{T} \sum_{s, t = 1}^{T}\widehat{\mbf g}_{s} \mbf g_{t}^\top \xi_{s, t} \r\Vert^{\delta} \r)
\nn \\
&=: T_1 + T_2 + T_3 + T_4. 
\label{dec-lemma-2} 
\end{align}
We begin with $T_4$ which is bounded as
\begin{align*}
T_4 \le 
T^{-\delta} \E\l( \l\Vert \sum_{s, t = 1}^T (\wh{\mbf g}_s - \mbf H^\top \mbf g_s) \mbf g_t^\top \xi_{s, t} \r\Vert^\delta \r)
+ T^{-\delta} \E\l( \l\Vert \mbf H^\top \sum_{s, t = 1}^T \mbf g_s \mbf g_t^\top \xi_{s, t} \r\Vert^\delta \r) =: T_{4, 1} + T_{4, 2}.
\end{align*}

By Assumptions~\ref{factors}, \ref{loadings} and~\ref{assum:cps} and Lemma~\ref%
{lem:Phi}, we have 
\begin{align}
\E\l( \l\Vert \mbf H \r\Vert^{p\delta} \r) 
&\leq 
\l\Vert \frac{\bm\Lambda^\top \bm\Lambda}{N} \r\Vert^{p\delta} \E\l( \l\Vert \frac{1}{\sqrt{T}} \widehat{\mbf G}\r\Vert^{p\delta} \l\Vert \frac{1}{\sqrt{T}} \mbf G \r\Vert^{p\delta} \Vert \bm\Phi_{NT}^{-1} \Vert^{p\delta} \r) 
\nn \\
&\leq c_{0} \l\{ \E\l[ \l( \frac{1}{T} \sum_{t = 1}^T \Vert \mbf g_t \Vert^2 \r)^{p\delta} \r] \; \E\l( \Vert \bm\Phi_{NT}^{-1} \Vert^{2p\delta} \r) \r\}^{1/2}
\nn \\
&\leq c_1 \l\{\frac{1}{T} \sum_{t = 1}^{T} \E\l( \l\Vert \mbf g_{t} \r\Vert^{2 p\delta} \r) \r\}^{1/2} \le c_2
\label{eq:H:bound}
\end{align}
for $1 \le p \le 2$. 
From this, we obtain 
\begin{align*}
T^\delta \cdot T_{4, 2} &\le \E\l( \Vert \mbf H^\top \Vert^\delta \l\Vert \sum_{s, t = 1}^\top \mbf g_s \mbf g_t^\top \xi_{s, t} \r\Vert^\delta \r)
\\
& \le \l[ \E\l( \Vert \mbf H^\top \Vert^{2\delta} \r) \r]^{1/2} \l[ \E\l( \l\Vert \sum_{s, t = 1}^\top \mbf g_s \mbf g_t^\top \xi_{s, t} \r\Vert^{2\delta} \r) \r]^{1/2}.
\end{align*} 
WLOG, we may set $r = d = 1$ for notational simplicity. 
Then by H\"{o}lder's inequality,
\begin{align*}
\E\l( \l\vert \sum_{s, t = 1}^T g_s g_t \xi_{s, t} \r\vert^{2\delta} \r)
&= \E\l[ \l( \sum_{t = 1}^T g_t^2 \r)^{2\delta} \l\vert \frac{1}{N} \sum_{s = 1}^T \sum_{i = 1}^N \lambda_i g_s e_{i, s} \r\vert^{2\delta} \r]
\\
&\le \l( \E\l[ \l( \sum_{t = 1}^T g_t^2 \r)^{4\delta} \r] \r)^{1/2} \l( \E\l[ \l\vert \frac{1}{N} \sum_{s = 1}^T \sum_{i = 1}^N \lambda_i g_s e_{i, s} \r\vert^{4\delta} \r] \r)^{1/2},
\end{align*}
where by~\eqref{eq:g:delta},
\begin{align*}
\l( \E\l[ \l( \sum_{t = 1}^T g_t^2 \r)^{4\delta} \r] \r)^{1/2} \le \l[ T^{4\delta - 1} \sum_{t = 1}^T \E( \vert g_t \vert^{8\delta}) \r]^{1/2} \le c_0 T^{2\delta},
\end{align*}
while by Assumption~\ref{depFE}~\ref{assum:depFE:two},
\begin{align*}
\l( \E\l[ \l\vert \frac{1}{N} \sum_{s = 1}^T \sum_{i = 1}^N \lambda_i g_s e_{i, s} \r\vert^{4\delta} \r] \r)^{1/2}
\le c_0 (T^{1/2} N^{-1/2})^{2\delta}.
\end{align*}
Altogether, this yields
\begin{align}
T_{4, 2} \le c_0 T^{\delta/2} N^{-\delta/2}.
\label{b-4-2}
\end{align}
Next, we note that by H\"{o}lder's inequality with some $1 < p \le 4$,
\begin{align*}
T_{4, 1} &=
T^{-\delta} \E\l( \l\vert \sum_{t = 1}^{T} g_{t}^{2} \cdot \sum_{s = 1}^{T} \l( \widehat{g}_{s} - g_{s} \r) \frac{1}{N} \sum_{i = 1}^{N} \lambda_{i} e_{i, s} \r\vert^{\delta} \r)
\\
&\leq 
T^{-\delta} \l[ \E\l( \l\vert \sum_{t = 1}^{T} g_{t}^{2} \r\vert^{p\delta} \r) \r]^{\frac{1}{p}} \l[ \E\l( \l\vert \sum_{s = 1}^{T} \l( \widehat{g}_{s} - g_{s} \r) \frac{1}{N} \sum_{i = 1}^{N} \lambda_{i} e_{i, s}\r\vert^{\frac{p\delta}{p-1}}\r) \r]^{\frac{p-1}{p}}
\\
&\leq c_0 \l[ \E\l( \l\vert \sum_{s = 1}^{T} \l( \widehat{g}_{s} - g_{s} \r) \frac{1}{N} \sum_{i = 1}^{N} \lambda_{i} e_{i, s}\r\vert^{\frac{p\delta}{p-1}}\r) \r]^{\frac{p-1}{p}},
\end{align*}
where the last passage follows from~\eqref{eq:g:delta}.
Again applying H\"{o}lder's inequality,
\begin{align*}
& \E\l( \l\vert \sum_{s = 1}^{T} \l( \widehat{g}_{s} - g_{s} \r) \l( \frac{1}{N} \sum_{i = 1}^{N} \lambda_{i} e_{i, s} \r) \r\vert^{\frac{p\delta}{p-1}}\r)
\\
&\leq 
\E\l( \l\vert \sum_{s = 1}^{T} \l( \widehat{g}_{s} - g_{s}\r)^{2} \r\vert^{\frac{p\delta}{2(p - 1)}} \l\vert \sum_{s = 1}^{T} \l( \frac{1}{N} \sum_{i = 1}^{N} \lambda_{i} e_{i, s} \r)^{2} \r\vert^{\frac{p\delta}{2(p - 1)}} \r)  
\\
&\leq \l[ \E\l( \l\vert \sum_{s = 1}^{T} \l(  \widehat{g}_{s} - g_{s}\r)^{2} \r\vert^{\frac{pq\delta}{2(p - 1)}} \r) \r]^{\frac{1}{q}} \l[ \E\l( \l\vert \sum_{s = 1}^{T} \l( \frac{1}{N}\sum_{i = 1}^{N} \lambda_{i} e_{i, s} \r)^{2} \r\vert^{\frac{pq\delta}{2(p - 1)(q - 1)}} \r) \r]^{\frac{q - 1}{q}}.
\end{align*}
Setting $p = 4$ and $q = 3/2$, we have $pq\delta/(p - 1) = 2\delta$ such that by Lemma~\ref{moment-1},
\begin{align*}
\E\l( \l\vert \sum_{s = 1}^{T} \l(  \widehat{g}_{s} - g_{s}\r)^{2} \r\vert^{\frac{pq\delta}{2(p - 1)}} \r) \leq c_{0} \l( T C_{NT}^{-2} \r)^{\delta}.
\end{align*}
Besides, since the choices of $p$ and $q$ lead to 
\begin{align*}
\frac{pq\delta}{2(p - 1)(q - 1)} = 2\delta,
\end{align*}
applying Assumption~\ref{idiosyncratic}~\ref{assum:idiosyncratic:four},
\begin{align*}
\E\l( \l\vert \sum_{s = 1}^{T} \l( \frac{1}{N}\sum_{i = 1}^{N} \lambda_{i} e_{i, s} \r)^{2} \r\vert^{2\delta} \r) 
& \leq 
T^{2\delta - 1} \sum_{s = 1}^{T}\E\l( \l\vert \frac{1}{N} \sum_{i = 1}^{N} \lambda_{i} e_{i, s} \r\vert^{4\delta} \r)
\\
& \leq c_{0}T^{2\delta} N^{- 2\delta}.
\end{align*}
From the above arguments and~\eqref{eq:g:delta}, it follows that
\begin{align}
\label{b-4-1}
T_{4, 1} \le c_0 T^{\delta} N^{-\delta/2} C_{NT}^{-\delta}. 
\end{align}
Collecting the bounds on $T_{4, 1}$ and $T_{4, 2}$, we have
\begin{align}
T_4 \leq c_{0} T^{\delta} N^{-\delta/2} C_{NT}^{-\delta}. 
\label{term-4}
\end{align}

For the rest of the terms, we may proceed analogously.
From~\eqref{eq:H:bound}, for simplicity, we check the steps by setting $r = d = 1$ and omitting $\mbf H$.
Note that
\begin{align*}
T_1 \leq T^{-\delta} \E\l( \l\vert
\sum_{s, t = 1}^{T}\l( \widehat{g}_{s} - g_{s}\r) g_{t} \gamma_{s, t} \r\vert^{\delta} \r) + T^{-\delta} \E\l( \l\vert \sum_{s, t = 1}^{T} g_{s} g_{t} \gamma_{s, t} \r\vert^{\delta} \r) =: T_{1, 1} + T_{1, 2}.
\end{align*}
Using the Cauchy-Schwartz inequality twice,
\begin{align*}
T_{1, 1} &\le T^{-\delta} \E\l( \l\vert \sum_{s = 1}^{T}\l( \widehat{g}_{s} - g_{s} \r)^{2} \r\vert^{\delta/2} \l\vert \sum_{s = 1}^{T} \l( \sum_{t = 1}^{T} g_{t} \gamma_{s, t} \r)^{2} \r\vert^{\delta /2} \r)
\\
&\leq T^{-\delta} \l[ \E\l( \l\vert \sum_{s = 1}^{T}\l( \widehat{g}_{s} - g_{s} \r)^{2} \r\vert^{\delta} \r) \r]^{1/2} \l[ \E\l( \l\vert \sum_{s = 1}^{T} \l( \sum_{t = 1}^{T} g_{t}\gamma_{s, t} \r)^{2} \r\vert^{\delta} \r) \r]^{1/2} 
\\
&\le c_0 T^{-\delta /2} C_{NT}^{-\delta} \l[ T^{\delta - 1} \sum_{s = 1}^{T} \E \l( \l\vert \sum_{t = 1}^{T} g_{t} \gamma_{s, t} \r\vert^{2\delta} \r) \r]^{1/2},
\end{align*}
having used Lemma~\ref{moment-1} in the last passage. 
Let $n_{0} = \lceil 4\delta \rceil$ and recall that, by Assumption~\ref{factors}~\ref{assum:factors:one}, $\E(\l\vert g_{t}\r\vert^{n_{0}}) < \infty$.
Using the $\mathcal{L}_{p}$-norm inequality, we have
\begin{align*}
\E\l( \l\vert \sum_{t = 1}^{T} g_{t} \gamma_{s, t}\r\vert^{2\delta} \r) &\leq \l[ \E\l( \l\vert \sum_{t = 1}^{T} g_{t} \gamma_{s, t} \r\vert^{n_{0}}\r) \r]^{2\delta /n_{0}}, \text{ \ and}
\\
\E\l( \l\vert \sum_{t = 1}^{T} g_{t} \gamma_{s, t}\r\vert^{n_{0}} \r) &=
\E\l( \l\vert \sum_{t_{1}, \ldots, t_{n_{0}} = 1}^{T} g_{t_{1}} \cdot \ldots \cdot g_{t_{n_{0}}} \cdot \gamma_{s, t_{1}} \cdot \ldots \cdot \gamma_{s, t_{n_{0}}} \r\vert \r)
\\
&\leq \sum_{t_{1}, \ldots, t_{n_{0}} = 1}^{T} 
\E\l( \prod\limits_{i=1}^{n_{0}} \l\vert g_{t_{i}}\r\vert \r) \prod\limits_{i = 1}^{n_{0}}\l\vert \gamma_{s, t_{i}} \r\vert 
\\
&\leq \max_{1 \leq t \leq T} \l\vert g_{t_{i}} \r\vert_{n_{0}} \cdot \l( \sum_{s = 1}^{T} \l\vert \gamma_{s, t} \r\vert \r)^{n_{0}} \leq c_{0},
\end{align*}
where the penultimate passage follows from H\"{o}lder's inequality, and the last passage from Assumption~\ref{idiosyncratic}~\ref{assum:idiosyncratic:two}.
The above entails that $\E(\vert \sum_{t = 1}^{T}g_{t}\gamma_{s, t} \vert^{2\delta}) \leq c_{0}$, and therefore $T_{1, 1} \le c_0 C_{NT}^{-\delta}$.
Similarly, we have from~\eqref{eq:g:delta},
\begin{align*}
T_{1, 2} & \leq T^{-\delta} \E \l( \l\vert \sum_{s = 1}^{T}g_{s}^{2} \r\vert^{\delta /2} \l\vert \sum_{s = 1}^{T} \l( \sum_{t = 1}^{T} g_{t}\gamma_{s, t} \r)^{2} \r\vert^{\delta /2} \r)
\\
&\leq T^{-\delta} \l[ \E\l( \l\vert \sum_{s = 1}^{T}g_{s}^{2} \r\vert^{\delta} \r) \r]^{1/2} \l[  T^{\delta - 1} \sum_{s = 1}^{T} \E\l( \l\vert \sum_{t = 1}^{T} g_{t} \gamma_{s, t} \r\vert^{2\delta} \r) \r]^{1/2} 
\\
&\leq c_{0} T^{-\delta} \cdot T^{\delta / 2} \cdot T^{\delta / 2} = c_0,
\end{align*}
so that we finally have
\begin{align}
T_1 \le c_0. \label{term-1}
\end{align}

We now turn to studying 
\begin{align*}
T_2 \leq T^{-\delta} \E\l( \l\vert \sum_{s = 1}^{T} \l( \widehat{g}_{s} - g_{s}\r) \l( \sum_{t = 1}^{T} g_{t} \zeta_{s, t} \r) \r\vert^{\delta} \r) 
+ T^{-\delta} \E\l( \l\vert \sum_{s = 1}^{T} g_{s} \l( \sum_{t = 1}^{T} g_{t} \zeta_{s, t} \r) \r\vert^{\delta} \r) =: T_{2, 1} + T_{2, 2}.
\end{align*}
Using the Cauchy-Schwartz inequality and Lemma~\ref{moment-1},
\begin{align*}
T_{2, 1} &\leq T^{-\delta} \E\l( \l\vert \sum_{s = 1}^{T} \l( \widehat{g}_{s} - g_{s} \r)^{2}\r\vert^{\delta / 2} \l\vert \sum_{s = 1}^{T} \l( \sum_{t = 1}^{T} g_{t} \zeta_{s, t} \r)^{2} \r\vert
^{\delta/2} \r)
\\
&\leq T^{-\delta} \l[ \E\l( \l\vert \sum_{s = 1}^{T}\l( \widehat{g}_{s} - g_{s} \r)^{2} \r\vert^{\delta} \r) \r]^{1/2} \l[ \E\l( \l\vert \sum_{s = 1}^{T} \l( \sum_{t = 1}^{T} g_{t}\zeta_{s, t} \r)^{2} \r\vert^{\delta} \r) \r]^{1/2} 
\\
&\le c_0 T^{-\delta/2} C_{NT}^{-\delta} \l[ T^{\delta - 1} \sum_{s = 1}^{T} \E\l( \l\vert \sum_{t = 1}^{T} g_{t}\zeta_{s, t} \r\vert^{2\delta} \r) \r]^{1/2};
\end{align*}
thence, Assumption~\ref{depFE}~\ref{assum:depFE:one} immediately entails that
\begin{align}
T_{2, 1} \le c_0 T^{-\delta/2} C_{NT}^{-\delta} \l( T^{\delta} \l( N^{-1}T\r)^{\delta} \r)^{1/2} = c_0 T^{\delta/2} N^{-\delta/2} C_{NT}^{-\delta}.  
\nn 
\end{align}
By the same token, with~\eqref{eq:g:delta},
\begin{align*}
T_{2, 2} 
&\leq T^{-\delta} \E\l( \l\vert \sum_{s = 1}^{T} g_{s}^{2}\r\vert^{\delta / 2} \l\vert \sum_{s = 1}^{T} \l( \sum_{t = 1}^{T} g_{t} \zeta_{s, t} \r)^{2} \r\vert^{\delta/2} \r)
\\
&\leq T^{-\delta} \l[ \E\l( \l\vert \sum_{s = 1}^{T} g_{s}^{2} \r\vert^{\delta} \r) \r]^{1/2} \l[ \E\l( \l\vert \sum_{s = 1}^{T} \l( \sum_{t = 1}^{T} g_{t}\zeta_{s, t} \r)^{2} \r\vert^{\delta} \r) \r]^{1/2} 
\\
&\le c_0 T^{-\delta/2} 
\l[ T^{\delta - 1} \sum_{s = 1}^{T} \E\l( \l\vert \sum_{t = 1}^{T} g_{t}\zeta_{s, t} \r\vert^{2\delta} \r) \r]^{1/2}
\le c_0 T^{\delta /2} N^{-\delta /2}.
\end{align*}
Putting together the bounds on $T_{2, 1}$ and $T_{2, 2}$, we have
\begin{align}
T_2 \le c_0 T^{\delta /2}N^{-\delta /2}.
\label{term-2}
\end{align}
Finally consider
\begin{align*}
T_3 \leq T^{-\delta} \E\l( \l\vert \sum_{s = 1}^{T} \l( \widehat{g}_{s} - g_{s} \r) \l( \sum_{t = 1}^{T}g_{t} \eta_{s, t} \r) \r\vert^{\delta} \r) + T^{-\delta} \E\l( \l\vert \sum_{s = 1}^{T} g_{s} \l( \sum_{t = 1}^{T} g_{t} \eta_{s, t} \r) \r\vert^{\delta} \r) =: T_{3, 1} + T_{3, 2}.
\end{align*}
It holds that
\begin{align*}
T_{3, 1} &\leq 
T^{-\delta} \E\l( \l\vert \sum_{s = 1}^{T} \l( \widehat{g}_{s} - g_{s} \r)^{2} \r\vert^{\delta /2} \l\vert \sum_{s = 1}^{T} \l( \sum_{t = 1}^{T} g_{t} \eta_{s, t} \r)^{2} \r\vert^{\delta /2} \r) 
\\
&\leq T^{-\delta} \l[ \E\l( \l\vert \sum_{s = 1}^{T}\l( \widehat{g}_{s} - g_{s} \r)^{2} \r\vert^{\delta} \r) \r]^{1/2} \l[ \E\l( \l\vert \sum_{s = 1}^{T} \l( \sum_{t = 1}^{T} g_{t} \eta_{s, t} \r)^{2} \r\vert^{\delta} \r) \r]^{1/2} 
\\
&\le c_0 T^{-\delta /2} C_{NT}^{-\delta} \l[ T^{\delta - 1} \sum_{s = 1}^{T} \E\l( \l\vert \sum_{t = 1}^{T} g_{t} \eta_{s, t} \r\vert^{2\delta} \r) \r]^{1/2}.
\end{align*}
By definition of $\eta_{s, t}$, making use of~\eqref{eq:g:delta}, 
\begin{align*}
T^{\delta - 1} \sum_{s = 1}^{T} \E\l( \l\vert \sum_{t = 1}^{T} g_{t} \eta_{s, t} \r\vert^{2\delta} \r)
&= T^{\delta - 1} \sum_{s = 1}^{T} \E\l( \l\vert
\sum_{t = 1}^{T} g_{t} g_{s} \frac{1}{N} \sum_{i = 1}^{N} \lambda_{i} e_{i, t} \r\vert^{2\delta} \r)
\\
&\leq T^{\delta - 1} \sum_{s = 1}^{T} \l[ \E\l( \l\vert g_{s} \r\vert^{4\delta} \r) \r]^{1/2} \l[ \E\l( \l\vert \frac{1}{N} \sum_{i = 1}^{N} \sum_{t = 1}^{T} g_{t} \lambda_{i} e_{i, t} \r\vert^{4\delta} \r) \r]^{1/2}
\\
&\leq c_{0} T^{2\delta} N^{-\delta}
\end{align*}
by Assumptions~\ref{factors}~\ref{assum:factors:one} and~\ref{depFE}~\ref{assum:depFE:two}.
Therefore,
\begin{align}
T_{3, 1} \leq c_0 T^{\delta/2} N^{-\delta/2} C_{NT}^{-\delta}. \nn 
\end{align}
Along the same lines, 
\begin{align*}
T_{3, 2} &\leq T^{-\delta} \E\l( \l\vert \sum_{s = 1}^{T} g_{s}^{2} \r\vert^{\delta/2} \l\vert \sum_{s = 1}^{T} \l( \sum_{t = 1}^{T} g_{t} \eta_{s, t} \r)^{2} \r\vert^{\delta /2} \r) 
\\
&\leq T^{-\delta} \l[ \E\l( \l\vert \sum_{s = 1}^{T} g_{s}^{2} \r\vert^{\delta} \r) \r]^{1/2} \l[ \E\l( \l\vert \sum_{s = 1}^{T} \l( \sum_{t = 1}^{T} g_{t} \eta_{s, t} \r)^{2} \r\vert^{\delta}\r) \r]^{1/2}
\\
&\le c_0 T^{-\delta/2} \l[ T^{\delta - 1} \sum_{s = 1}^{T} \E\l( \l\vert \sum_{t = 1}^{T} g_{t} \eta_{s, t} \r\vert^{2\delta} \r) \r]^{1/2}
\le c_0 T^{\delta/2} N^{-\delta/2},
\end{align*}
which entails that
\begin{align}
T_3 \le c_0 T^{\delta/2} N^{-\delta/2}.
\label{term-3}
\end{align}
The desired result now follows from plugging~\eqref{term-4}, \eqref{term-1}, \eqref{term-2} and \eqref{term-3} into~\eqref{dec-lemma-2}.
\end{proof}

\begin{lem}
\label{lem:gghat} Suppose that Assumptions~\ref{factors}--\ref{assum:cps} hold with $\rho = 1$ in Assumptions~\ref{factors}, \ref{idiosyncratic} and~\ref{depFE}. Then it
follows that 
\begin{equation*}
\left\Vert \left( \frac{\widehat{\mathbf{G}}^{\top }\mathbf{G}}{T}\right)
\left( \frac{\bm\Lambda^{\top }\bm\Lambda }{N}\right) \left( \frac{\mathbf{G%
}^{\top }\widehat{\mathbf{G}}}{T}\right) -\bm\Phi \right\Vert =O_{P}\left(\frac{1}{C_{NT}}\right) .
\end{equation*}
\end{lem}

\begin{proof}
By construction, we have
\begin{align}
\wh{\mbf G} \bm\Phi_{NT} &= \l( \frac{\mbf X\mbf X^\top}{NT} \r) \wh{\mbf G}, \text{ \ hence}
\label{eq:ghat:eigvec} \\
\bm\Phi_{NT} &= \frac{1}{T} \wh{\mbf G}^\top \l( \frac{\mbf X\mbf X^\top}{NT} \r) \wh{\mbf G} =  \frac{\wh{\mbf G}^\top \mbf G}{T}  \l( \frac{\bm\Lambda^\top \bm\Lambda}{N} \r) \frac{\mbf G^\top \wh{\mbf G}}{T} + \frac{1}{T} \wh{\mbf G}^\top \mbf R_{NT} \wh{\mbf G}, \text{ \ where}
\nn \\
\mbf R_{NT} &= \frac{1}{NT} \l( \mbf G \bm\Lambda^\top \mbf E^\top + \mbf E \bm\Lambda \mbf G^\top + \mbf E \mbf E^\top \r).
\nn
\end{align}
By~\eqref{eq:bai} and Lemmas~\ref{moment-1} and~\ref{moment-2},
\begin{align*}
\l\Vert \frac{1}{T} \wh{\mbf G}^\top \mbf R_{NT} \wh{\mbf G} \r\Vert 
&\le 
\l( \l\Vert \frac{1}{T} (\wh{\mbf G} - \mbf G \mbf H)^\top (\wh{\mbf G} - \mbf G \mbf H) \r\Vert + 
\l\Vert \mbf H \r\Vert \; \l\Vert \frac{1}{T} \mbf G^\top (\wh{\mbf G} - \mbf G \mbf H) \r\Vert \r) \l\Vert \bm\Phi_{NT}^{-1} \r\Vert 
\\
&= O_P\l(\frac{1}{C_{NT}^2}\r),
\end{align*}
where we also use that $\Vert \bm\Phi_{NT}^{-1} \Vert = O_P(1)$ from Lemma~\ref{lem:Phi}, and $\Vert \mbf H \Vert = O_P(1)$ from~\eqref{eq:H:bound}.
Then, the conclusion follows from Lemma~\ref{lem:Phi}.
\end{proof}

\begin{lem}
\label{lem:hi} Suppose that Assumptions~\ref{factors}--\ref{assum:cps} hold with $\rho = 1$ in Assumptions~\ref{factors}, \ref{idiosyncratic} and~\ref{depFE}.
For $\mathbf{H}_0$ is defined in~\eqref{h-tilde}, we have $\Vert \mathbf{H }- \mathbf{H}_0 \Vert = O_P(C_{NT}^{-1})$.
\end{lem}

\begin{proof}
From~\eqref{eq:ghat:eigvec}, $\wh{\mbf G}$ satisfies
\begin{align*}
\l( \frac{\bm\Lambda^\top\bm\Lambda}{N} \r)^{1/2} \frac{1}{T} \mbf G^\top \l( \frac{\mbf X\mbf X^\top}{NT} \r) \wh{\mbf G} &= \l( \frac{\bm\Lambda^\top\bm\Lambda}{N} \r)^{1/2} \l( \frac{\mbf G^\top \wh{\mbf G}}{T} \r) \bm\Phi_{NT}.
\end{align*}
Substituting $\mbf X = \mbf G \bm\Lambda^\top + \mbf E$ into the above equations, we have
\begin{align*}
\l( \frac{\bm\Lambda^\top\bm\Lambda}{N} \r)^{1/2} \l( \frac{\mbf G^\top \mbf G}{T} \r) \l( \frac{\bm\Lambda^\top\bm\Lambda}{N} \r) \l( \frac{\mbf G^\top \wh{\mbf G}}{T} \r) + c_{NT} &= \l( \frac{\bm\Lambda^\top\bm\Lambda}{N} \r)^{1/2} \l( \frac{\mbf G^\top \wh{\mbf G}}{T} \r) \bm\Phi_{NT}
\end{align*}
where, recalling the definition of $\mbf R_{NT}$ in the proof of Lemma~\ref{lem:gghat}, we have
\begin{align*}
c_{NT} &= \l( \frac{\bm\Lambda^\top\bm\Lambda}{N} \r)^{1/2} \frac{1}{T} \mbf G^\top \mbf R_{NT} \wh{\mbf G}
\text{ \ such that \ }
\\
\Vert c_{NT} \Vert &\le \l\Vert \frac{\bm\Lambda^\top\bm\Lambda}{N} \r\Vert^{1/2} \l\Vert \frac{1}{T} \mbf G^\top (\wh{\mbf G} - \mbf G \mbf H) \r\Vert = O_P\l( \frac{1}{C_{NT}^2} \r).
\end{align*}
by Assumption~\ref{loadings}~\ref{assum:loadings:two} and Lemma~\ref{moment-2}.
Recall the definitions of $\mbf B$ and $\mbf B_{NT}$ in the proof of Lemma~\ref{lem:Phi}, and let us define 
\begin{align*}
\mbf C_{NT} = \l( \frac{\bm\Lambda^\top\bm\Lambda}{N} \r)^{1/2} \l( \frac{\mbf G^\top \wh{\mbf G}}{T} \r).
\end{align*}
Then, $\Vert \mbf B_{NT} - \mbf B \Vert = O_P(C_{NT}^{-1})$ by Assumption~\ref{loadings}~\ref{assum:loadings:two} and Lemma~\ref{lem:gg}.
Also, $\mbf C_{NT}$ is $O_P(1)$ and (asymptotically) invertible from Lemma~\ref{lem:gghat}.
Denote by $\mbf W$ the $r \times r$-matrix containing the (normalised) eigenvectors of $\mbf B$ corresponding to the eigenvalues on the diagonal of $\bm\Phi$. 
Then, the remainder of the proof proceeds analogously as that of Lemma~6 of \cite{han2014} which shows that $\mbf H_0 = \text{plim}_{\min(N, T) \rightarrow \infty} \mbf H = \bm\Sigma_{\Lambda}^{1/2} \mbf W \bm\Phi^{-1/2}$ and $\Vert \mbf H - \mbf H_0 \Vert = O_P(C_{NT}^{-1})$.
\end{proof}

The next two lemmas contain two maximal inequalities which are required to
bound the difference between the partial sums of $\widehat{\mathbf{g}}_{t}%
\widehat{\mathbf{g}}_{t}^\top$ and those of $\mathbf{H}^\top \mathbf{g}_{t}%
\mathbf{g}_{t}^\top \mathbf{H}$.

\begin{lem}
\label{max-inequ} Suppose that the assumptions of Lemma~\ref{moment-1} hold.
Then, we have 
\begin{equation*}
\mathsf{E}\left( \max_{1\leq k\leq T}\left\Vert \sum_{t=1}^{k}\left( 
\widehat{\mathbf{g}}_{t}-\mathbf{H}^{\top }\mathbf{g}_{t}\right) \left( 
\widehat{\mathbf{g}}_{t}-\mathbf{H}^{\top }\mathbf{g}_{t}\right)^{\top
}\right\Vert^{\delta }\right) \leq c_{0}T^{\delta }C_{NT}^{-2\delta },
\end{equation*}%
for all $1\leq \delta \leq 2 + \epsilon$.
\end{lem}

\begin{proof}
The proof of follows immediately upon noting that
\begin{align*}
\E\l( \max_{1 \leq k \leq T} \l\Vert \sum_{t = 1}^{k}\l( \widehat{\mbf g}_{t} - \mbf H^{\top} \mbf g_{t}\r) \l( \widehat{\mbf g}_{t} - \mbf H^{\top} \mbf g_{t}\r)^\top \r\Vert^{\delta} \r)
\le 
\E\l( \l\Vert \sum_{t = 1}^{T} \l( \widehat{\mbf g}_{t} - \mbf H^{\top} \mbf g_{t}\r) \l( \widehat{\mbf g}_{t} - \mbf H^{\top} \mbf g_{t}\r)^\top \r\Vert^{\delta} \r)
\end{align*}
and using Lemma~\ref{moment-1}.
\end{proof}

\begin{lem}
\label{max-inequ-2} Suppose that the assumptions of Lemma~\ref{moment-2} hold with $\rho = 1$ in Assumptions~\ref{factors}, \ref{idiosyncratic} and~\ref{depFE}. Then,
\begin{align}
T^{-\delta}\mathsf{E}\left( \max_{1\leq k\leq T}\left\Vert
\sum_{t=1}^{k}\sum_{s=1}^{T}\widehat{\mathbf{g}}_{s}\mathbf{g}_{t}^{\top
}\gamma_{s,t}\right\Vert^{\delta}\right) & \leq c_{0},  \label{mc-1}
\\
T^{-\delta}\mathsf{E}\left( \max_{1\leq k\leq T}\left\Vert
\sum_{t=1}^{k}\sum_{s=1}^{T}\widehat{\mathbf{g}}_{s}\mathbf{g}_{t}^{\top
}\zeta_{s,t}\right\Vert^{\delta}\right) & \leq c_{0}T^{{\delta/2}%
}N^{-\delta/2},  \label{mc-2} \\
T^{-\delta}\mathsf{E}\left( \max_{1\leq k\leq T}\left\Vert
\sum_{t=1}^{k}\sum_{s=1}^{T}\widehat{\mathbf{g}}_{s}\mathbf{g}_{t}^{\top
}\eta_{s,t}\right\Vert^{\delta}\right) & \leq c_{0}T^{\delta
_{2}/2}N^{-\delta/2},  \label{mc-3} \\
T^{-\delta}\mathsf{E}\left( \max_{1\leq k\leq T}\left\Vert
\sum_{t=1}^{k}\sum_{s=1}^{T}\widehat{\mathbf{g}}_{s}\mathbf{g}_{t}^{\top
}\xi_{s,t}\right\Vert^{\delta}\right) & \leq c_{0}T^{\delta
_{3}/2}N^{-\delta/2}+c_{1}T^{\delta}N^{-\delta/2}{%
C_{NT}^{-\delta}},  \label{mc-4}
\end{align}
for $ 1\leq \delta \leq 1+\epsilon $.
\end{lem}

\begin{proof}
The proof is based on very similar passages as the proof of Lemma~\ref{moment-2}, which we omit when possible. 
As before, we start with~\eqref{mc-4}.
Note that
\begin{align*}
T^{-\delta} \E\left( \max_{1 \leq k \leq T} \left\Vert
\sum_{t = 1}^{k} \sum_{s = 1}^{T} \widehat{\mathbf{g}}_{s} \mathbf{g}_{t}^{\top} \xi_{s, t} \right\Vert^{\delta} \right)
\le T^{-\delta} \E\left( \max_{1 \leq k \leq T} \left\Vert \sum_{t = 1}^{k} \sum_{s = 1}^{T} (\widehat{\mathbf{g}}_{s} - \mbf H^\top \mbf g_s) \mathbf{g}_{t}^{\top} \xi_{s, t} \right\Vert^{\delta} \right)
\\
+ T^{-\delta} \E\left( \max_{1 \leq k \leq T} \left\Vert \mbf H^\top \sum_{t = 1}^{k} \sum_{s = 1}^{T} \mbf g_s \mathbf{g}_{t}^{\top} \xi_{s, t} \right\Vert^{\delta} \right) =: T_{4, 1} + T_{4, 2}.
\end{align*}
Repeating the same passages as in the proof of Lemma~\ref{moment-2} leading to~\eqref{b-4-2}, 
it is easily seen that 
\begin{align*}
T_{4, 2} &\le T^{-\delta} \l( \E( \Vert \mbf H^\top \Vert^{2\delta} \r)^{1/2} \l[ \E\left( \max_{1 \leq k \leq T} \left\Vert \sum_{t = 1}^{k} \sum_{s = 1}^{T} \mbf g_s \mathbf{g}_{t}^{\top} \xi_{s, t} \right\Vert^{2\delta} \right) \r]^{1/2}
\\
&\le c_0 T^{-\delta} \l\{ 
\l[ \E\l( \max_{1 \le k \le T} \l\Vert \sum_{t = 1}^k \mbf g_t \mbf g_t^\top \r\Vert^{4\delta} \r) \r]^{1/2} 
\l[ \E\l( \l\Vert \frac{1}{N} \sum_{s = 1}^T \sum_{i = 1}^N \bm\lambda_i^\top \mbf g_s e_{i, s} \r\Vert^{4\delta} \r) \r]^{1/2}
\r\}^{1/2}
\\
&\le c_0 T^{-\delta} \l\{ 
\l[ \E\l( \l\Vert \sum_{t = 1}^T \mbf g_t \mbf g_t^\top \r\Vert^{4\delta} \r) \r]^{1/2} 
\l[ \E\l( \l\Vert \frac{1}{N} \sum_{s = 1}^T \sum_{i = 1}^N \bm\lambda_i^\top \mbf g_s e_{i, s} \r\Vert^{4\delta} \r) \r]^{1/2}
\r\}^{1/2}
\\
&\le c_0 T^{\delta/2} N^{-\delta/2}
\end{align*}
As for $T_{4, 1}$, we may set $r = d = 1$ and omit $\mbf H$, which gives
\begin{align*}
T_{4, 1} &=
T^{-\delta} \E\l( \max_{1 \le k \le T} \l\vert \sum_{t = 1}^{k} g_{t}^{2} \cdot \sum_{s = 1}^{T} \l( \widehat{g}_{s} - g_{s} \r) \frac{1}{N} \sum_{i = 1}^{N} \lambda_{i} e_{i, s} \r\vert^{\delta} \r)
\\
&\leq 
T^{-\delta} \l[ \E\l( \l\vert \sum_{t = 1}^{T} g_{t}^{2} \r\vert^{p\delta} \r) \r]^{\frac{1}{p}} \l[ \E\l( \l\vert \sum_{s = 1}^{T} \l( \widehat{g}_{s} - g_{s} \r) \frac{1}{N} \sum_{i = 1}^{N} \lambda_{i} e_{i, s}\r\vert^{\frac{p\delta}{p - 1}}\r) \r]^{\frac{p - 1}{p}}.
\end{align*}
Setting $p = 4$ and applying the arguments analogous to those adopted in \eqref{b-4-1}, we obtain $T_{4, 1} \le c_0 T^{\delta} N^{-\delta/2} C_{NT}^{-\delta}$  which completes the proof of~\eqref{mc-4}.

For the rest of the proof, we proceed analogously and check the steps for the case of $r = d = 1$ for simplicity. 
For~\eqref{mc-1}, we have
\begin{align*}
& T^{-\delta} \E\l( \max_{1 \leq k \leq T} \l\vert \sum_{t = 1}^{k} \sum_{s = 1}^{T} \widehat{g}_{s} g_{t} \gamma_{s, t} \r\vert^{\delta} \r) 
\leq T^{-\delta} \E\l( \l\vert \max_{1 \leq k \leq T} \sum_{t = 1}^k \sum_{s = 1}^{T} \l( \widehat{g}_{s} - g_{s}\r) g_{t} \gamma_{s, t} \r\vert^{\delta} \r) 
\\
&+ T^{-\delta} \E\l( \l\vert \max_{1 \leq k \leq T} \sum_{t = 1}^k \sum_{s = 1}^{T} g_{s} g_{t} \gamma_{s, t} \r\vert^{\delta} \r) =: T_{1, 1} + T_{1, 2}.
\end{align*}

Using the Cauchy-Schwartz inequality twice,
\begin{align*}
T_{1, 1} &\le T^{-\delta} \E\l( \l\vert \sum_{s = 1}^{T}\l( \widehat{g}_{s} - g_{s} \r)^{2} \r\vert^{\delta/2} \max_{1 \le k \le T} \l\vert \sum_{s = 1}^{T} \l( \sum_{t = 1}^{k} g_{t} \gamma_{s, t} \r)^{2} \r\vert^{\delta /2} \r)
\\
&\le c_0 T^{-\delta /2} C_{NT}^{-\delta} \l[ T^{\delta - 1} \sum_{s = 1}^{T} \E \l( \max_{1 \le k \le T} \l\vert \sum_{t = 1}^{k} g_{t} \gamma_{s, t} \r\vert^{2\delta} \r) \r]^{1/2},
\end{align*}
having used Lemma~\ref{moment-1} in the last passage. 
Let $n_{0} = \lceil 4\delta \rceil$.
Using the $\mathcal{L}_{p}$-norm inequality, we have
\begin{align*}
\E\l( \max_{1 \le k \le T} \l\vert \sum_{t = 1}^{k} g_{t} \gamma_{s, t}\r\vert^{n_{0}} \r) &=
\E\l( \max_{1 \le k \le T} \l\vert \sum_{t_{1}, \ldots, t_{n_{0}} = 1}^{k} g_{t_{1}} \cdot \ldots \cdot g_{t_{n_{0}}} \cdot \gamma_{s, t_{1}} \cdot \ldots \cdot \gamma_{s, t_{n_{0}}} \r\vert \r)
\\
&\leq \sum_{t_{1}, \ldots, t_{n_{0}} = 1}^{T} 
\E\l( \prod\limits_{i=1}^{n_{0}} \l\vert g_{t_{i}}\r\vert \r) \prod\limits_{i = 1}^{n_{0}}\l\vert \gamma_{s, t_{i}} \r\vert 
\leq c_{0},
\end{align*}
and therefore $T_{1, 1} \le c_0 C_{NT}^{-\delta}$.
Analgously, combined with~\eqref{eq:g:delta},
\begin{align*}
T_{1, 2} & \leq T^{-\delta} \E \l( \l\vert \sum_{s = 1}^{T} g_{s}^{2} \r\vert^{\delta /2} \max_{1 \le k \le T} \l\vert \sum_{s = 1}^{T} \l( \sum_{t = 1}^{k} g_{t}\gamma_{s, t} \r)^{2} \r\vert^{\delta /2} \r)
\leq c_{0} T^{-\delta} \cdot T^{\delta / 2} \cdot T^{\delta / 2} = O(1),
\end{align*}
so we have the desired result. 
Next, let us consider~\eqref{mc-2}.
\begin{align*}
& T^{-\delta} \E\l( \max_{1 \leq k \leq T} \l\vert \sum_{t = 1}^{k} \sum_{s = 1}^{T} \wh g_s g_t \zeta_{s, t} \r\vert^\delta \r)
\leq T^{-\delta} \E\l( \max_{1 \leq k \leq T} \l\vert \sum_{s = 1}^{T} \l( \widehat{g}_{s} - g_{s}\r) \l( \sum_{t = 1}^{k} g_{t} \zeta_{s, t} \r) \r\vert^{\delta} \r) 
\\
& + T^{-\delta} \E\l( \max_{1 \leq k \leq T} \l\vert \sum_{s = 1}^{T} g_{s} \l( \sum_{t = 1}^{k} g_{t} \zeta_{s, t} \r) \r\vert^{\delta} \r) =: T_{2, 1} + T_{2, 2}.
\end{align*}
Using the Cauchy-Schwartz inequality and Lemma~\ref{moment-1},
\begin{align*}
T_{2, 1} &\leq T^{-\delta} \E\l( \l\vert \sum_{s = 1}^{T} \l( \widehat{g}_{s} - g_{s} \r)^{2}\r\vert^{\delta / 2} \max_{1 \leq k \leq T} \l\vert \sum_{s = 1}^{T} \l( \sum_{t = 1}^{k} g_{t} \zeta_{s, t} \r)^{2} \r\vert
^{\delta/2} \r)
\\
&\leq T^{-\delta} \l[ \E\l( \l\vert \sum_{s = 1}^{T}\l( \widehat{g}_{s} - g_{s} \r)^{2} \r\vert^{\delta} \r) \r]^{1/2} \l[ \E\l( \max_{1 \leq k \leq T} \l\vert \sum_{s = 1}^{T} \l( \sum_{t = 1}^{k} g_{t}\zeta_{s, t} \r)^{2} \r\vert^{\delta} \r) \r]^{1/2} 
\\
&\le c_0 T^{-\delta/2} C_{NT}^{-\delta} \l[ T^{\delta - 1} \sum_{s = 1}^{T} \E\l( \max_{1 \leq k \leq T} \l\vert \sum_{t = 1}^{k} g_{t}\zeta_{s, t} \r\vert^{2\delta} \r) \r]^{1/2}.
\end{align*}
Under Assumption~\ref{depFE}~\ref{assum:depFE:one}, it follows that 
\begin{align*}
\E\l( \max_{1 \leq k \leq T} \l\vert \sum_{t = 1}^{k} g_{t}\zeta_{s, t} \r\vert^{2\delta} \r) \le c_0 T^\delta N^{-\delta}
\end{align*}
by Theorem~3.1 in \citet{moricz1982}, which
immediately entails that
\begin{align}
T_{2, 1} \le c_0 T^{\delta/2} N^{-\delta/2} C_{NT}^{-\delta}.
\end{align}
By the same token, with~\eqref{eq:g:delta},
\begin{align*}
T_{2, 2} 
&\leq T^{-\delta} \E\l( \l\vert \sum_{s = 1}^{T} g_{s}^{2} \r\vert^{\delta / 2} \max_{1 \leq k \leq T} \l\vert \sum_{s = 1}^{T} \l( \sum_{t = 1}^{k} g_{t} \zeta_{s, t} \r)^{2} \r\vert^{\delta/2} \r)
\\
&\le c_0 T^{-\delta/2} 
\l[ T^{\delta - 1} \sum_{s = 1}^{T} \E\l( \max_{1 \leq k \leq T} \l\vert \sum_{t = 1}^{k} g_{t} \zeta_{s, t} \r\vert^{2\delta} \r) \r]^{1/2}
\le c_0 T^{\delta /2} N^{-\delta /2}.
\end{align*}
Putting together the bounds on $T_{2, 1}$ and $T_{2, 2}$, we have the desired result. 
Finally, as for~\eqref{mc-3}, 
\begin{align*}
& T^{-\delta} \E\l( \max_{1 \leq k \leq T} \l\vert \sum_{s = 1}^{T} \sum_{t = 1}^{k} \wh g_s g_t \eta_{s, t} \r\vert^{\delta} \r) 
\le 
T^{-\delta} \E\l( \l\vert \sum_{s = 1}^{T} \l( \widehat{g}_{s} - g_{s} \r) \max_{1 \le k \le T} \l( \sum_{t = 1}^{k} g_{t} \eta_{s, t} \r) \r\vert^{\delta} \r) 
\\
& + T^{-\delta} \E\l( \l\vert \sum_{s = 1}^{T} g_{s} \max_{1 \le k \le T} \l( \sum_{t = 1}^{k} g_{t} \eta_{s, t} \r) \r\vert^{\delta} \r) =: T_{3, 1} + T_{3, 2}.
\end{align*}
It holds that
\begin{align*}
T_{3, 1} 
&\leq T^{-\delta} \E\l( \l\vert \sum_{s = 1}^{T} \l( \widehat{g}_{s} - g_{s} \r)^{2} \r\vert^{\delta /2} \max_{1 \le k \le T} \l\vert \sum_{s = 1}^{T} \l( \sum_{t = 1}^{k} g_{t} \eta_{s, t} \r)^{2} \r\vert^{\delta /2} \r) 
\\
&\leq T^{-\delta} \l[ \E\l( \l\vert \sum_{s = 1}^{T}\l( \widehat{g}_{s} - g_{s} \r)^{2} \r\vert^{\delta} \r) \r]^{1/2} \l[ \E\l( \max_{1 \le k \le T} \l\vert \sum_{s = 1}^{T} \l( \sum_{t = 1}^{k} g_{t} \eta_{s, t} \r)^{2} \r\vert^{\delta} \r) \r]^{1/2} 
\\
&\le c_0 T^{-\delta /2} C_{NT}^{-\delta} \l[ T^{\delta - 1} \sum_{s = 1}^{T} \E\l( \max_{1 \le k \le T} \l\vert \sum_{t = 1}^{k} g_{t} \eta_{s, t} \r\vert^{2\delta} \r) \r]^{1/2}.
\end{align*}
By definition of $\eta_{s, t}$, making use of~\eqref{eq:g:delta}, 
\begin{align*}
& T^{\delta - 1} \sum_{s = 1}^{T} \E\l( \max_{1 \le k \le T} \l\vert \sum_{t = 1}^{k} g_{t} \eta_{s, t} \r\vert^{2\delta} \r)
= T^{\delta - 1} \sum_{s = 1}^{T} \E\l( \max_{1 \le k \le T} \l\vert \sum_{t = 1}^{k} g_{t} g_{s} \frac{1}{N} \sum_{i = 1}^{N} \lambda_{i} e_{i, t} \r\vert^{2\delta} \r)
\\
&\leq T^{\delta - 1} \sum_{s = 1}^{T} \l[ \E\l( \l\vert g_{s} \r\vert^{4\delta} \r) \r]^{1/2} \l[ \E\l( \max_{1 \le k \le T} \l\vert \frac{1}{N} \sum_{i = 1}^{N} \sum_{t = 1}^{k} g_{t} \lambda_{i} e_{i, t} \r\vert^{4\delta} \r) \r]^{1/2}
\leq c_{0} T^{2\delta} N^{-\delta}
\end{align*}
by Assumptions~\ref{factors}~\ref{assum:factors:one} and~\ref{depFE}~\ref{assum:depFE:two}, and Theorem~3.1 of \cite{moricz1982}.
Therefore,
\begin{align}
T_{3, 1} \leq c_0 T^{\delta/2} N^{-\delta/2} C_{NT}^{-\delta}. \nn 
\end{align}
Along the same lines, 
\begin{align*}
T_{3, 2} &\leq T^{-\delta} \E\l( \l\vert \sum_{s = 1}^{T} g_{s}^{2} \r\vert^{\delta/2} \max_{1 \le k \le T} \l\vert \sum_{s = 1}^{T} \l( \sum_{t = 1}^{k} g_{t} \eta_{s, t} \r)^{2} \r\vert^{\delta /2} \r) 
\\
&\leq T^{-\delta} \l[ \E\l( \l\vert \sum_{s = 1}^{T} g_{s}^{2} \r\vert^{\delta} \r) \r]^{1/2} \l[ \E\l( \max_{1 \le k \le T} \l\vert \sum_{s = 1}^{T} \l( \sum_{t = 1}^{k} g_{t} \eta_{s, t} \r)^{2} \r\vert^{\delta}\r) \r]^{1/2}
\\
&\le c_0 T^{-\delta/2} \l[ T^{\delta - 1} \sum_{s = 1}^{T} \E\l( \max_{1 \le k \le T} \l\vert \sum_{t = 1}^{k} g_{t} \eta_{s, t} \r\vert^{2\delta} \r) \r]^{1/2}
\le c_0 T^{\delta/2} N^{-\delta/2},
\end{align*}
which completes the proof.
\end{proof}

The following three lemmas provide an estimate of the rate of convergence of
the partial sums of (appropriately centered) $\vech(\widehat{\mathbf{g}}_{t}\widehat{\mathbf{g}}_{t}^\top)$ to their weak limit.

\begin{lem}
\label{sip-1} Suppose that Assumption~\ref{factors} holds.
Then for each $j\in \{0,\ldots ,R\}$, on a suitably enlarged probability
space, there exist some constant $\zeta_{1}\in (0,1/2)$ and two independent 
$d$-dimensional Wiener processes $\{W_{1,dT}^{(j)}(k),\,1\leq k\leq \Delta
_{j}/2\}$ and $\{W_{2,dT}^{(j)}(k),\,1\leq k\leq \Delta_{j}/2\}$ with $\Delta_{j}=k_{j+1}-k_{j}$, such that 
\begin{align*}
& \max_{1\leq k\leq \Delta_{j}/2}\frac{1}{k^{\zeta_{1}}}\left\Vert
\sum_{t=k_{j}+1}^{k_{j}+k}\mathsf{Vech}\left( \mathbf{g}_{t}\mathbf{g}%
_{t}^{\top }-\mathsf{E}\left( \mathbf{g}_{t}\mathbf{g}_{t}^{\top }\right)
\right) -\mathbf{D}_{j}^{1/2}W_{1,dT}^{(j)}(k)\right\Vert =O_{P}(1), \\
& \max_{\Delta_{j}/2<k<\Delta_{j}}\frac{1}{(\Delta_{j}-k)^{\zeta_{1}}}%
\left\Vert \sum_{t=k_{j}+k+1}^{k_{j+1}}\mathsf{Vech}\left( \mathbf{g}_{t}%
\mathbf{g}_{t}^{\top }-\mathsf{E}\left( \mathbf{g}_{t}\mathbf{g}_{t}^{\top
}\right) \right) -\mathbf{D}_{j}^{1/2}W_{2,dT}^{(j)}\left( \Delta
_{j}-k\right) \right\Vert =O_{P}(1)
\end{align*}
where, with $\mathbf{A}_{j}$ defined in~\eqref{eq:model:two} and $\mathbf{D}$
in~\eqref{eq:D}, we have 
\begin{equation}
\mathbf{D}_{j}=\mathbf{L}_{r}(\mathbf{A}_{j}\otimes \mathbf{A}_{j})\mathbf{K}%
_{r}\mathbf{D}\mathbf{K}_{r}^{\top }(\mathbf{A}_{j}^{\top }\otimes \mathbf{A}%
_{j}^{\top })\mathbf{L}_{r}^{\top }. \label{eq:dj}
\end{equation}%
\end{lem}

\begin{proof}
We first show that on a suitably enlarged probability space, there exist some constant $\zeta_{1} \in (0, 1/2)$ and two independent $d$-dimensional Wiener processes $\{ W_{1, dT}(k), \, 1 \le k \le T/2 \}$ and $\{ W_{2, dT}(k), \, 1 \le k \le T/2 \}$ such that
\begin{align*}
& \max_{1 \leq k \leq T/2} \frac{1}{k^{\zeta_{1}}} \l\Vert \sum_{t = 1}^{k} \vech\l( \mbf f_{t}\mbf f_{t}^\top - \E\l( \mbf f_{t} \mbf f_{t}^\top \r) \r) - \mbf D^{1/2} W_{1, dT}(k) \r\Vert = O_P(1),
\\
& \max_{T/2 < k < T} \frac{1}{(T - k)^{\zeta_{1}}} \l\Vert \sum_{t = k + 1}^{T} \vech\l( \mbf f_{t}\mbf f_{t}^\top - \E\l( \mbf f_{t} \mbf f_{t}^\top \r) \r) - \mbf D^{1/2} W_{2, dT}(T - k) \r\Vert = O_P(1).
\end{align*}
We begin by noting that Assumption~\ref{factors}~\ref{assum:factors:one} entails that $\{ \mbf f_t\mbf f_t^\top - \bm\Sigma_F \}$ is an $\mc L_{\phi}$-decomposable Bernoulli shift with some $\phi > 2$, see the proof of Lemma~\ref{lem:gg}. 
Then, the desired result follows immediately from Theorem~S2.1 of \citet{aue2014}; note that the proofs in \citet{aue2014} are based on the
blocking argument, and therefore this leads to the independence between $\{ W_{1, dT}(k), \, 1 \le k \le T/2 \}$ and $\{ W_{2, dT}(k), \, 1 \le k \le T/2\}$.
The claim of the lemma follows from this, by noting that there are finitely many change points and also from~\eqref{eq:model:two}, we have $\vech(\mbf g_t \mbf g_t^\top) = \mbf L_r (\mbf A_j \otimes \mbf A_j) \mbf K_r \vech(\mbf f_t \mbf f_t^\top)$.
\end{proof}

\begin{lem}
\label{sip-2} Suppose that the assumptions of Lemmas~\ref{moment-1} and \ref{moment-2} hold with with $\rho = 1$ in Assumptions~\ref{factors}, \ref{idiosyncratic} and~\ref{depFE}, as well as Assumption~\ref{assum:nt}. Then there exists
some constant $\zeta_{2}\in (0,1/2)$ such that 
\begin{align}
\max_{1\leq k\leq T}\frac{1}{k^{\zeta_{2}}}\left\Vert \sum_{t=1}^{k}\left( 
\widehat{\mathbf{g}}_{t}-\mathbf{H}^{\top }\mathbf{g}_{t}\right) \left( 
\widehat{\mathbf{g}}_{t}-\mathbf{H}^{\top }\mathbf{g}_{t}\right)^{\top
}\right\Vert & =O_{P}(1),  \label{sip-2-1} \\
\max_{1\leq k\leq T}\frac{1}{k^{\zeta_{2}}}\left\Vert \sum_{t=1}^{k}\mathbf{%
g}_{t}\left( \widehat{\mathbf{g}}_{t}-\mathbf{H}^{\top }\mathbf{g}%
_{t}\right)^{\top}\right\Vert & =O_{P}(1).  \label{sip-2-2}
\end{align}
\end{lem}

\begin{proof}
We begin with~\eqref{sip-2-1}. 
Standard arguments entail that
\begin{align}
& \p\l( \max_{1 \leq k \leq T} \frac{1}{k^{\zeta_{2}}} \l\Vert \sum_{t = 1}^{k} \l( \widehat{\mbf g}_{t} - \mbf H^\top \mbf g_{t} \r) \l( \widehat{\mbf g}_{t} - \mbf H^\top \mbf g_{t} \r)^\top \r\Vert > x \r)  
\nn \\
\leq & \, \p\l( \max_{0 \leq \ell \leq \lfloor \log(T) \rfloor} \max_{\exp(\ell) \le k \le \exp(\ell + 1)} \frac{1}{k^{\zeta_{2}}} \l\Vert \sum_{t = 1}^{k} \l( \widehat{\mbf g}_{t} - \mbf H^\top \mbf g_{t} \r) \l( \widehat{\mbf g}_{t} - \mbf H^\top \mbf g_{t} \r)^\top \r\Vert > x \r)  
\nn \\
\leq & \, \sum_{\ell = 0}^{\lfloor \log(T) \rfloor} \p\l( \max_{\exp(\ell) \le k \le \exp(\ell + 1)} \frac{1}{k^{\zeta_{2}}} \l\Vert \sum_{t = 1}^{k} \l( \widehat{\mbf g}_{t} - \mbf H^\top \mbf g_{t} \r) \l( \widehat{\mbf g}_{t} - \mbf H^\top \mbf g_{t} \r)^\top \r\Vert > x \r) 
\nn \\
\leq & \, \sum_{\ell = 0}^{\lfloor \log(T) \rfloor} \p\l( \max_{\exp(\ell) \le k \le \exp(\ell + 1)} \l\Vert \sum_{t = 1}^{k} \l( \widehat{\mbf g}_{t} - \mbf H^\top \mbf g_{t} \r) \l( \widehat{\mbf g}_{t} - \mbf H^\top \mbf g_{t} \r)^\top \r\Vert > x \exp(\zeta_2 \ell) \r) 
\nn \\
\leq & \, \frac{1}{x} \sum_{\ell = 0}^{\lfloor \log(T) \rfloor} \exp(-\zeta_2 \ell) \E\l[ \max_{\exp(\ell) \le k \le \exp(\ell + 1)} \l\Vert \sum_{t = 1}^{k} \l( \widehat{\mbf g}_{t} - \mbf H^\top \mbf g_{t} \r) \l( \widehat{\mbf g}_{t} - \mbf H^\top \mbf g_{t} \r)^\top \r\Vert \r]
\nn \\
\leq & \, \frac{1}{x} \sum_{\ell = 0}^{\lfloor \log(T) \rfloor} \exp(-\zeta_2 \ell) \frac{\exp(\ell + 1)}{\min\{N, \exp(\ell + 1)\}},
\label{eq:lem:sip-2:one}
\end{align}
where the last passage follows from Lemma~\ref{max-inequ}.
If $N \ge T$, the conclusion follows trivially.
On the other hand, if $\min\{N, \exp(\ell + 1)\} = N$ for some $\ell$, we have $N = T^\beta$ with some $\beta \in (1/2 + \epsilon_\circ, 1)$ under Assumption~\ref{assum:nt}.
Then, the RHS of~\eqref{eq:lem:sip-2:one} is bounded by
\begin{align*}
\frac{1}{x} \sum_{\ell = 0}^{\lfloor \log(T) \rfloor} \exp(-\zeta_2 \ell) + \frac{\exp(-\beta\log(T))}{x} \sum_{\ell = 0}^{\lfloor \log(T) \rfloor} \exp((1 -\zeta_2) \ell + 1) \le \frac{c_0}{x},
\end{align*}
provided that $1 - \beta \le \zeta_2$, which follows for $\zeta_2 = 1/2 - \epsilon$ with some $\epsilon \in (0, \epsilon_\circ)$.
This proves the desired result. 
The proof of~\eqref{sip-2-2} takes analogous steps and we discuss it only briefly. 
Note that, setting $r = d = 1$ and omitting $\mbf H$ for simplicity,
\begin{align}
& \max_{1 \leq k \leq T} \frac{1}{k^{\zeta_{2}}} \l\vert
\sum_{t = 1}^{k} g_{t} \l( \widehat{g}_{t} - g_{t} \r) \r\vert 
\nn \\
\leq &\, \max_{1 \leq k \leq T} \frac{1}{k^{\zeta_{2}}} \l\vert \frac{1}{T} \sum_{t = 1}^{k} \sum_{s = 1}^{T} \widehat{g}_{s} g_{t} \gamma_{s, t} \r\vert + 
\max_{1 \leq k \leq T} \frac{1}{k^{\zeta_{2}}} \l\vert  \frac{1}{T} \sum_{t = 1}^{k} \sum_{s = 1}^{T} \widehat{g}_{s} g_{t} \zeta_{s, t} \r\vert 
\nn \\
\leq &\, + \max_{1 \leq k \leq T} \frac{1}{k^{\zeta_{2}}} \l\vert \frac{1}{T} \sum_{t = 1}^{k} \sum_{s = 1}^{T} \widehat{g}_{s} g_{t} \eta_{s, t} \r\vert +
\max_{1 \leq k \leq T} \frac{1}{k^{\zeta_{2}}} \l\vert \frac{1}{T} \sum_{t = 1}^{k} \sum_{s = 1}^{T} \widehat{g}_{s} g_{t} \xi_{s, t} \r\vert
\label{m-ineq-2} 
\end{align}
by applying~\eqref{eq:bai} as in~\eqref{bai03-1}.
Then, the proof proceeds as in the proof of~\eqref{sip-2-1} to each term in the RHS of~\eqref{m-ineq-2} using Lemma~\ref{max-inequ-2}.
\end{proof}

\begin{lem}
\label{bound} Suppose that Assumptions~\ref{factors}--\ref{assum:nt} hold with $\rho = 1$ in Assumptions~\ref{factors}, \ref{idiosyncratic} and~\ref{depFE}. Then on
a suitably enlarged probability space, there exists some constant $\zeta \in
(0,1/2)$ 
such that 
\begin{align}
& \max_{1\leq k\leq T/2}\frac{1}{k^{\zeta }}\left\Vert \sum_{t=1}^{k}\mathsf{%
Vech}\left( \widehat{\mathbf{g}}_{t}\widehat{\mathbf{g}}_{t}^{\top }-\mathbf{%
H}^{\top }\mathsf{E}\left( \mathbf{g}_{t}\mathbf{g}_{t}^{\top }\right) 
\mathbf{H}\right) \right. -  \notag \\
& \qquad \qquad \qquad \left. \sum_{j=0}^{R}\mathbb{I}_{\{k_{j}\leq
k\}}\cdot \mathbf{L}_{r}(\mathbf{H}^{\top }\otimes \mathbf{H}^{\top })%
\mathbf{K}_{r}\mathbf{D}_{j}^{1/2}W_{1,dT}^{(j)}\left( \min
(k,k_{j+1})-k_{j}\right) \right\Vert =O_{P}(1),  \label{bound-1} \\
& \max_{T/2<k<T}\frac{1}{(T-k)^{\zeta }}\left\Vert \sum_{t=1}^{k}\mathsf{Vech%
}\left( \widehat{\mathbf{g}}_{t}\widehat{\mathbf{g}}_{t}^{\top }-\mathbf{H}%
^{\top }\mathsf{E}\left( \mathbf{g}_{t}\mathbf{g}_{t}^{\top }\right) \mathbf{%
H}\right) \right. -  \notag \\
& \qquad \qquad \left. \sum_{j=0}^{R}\mathbb{I}_{\{k_{j+1}\geq k\}}\cdot 
\mathbf{L}_{r}(\mathbf{H}^{\top }\otimes \mathbf{H}^{\top })\mathbf{K}_{r}%
\mathbf{D}_{j}^{1/2}W_{2,dT}^{(j)}\left( k_{j+1}-\max (k,k_{j})\right)
\right\Vert =O_{P}(1),  \label{bound-2}
\end{align}
where 
$\mathbf{D}_{j}$ is defined in~\eqref{eq:dj} and $W_{\ell ,dT}^{(j)}(\cdot
),\,\ell =1,2$, in Lemma~\ref{sip-1}.
\end{lem}

\begin{proof}
The proof follows immediately from Lemmas~\ref{sip-1} and~\ref{sip-2}.
We prove~\eqref{bound-1} only since the arguments for~\eqref{bound-2} are analogous.
Let $\zeta = \max(\zeta_1, \zeta_2)$ where $\zeta_1$ and $\zeta_2$ are defined in Lemmas~\ref{sip-1} and~\ref{sip-2}, respectively.
Then we have
\begin{align*}
& \max_{1 \leq k \leq T/2} \frac{1}{k^{\zeta}} \l\Vert \sum_{t = 1}^{k} \vech\l( \widehat{\mbf g}_{t} \widehat{\mbf g}_{t}^\top - \mbf H^{\top}\E\l(  \mbf g_{t} \mbf g_{t}^\top \r) \mbf H \r) - \r.
\\
& \qquad \qquad \qquad \qquad \l. \sum_{j = 0}^R \mathbb{I}_{\{k_j \le k\}} \cdot \mbf L_r (\mbf H^\top \otimes \mbf H^\top) \mbf K_r \mbf D_j^{1/2} W^{(j)}_{1, dT}\l( \min(k, k_{j + 1}) - k_j \r) \r\Vert 
\\
\le & \, \max_{1 \leq k \leq T/2} \frac{1}{k^{\zeta}} \l\Vert \sum_{t = 1}^{k} \vech\l(  \widehat{\mbf g}_{t} \widehat{\mbf g}_{t}^\top - \mbf H^{\top} \mbf g_{t} \mbf g_{t}^\top \mbf H \r)  \r\Vert 
\\
& \, + \max_{1 \leq k \leq T/2} \frac{1}{k^{\zeta}} \l\Vert \mbf L_r (\mbf H^\top \otimes \mbf H^\top) \mbf K_r \l( \sum_{t = 1}^{k} \vech\l( \mbf g_{t} \mbf g_{t}^\top - \E\l( \mbf g_{t} \mbf g_{t}^\top \r) \r) -  \r. \r. 
\\
& \qquad \qquad \qquad \qquad \qquad \qquad \qquad \l. \l. \sum_{j = 0}^R \mathbb{I}_{\{k_j \le k\}} \cdot \mbf D_j^{1/2} W^{(j)}_{1, dT}\l( \min(k, k_{j + 1}) - k_j \r) \r) \r\Vert
=: T_1 + T_2. 
\end{align*}
Lemma~\ref{sip-2} immediately yields that $T_1 = O_P(1)$. 
We have $T_2$ bounded in light of Lemma~\ref{sip-1} since $\Vert \mbf H \Vert = O_P(1)$ (see~\eqref{eq:H:bound}) and $\Vert \mbf L_r \Vert = O(1)$ and $\Vert \mbf K_r \Vert = O(1)$ by their construction.
\end{proof}

The following two lemmas are useful in studying the behaviour of MOSUM
statistics in~\eqref{mos} in the presence of multiple change points.

\begin{lem}
\label{mosum-alt-1} Suppose that Assumptions~\ref{factors}--\ref{assum:nt}
hold with $\rho = 1$ in Assumptions~\ref{factors}, \ref{idiosyncratic} and~\ref{depFE}. Then it holds that 
\begin{equation*}
\max_{0\leq k\leq T-\gamma }\frac{1}{\sqrt{\gamma }}\left\Vert
\sum_{t=k+1}^{k+\gamma }\mathsf{Vech}\left( \widehat{\mathbf{g}}_{t}\widehat{\mathbf{g}}_{t}^{\top }-\mathbf{H}^{\top }\mathsf{E}\left( \mathbf{g}_{t}%
\mathbf{g}_{t}^{\top }\right) \mathbf{H}\right) \right\Vert =O_{P}(\sqrt{%
\log (T/\gamma )}).
\end{equation*}
\end{lem}

\begin{proof}
By Lemma~\ref{bound}, on a suitably enlarged probability space, it holds that
\begin{align*}
& \max_{0 \leq k \leq T - \gamma} \frac{1}{\sqrt{\gamma}} \l\Vert \sum_{t = k + 1}^{k + \gamma} \vech \l( \widehat{\mbf g}_{t} \widehat{\mbf g}_{t}^\top - \mbf H^{\top} \E\l( \mbf g_{t} \mbf g_{t}^\top \r) \mbf H \r) \r\Vert
\\
\le & \frac{1}{\sqrt{\gamma}} \max_{0 \le k \le T - \gamma} \l\Vert \sum_{t = 1}^k \vech \l( \widehat{\mbf g}_{t} \widehat{\mbf g}_{t}^\top - \mbf H^{\top} \E\l( \mbf g_{t} \mbf g_{t}^\top \r) \mbf H \r) - \sum_{j = 0}^R \mathbb{I}_{\{k_j \le k\}} \cdot \mbf V_j^{1/2} W^{(j)}_{1, dT}\l( \min(k, k_{j + 1}) - k_j \r) \r\Vert
\\
&+ \frac{1}{\sqrt{\gamma}} \max_{0 \le k \le T - \gamma} \l\Vert \sum_{t = 1}^{k + \gamma} \vech \l( \widehat{\mbf g}_{t} \widehat{\mbf g}_{t}^\top - \mbf H^{\top} \E\l( \mbf g_{t} \mbf g_{t}^\top \r) \mbf H \r) - \r. 
\\
& \qquad \qquad \qquad \l. \sum_{j = 0}^R \mathbb{I}_{\{k_j \le k + \gamma \}} \cdot \mbf V_j^{1/2} W^{(j)}_{1, dT}\l( \min( k + \gamma, k_{j + 1} ) - k_j \r) \r\Vert
\\
& + \max_{0 \le k \le T - \gamma} \frac{1}{\sqrt{\gamma}} \l\Vert \sum_{j = 0}^R \mathbb{I}_{\{k < k_j \le k + \gamma \}} \cdot \mbf V_j^{1/2} W^{(j)}_{1, dT}\l( \min( k + \gamma, k_{j + 1} ) - \max(k, k_j) \r) \r\Vert 
\\
=:& \, T_1 + T_2 + T_3.
\end{align*}
Using Lemma~\ref{bound},
\begin{align*}
T_1 \leq & \, \frac{1}{\sqrt{\gamma}} \max_{0 \le k \le T - \gamma} k^{\zeta} \cdot k^{-\zeta} \l\Vert \sum_{t = 1}^k \vech \l( \widehat{\mbf g}_{t} \widehat{\mbf g}_{t}^\top - \mbf H^{\top} \E\l( \mbf g_{t} \mbf g_{t}^\top \r) \mbf H \r) - \r.
\\
& \qquad \qquad \qquad \qquad \qquad \l. \sum_{j = 0}^R \mathbb{I}_{\{k_j \le k\}} \cdot \mbf V_j^{1/2} W^{(j)}_{1, dT}\l( \min( k, k_{j + 1} ) - k_j \r) \r\Vert
\\
= & \, \frac{1}{\sqrt{\gamma}} \cdot O_P(1) \max_{0 \le k \le T - \gamma} k^{\zeta} = O_P\l( \frac{T^\zeta}{\sqrt{\gamma}} \r) = O_P\l(\frac{1}{\sqrt{\log(T/\gamma)}} \r),
\end{align*}
where the last equality follows from~\eqref{b-mosum}; the term $T_2$ is analogously bounded.
From Theorem~1 in \citet{shao1995conjecture} and the fact that there are finitely many change points, we have $T_3 = O_P(\sqrt{\log(T)})$, which completes the proof.
\end{proof}


\begin{lem}
\label{mosum-alt-2} Suppose that Assumptions~\ref{factors}--\ref{kirch} hold with $\rho = 2$ in Assumptions~\ref{factors}, \ref{idiosyncratic} and~\ref{depFE}. 
Let us define $D_{T} = \min_{1\leq j\leq R}d_j \sqrt{\gamma}$.
Then for any sequence $a_{T}$ satisfying $1\leq a_{T}\leq D_{T}$, and a (slowly varying) sequence $\omega_{T}\rightarrow \infty$, define 
\begin{align*}
\mathcal{M}_{T}^{(\ell )}=& \,\left\{ \max_{1\leq j\leq
R}\max_{d_j^{-2}a_{T}\leq k\leq k_{j}-k_{j-1}}\frac{\sqrt{d_j^{-2}a_{T}}%
}{k}\left\Vert \sum_{t=k_{j}+\ell \gamma -k+1}^{k_{j}+\ell \gamma }\left( 
\widehat{\mathbf{g}}_{t}\widehat{\mathbf{g}}_{t}^{\top }-\mathbf{H}^{\top }%
\mathsf{E}\left( \mathbf{g}_{t}\mathbf{g}_{t}^{\top }\right) \mathbf{H}%
\right) \right\Vert_{F}\leq \omega_{T}\right\} \\
& \bigcap \left\{ \max_{1\leq j\leq R}\max_{d_j^{-2}a_{T}\leq k\leq
k_{j}-k_{j-1}}\frac{\sqrt{d_j^{-2}a_{T}}}{k}\left\Vert \sum_{t=k_{j}+\ell
\gamma +1}^{k_{j}+\ell \gamma +k}\left( \widehat{\mathbf{g}}_{t}\widehat{\mathbf{g}}_{t}^{\top }-\mathbf{H}^{\top }\mathsf{E}\left( \mathbf{g}_{t}\mathbf{g}_{t}^{\top }\right) \mathbf{H}\right) \right\Vert_{F}\leq \omega
_{T}\right\}
\end{align*}
for some $\epsilon >0$ and $\ell \in \{0,\pm 1\}$. Then it holds that, as $\min (N,T)\rightarrow \infty$, 
\begin{equation}
\mathsf{P}\left( \cap_{\ell \in \{0,\pm 1\}}\mathcal{M}_{T}^{(\ell
)}\right) \rightarrow 1. \label{mosum-alt-2-b}
\end{equation}
\end{lem}

\begin{proof}
We base the proof on Proposition~2.1~(c.ii) in \citet{chok}, where a sufficient condition for~\eqref{mosum-alt-2-b} is that 
\begin{align}
\E\l( \l\Vert \sum_{t = a + 1}^b \vech \l( \widehat{\mbf g}_{t} \widehat{\mbf g}_{t}^\top - \mbf H^{\top} \E\l( \mbf g_{t} \mbf g_{t}^\top \r) \mbf H \r) \r\Vert^{2 + \epsilon} \r) \leq c_{0} (b - a)^{1 + \epsilon/2}
\label{cho-kirch}
\end{align}
for some $\epsilon > 0$. 
This in turn follows if we show that
\begin{align}
\E\l( \l\Vert \sum_{t = a + 1}^b \vech \l( \widehat{\mbf g}_{t} \widehat{\mbf g}_{t}^\top - \mbf H^{\top} \mbf g_{t} \mbf g_{t}^\top \mbf H \r) \r\Vert^{2 + \epsilon} \r) &\leq c_{0} (b - a)^{1 + \epsilon/2},
\label{berkes-1} \\
\E\l( \l\Vert \sum_{t = a + 1}^b \vech \l( \mbf g_{t} \mbf g_{t}^\top - \E\l( \mbf g_{t} \mbf g_{t}^\top \r) \r) \r\Vert^{2+\epsilon} \r) &\leq c_{0} (b - a)^{1 + \epsilon/2},
\label{berkes-2}
\end{align}
together with~\eqref{eq:H:bound}.
Equation~\eqref{berkes-2} follows immediately from Proposition~4 of \cite{berkeshormann}, which entails that
\begin{align*}
& \E\l( \l\Vert \sum_{t = a + 1}^b \vech \l( \mbf g_{t} \mbf g_{t}^\top - \E\l( \mbf g_{t} \mbf g_{t}^\top \r) \r) \r\Vert^{2 + \epsilon} \r) 
\\
=& \,
\E\l( \l\Vert \sum_{j = 0}^R \mathbb{I}_{\{a \le k_j < b\}} \cdot \sum_{t = \max(a, k_j) + 1}^{\min(b, k_{j + 1})} \vech \l( \mbf g_{t} \mbf g_{t}^\top - \E\l( \mbf g_{t} \mbf g_{t}^\top \r) \r) \r\Vert^{2 + \epsilon} \r) 
\\
\le & \, \sum_{j = 0}^R \mathbb{I}_{\{a \le k_j < b\}} \cdot \Vert \mbf A_j \Vert^{4 + 2\epsilon} \E\l( \l\Vert \sum_{t = \max(a, k_j) + 1}^{\min(b, k_{j + 1})} \vech \l( \mbf f_{t} \mbf f_{t}^\top - \E\l( \mbf f_{t} \mbf f_{t}^\top \r) \r) \r\Vert^{2 + \epsilon} \r)
\\ 
\le & \, c_{0} R (b - a)^{1 + \epsilon/2},
\end{align*}
by Assumption~\ref{factors}, Assumption~\ref{assum:cps}~\ref{assum:cp:one} and~\ref{assum:cp:three}.
As for~\eqref{berkes-1}, mechanically repeating the arguments in the proofs of Lemmas~\ref{moment-1} and \ref{moment-2}, we obtain that
\begin{align*}
\E\l( \l\Vert \sum_{t = a + 1}^b \vech \l( \widehat{\mbf g}_{t} \widehat{\mbf g}_{t}^\top - \mbf H^\top \mbf g_{t} \mbf g_{t}^\top \mbf H \r) \r\Vert^{2 + \epsilon} \r) \leq c_{0} \l( \frac{b - a}{\min( N, b - a )} \r)^{2 + \epsilon}.
\end{align*}
By elementary arguments 
\begin{align*}
\frac{b - a}{\min(N, b - a)} = \max\l( 1, \frac{b - a}{N} \r) \le \max\l\{ 1, \frac{T^{1/2 + \epsilon^\prime}}{N} (b - a)^{1/2 - \epsilon^\prime} \r\} = o\l( (b - a)^{1/2 - \epsilon^\prime} \r)
\end{align*}
for some $\epsilon^{\prime} \in (0, \epsilon_\circ)$ under Assumption~\ref{assum:nt}.
Therefore,
\begin{align*}
\E\l( \l\Vert \sum_{t = a + 1}^b \vech \l( \widehat{\mbf g}_{t} \widehat{\mbf g}_{t}^\top - \mbf H^\top \mbf g_{t} \mbf g_{t}^\top \mbf H \r) \r\Vert^{2 + \epsilon} \r) \leq c_{0} (b - a)^{ 1/2 + \epsilon/2}.
\end{align*}
Putting all together, the condition in~\eqref{cho-kirch}, which completes the proof. 
\end{proof}

\subsection{Proof of Theorem \protect\ref{thm:null}}

\begin{proof}[Proof of Theorem \ref{thm:null} \ref{thm:null:one}]
Let us define a symmetric, $d \times d$-matrix
\begin{align*}
\wt{\mbf V} = \mbf L_r(\mbf H^\top \otimes \mbf H^\top) \mbf K_r \mbf D_j \mbf K_r^\top (\mbf H \otimes \mbf H) \mbf L_r^\top.
\end{align*}
From Lemma~\ref{lem:hi} and its proof, we have $\mbf H$ asymptotically invertible and $\Vert \mbf H \Vert = O_P(1)$.
Also, from that $\Vert \mbf D \Vert = O(1)$ (due to Assumption~\ref{factors}~\ref{assum:factors:one}), $\Vert \mbf D^{-1} \Vert = O(1)$ (Assumption~\ref{factors}~\ref{assum:factors:three}) and 
\begin{align*}
\Lambda_{\min}(\wt{\mbf V}) \ge \Lambda_{\min}(\mbf D) \l\Vert \mbf L_r(\mbf H^\top \otimes \mbf H^\top) \mbf K_r \r\Vert_F^2,
\end{align*}
we have $\wt{\mbf V}$ asymptotically invertible with 
\begin{align}
\Vert \wt{\mbf V} \Vert = O_P(1) \text{ \ and \ } \Vert \wt{\mbf V}^{-1} \Vert = O_P(1).
\label{eq:wtV}
\end{align}

Then by Lemma~\ref{bound} (with $R = 0$ under $\mc H_0$), there exist two independent $d$-dimensional Wiener processes $W_{\ell, dT}(\cdot), \, \ell = 1, 2$, such that
\begin{align}
\max_{1 \leq k \leq T} \frac{1}{k^{\zeta}} \l\Vert \sum_{t = 1}^{k} \wt{\mbf V}^{-1/2} \vech\l( \widehat{\mbf g}_{t} \widehat{\mbf g}_{t}^{\top} - \mbf H^{\top} \E\l( \mbf g_{t} \mbf g_{t}^\top \r) \mbf H \r) - W_{dT}(k) \r\Vert & = O_P(1), 
\label{sip-bound}
\end{align}
where $W_{dT}(k) = W_{1, dT}(\min(k, T/2)) + (W_{2, dT}(T/2) - W_{2, dT}(T - k)) \cdot \mathbb{I}_{\{k > T/2\}}$, and $\zeta = \max(\zeta_1, \zeta_2)$ with $\zeta_1$ and $\zeta_2$ defined in Lemmas~\ref{sip-1} and~\ref{sip-2},
respectively. 
Following Theorem~S2.1 in \citet{berkesliuwu}, which is referred to in the proof of Lemma~\ref{sip-1}, we have $\zeta_{1} = 2/\nu$ with $\nu$ denote the largest number such that $\E(\vert g_t \vert^\nu) < \infty$; under Assumption~\ref{factors}~\ref{assum:factors:one}, we can set e.g.\ $\nu = 8$. 
Further, inspecting the proof of Lemma~\ref{sip-2}, it emerges that whenever 
\begin{align*}
\frac{1}{2} + \epsilon_\circ < \beta = \frac{\log(N)}{\log(T)} \leq 1,
\end{align*}
it must hold that $1 - \zeta_2 \le \beta$, whereas $\zeta_2 > 0$ can be arbitrarily small when $\beta > 1$.
Hence we set $\zeta_{2} = 1 - \min(1, \beta)$. 
Thus, the statement in~\eqref{sip-bound} holds with $\zeta$ chosen as in~\eqref{eq:zeta}.

The rest of the proof now is similar to that of Theorem~2.1 in \citet{huvskova2001permutation}. 
Note that
\begin{multline*}
\sum_{t = k + 1}^{k + \gamma} \wt{\mbf V}^{-1/2} \vech\l( \widehat{\mbf g}_{t} \widehat{\mbf g}_{t}^{\top} - \mbf H^{\top} \E\l( \mbf g_{t} \mbf g_{t}^\top \r) \mbf H \r) 
\\
= \sum_{t = 1}^{k + \gamma} \wt{\mbf V}^{-1/2} \vech\l( \widehat{\mbf g}_{t} \widehat{\mbf g}_{t}^{\top} - \mbf H^{\top} \E\l( \mbf g_{t} \mbf g_{t}^\top \r) \mbf H \r) - \sum_{t = 1}^{k} \wt{\mbf V}^{-1/2} \vech\l( \widehat{\mbf g}_{t} \widehat{\mbf g}_{t}^{\top} - \mbf H^{\top} \E\l( \mbf g_{t} \mbf g_{t}^\top \r) \mbf H \r).
\end{multline*}
It holds that
\begin{align*}
& \max_{1 \leq k \leq T - \gamma} \frac{1}{\sqrt{2\gamma}} \l\Vert \sum_{t = 1}^{k + \gamma} \wt{\mbf V}^{-1/2} \vech\l( \widehat{\mbf g}_{t} \widehat{\mbf g}_{t}^{\top} - \mbf H^{\top} \E\l( \mbf g_{t} \mbf g_{t}^\top \r) \mbf H \r) - W_{dT}(k + \gamma) \r\Vert 
\\
&= \frac{1}{\sqrt{2\gamma}} \max_{1 \leq k \leq T - \gamma} \frac{1}{(k + \gamma)^\zeta} \l\Vert \sum_{t = 1}^{k + \gamma} \wt{\mbf V}^{-1/2} \vech\l( \widehat{\mbf g}_{t} \widehat{\mbf g}_{t}^{\top} - \mbf H^{\top} \E\l( \mbf g_{t} \mbf g_{t}^\top \r) \mbf H \r) - W_{dT}(k + \gamma) \r\Vert 
\\
& \qquad \qquad \qquad \cdot \max_{1 \leq k \leq T - \gamma} (k + \gamma)^\zeta = O_P\l( T^\zeta \gamma^{-1/2} \r) = o_P\l( \log^{-1/2}(T/\gamma) \r)
\end{align*}
under~\eqref{sip-bound}, \eqref{eq:zeta} and~\eqref{b-mosum}. 
Similarly,
\begin{align*}
& \max_{\gamma < k \leq T} \frac{1}{\sqrt{2\gamma}} \l\Vert \sum_{t = 1}^k \wt{\mbf V}^{-1/2} \vech\l( \widehat{\mbf g}_{t} \widehat{\mbf g}_{t}^{\top} - \mbf H^{\top} \E\l( \mbf g_{t} \mbf g_{t}^\top \r) \mbf H \r) - W_{dT}(k) \r\Vert
\\
=& \, O_P\l( T^\zeta \gamma^{-1/2} \r) = o_P\l( \log^{-1/2}(T/\gamma) \r).
\end{align*}
From the above, we obtain
\begin{multline}
\frac{1}{\sqrt{2\gamma}} \max_{\gamma \leq k \leq T - \gamma} \l\Vert \sum_{t = k + 1}^{k + \gamma} \wt{\mbf V}^{-1/2} \vech\l( \widehat{\mbf g}_{t} \widehat{\mbf g}_{t}^{\top} - \mbf H^{\top} \E\l( \mbf g_{t} \mbf g_{t}^\top \r) \mbf H \r) - \l( W_{dT}(k + \gamma) - W_{dT}(k) \r) \r\Vert
\\ 
= o_P\l( \log^{-1/2}(T/\gamma) \r). \label{mos-1}
\end{multline}
Analogously, we can show that
\begin{multline}
\frac{1}{\sqrt{2\gamma}} \max_{\gamma \leq k \leq T - \gamma} \l\Vert \sum_{t = k - \gamma + 1}^k \wt{\mbf V}^{-1/2} \vech\l( \widehat{\mbf g}_{t} \widehat{\mbf g}_{t}^{\top} - \mbf H^{\top} \E\l( \mbf g_{t} \mbf g_{t}^\top \r) \mbf H \r) - \l( W_{dT}(k) - W_{dT}(k - \gamma) \r) \r\Vert 
\\
= o_P\l( \log^{-1/2}(T/\gamma) \r).\label{mos-2}
\end{multline}
Combining~\eqref{mos-1} and~\eqref{mos-2}, and from the fact that $\E( \mbf g_{t} \mbf g_{t}^\top )$ is time-invariant, we obtain
\begin{multline}
\max_{\gamma \le k \le T -\gamma} \l\Vert \wt{\mbf V}^{-1/2} \mbf M_{N, T, \gamma}(k) \r\Vert
\\
= \frac{1}{\sqrt{2\gamma}} \max_{\gamma \le k \le T - \gamma} \l\Vert W_{dT}(k + \gamma) - W_{dT}(k - \gamma) \r\Vert + o_P\l( \log^{-1/2}(T/\gamma) \r). 
\label{mos-3}
\end{multline}

Let $k = \lfloor \gamma t \rfloor$ with $1 \le t \le T/\gamma - 1$. On account of~\eqref{mos-3}, we will study
\begin{align*}
& \frac{1}{\sqrt{2\gamma}} \max_{\gamma \le k \le T - \gamma} \l\Vert W_{dT}(k + \gamma) - W_{dT}(k - \gamma) \r\Vert
\\
=& \, \frac{1}{\sqrt{2\gamma}} \max_{1 \le t \le T/\gamma} \l\Vert W_{dT}\l( \lfloor \gamma t \rfloor + \gamma \r) - W_{dT}\l( \lfloor \gamma t \rfloor - \gamma \r) \r\Vert
\\
\overset{\mathcal{D}}{=}& \,  
\frac{1}{\sqrt{2}} \max_{1 \le t \le T/\gamma - 1} \l\Vert W_{dT}(t + 1) - W_{dT}(t - 1) \r\Vert,
\end{align*}
having used the scale transformation of the Wiener process. 
Since the distribution of $W_{dT}(t)$ does not depend on $T$, we also have that, as $T \rightarrow \infty$ with~\eqref{b-mosum},
\begin{align*}
\frac{1}{\sqrt{2}} \max_{1 \le t \le T/\gamma - 1} \l\Vert W_{dT}(t + 1) - W_{dT}(t - 1) \r\Vert 
\rightarrow 
\frac{1}{\sqrt{2}} \max_{1 \le t < \infty} \l\Vert W_{dT}(t + 1) - W_{dT}(t - 1) \r\Vert
\end{align*}
almost surely. 
The $d$-dimensional process 
\begin{align}
\omega(t) = \frac{1}{\sqrt 2} \l( W_{dT}(t + 1) - W_{dT}(t - 1) \r),  
\label{omega-t}
\end{align}
has mean zero; elementary calculations yield that it has unit variance and that its coordinates $\omega_i(t), \, 1 \le i \le d$, have covariance given by
\begin{align}
\E\l( \omega_i(t) \omega_i(t + h) \r) = \l\{ 
\begin{array}{ll}
1 - \frac{1}{2} \vert h \vert & \text{for \ } 0 \le \vert h \vert \le 1,
\\
0 & \text{for \ } \vert h \vert > 2.
\end{array}\r. \label{cov-omega}
\end{align}
Hence, $\{ \Vert \omega(t) \Vert, \, 1 \le t < \infty \}$ is a Rayleigh process with index $\alpha = 1$ (for its definition, refer to Section~3 of \citealp{steinebach1996extreme}).
Thus, by Lemma~3.1 of \citet{steinebach1996extreme} and Slutsky's theorem, we have
\begin{align}
\lim_{\min(N, T) \rightarrow \infty} \p\l( a\l( \frac{T}{\gamma} \r) \max_{\gamma \le k \le T -\gamma} \l\vert \mbf M_{N, T, \gamma}^\top(k) \wt{\mbf V}^{-1} \mbf M_{N, T, \gamma}(k) \r\vert^{1/2} - b_d\l( \frac{T}{\gamma} \r) \le x \r) 
\nn \\
= \exp \l(- 2\exp(-x) \r).
\label{eq:null:one}
\end{align}
This, together with~\eqref{eq:wtV}, implies that
\begin{align}
\label{eq:M:bound}
\max_{\gamma \le k \le T - \gamma} \l\Vert \mbf M_{N, T, \gamma}(k) \r\Vert = O_P\l( \sqrt{\log(T/\gamma)} \r).
\end{align}
Also, we have
\begin{align}
\l\Vert \wt{\mbf V} - \mbf V \r\Vert &= O_P\l( \l\Vert (\mbf H \otimes \mbf H) - (\mbf H_0 \otimes \mbf H_0) \r\Vert \r) = O_P\l( \l\Vert \mbf H - \mbf H_0 \r\Vert \r) 
\nn \\
&= O_P\l( C_{NT}^{-1} \r) = o_P\l( \log^{-1}(T/\gamma) \r)
\end{align}
by Assumption~\ref{assum:nt} and Lemma~\ref{lem:hi}.
Then, e.g.\ by Lemma~4.1 of \cite{powers1970free}, we have
\begin{align*}
\l\Vert \wt{\mbf V}^{-1/2} - \mbf V^{-1/2} \r\Vert
&= 
o_P\l( \log^{-1}(T/\gamma) \r).
\end{align*}
This, together with~\eqref{eq:M:bound}, establishes that we can replace $\wt{\mbf V}$ with $\mbf V$ and continue to have the asymptotic distribution in~\eqref{eq:null:one} hold, since
\begin{align}
& \l\vert \max_{\gamma \le k \le T -\gamma} \l\Vert \wt{\mbf V}^{-1/2} \mbf M_{N, T, \gamma}(k) \r\Vert - \max_{\gamma \le k \le T -\gamma} \l\Vert \mbf V^{-1/2} \mbf M_{N, T, \gamma}(k) \r\Vert \r\vert
\nn \\
\le & \, \max_{\gamma \le k \le T -\gamma} \l\Vert \l( \wt{\mbf V}^{-1/2} - \mbf V^{-1/2} \r) \mbf M_{N, T, \gamma}(k) \r\Vert 
\nn \\
\le & \, \l\Vert \wt{\mbf V}^{-1/2} - \mbf V^{-1/2} \r\Vert \cdot \max_{\gamma \le k \le T -\gamma} \l\Vert \mbf M_{N, T, \gamma}(k) \r\Vert
\nn \\
=& \, o_P\l( \log^{-1}(T/\gamma) \r) \cdot O_P\l( \sqrt{\log(T/\gamma)} \r) = o_P\l( a^{-1}(T/\gamma) \r).
\label{eq:vtide2v}
\end{align}
\end{proof}

\begin{proof}[Proof of Theorem \ref{thm:null} \ref{thm:null:two}]
Arguments analogous to those leading to~\eqref{eq:vtide2v} can be adopted under~\eqref{cov-stronger}, and thus we omit the proof.
\end{proof}

\subsection{Proof of Proposition~\protect\ref{cor:V}}

The proof follows similar passages to \citet{han2014}, and therefore we
focus only on some aspects of it. We will use the following notations 
\begin{align*}
\widehat{\bm\Gamma }(\ell )& =\frac{1}{T}\sum_{t=\ell +1}^{T}\mathbf{Z}_{t}%
\mathbf{Z}_{t-\ell }^{\top }\text{ \ with \ }\mathbf{Z}_{t}=\mathsf{Vech}%
\left( \widehat{\mathbf{g}}_{t}\widehat{\mathbf{g}}_{t}^{\top }-\mathbf{I}%
_{r}\right) , \\
\bm\Gamma (\ell )& =\mathsf{E}\left( \widetilde{\mathbf{U}}_{t}\widetilde{%
\mathbf{U}}_{t-\ell }^{\top }\right) \text{ \ with \ }\widetilde{\mathbf{U}}%
_{t}=\mathsf{Vech}\left( \mathbf{H}_{0}^{\top }\mathbf{g}_{t}\mathbf{g}%
_{t}^{\top }\mathbf{H}_{0}-\mathbf{I}_{r}\right) ,\text{ \ and} \\
\mathbf{U}_{t}& =\mathsf{Vech}\left( \mathbf{H}^{\top }\mathbf{g}_{t}\mathbf{%
g}_{t}^{\top }\mathbf{H}-\mathbf{H}^{\top }\mathsf{E}(\mathbf{g}_{t}\mathbf{g%
}_{t}^{\top })\mathbf{H}\right) .
\end{align*}
Noting that $\mathbf{H}_{0}^{\top }\bm\Sigma_{G}\mathbf{H}_{0}=\mathbf{I}%
_{r}$ (see the proof of Lemma~\ref{lem:hi}), it follows that 
\begin{equation*}
\mathbf{V}=\bm\Gamma (0)+\sum_{\ell =1}^{\infty }\left( \bm\Gamma (\ell )+\bm%
\Gamma (\ell )^{\top }\right) ,
\end{equation*}%
so that 
\begin{align}
\widehat{\mathbf{V}}-\mathbf{V}=& \,\left( \widehat{\bm\Gamma }(0)-\bm\Gamma
(0)\right) +\sum_{\ell =1}^{m}\left( 1-\frac{\ell }{m+1}\right) \left[
\left( \widehat{\bm\Gamma }(\ell )-\bm\Gamma (\ell )\right) +\left( \widehat{\bm\Gamma }(\ell )-\bm\Gamma (\ell )\right)^{\top}\right]  \notag \\
& +\sum_{\ell =1}^{m}\frac{\ell }{m+1}\left( \bm\Gamma (\ell )+\bm\Gamma
(\ell )^{\top }\right) +\sum_{\ell =m+1}^{\infty }\left( \bm\Gamma (\ell )+%
\bm\Gamma (\ell )^{\top }\right) .  \label{vcv-error}
\end{align}
Since $\mathbf{g}_{t}$ is an $\mathcal{L}_{8 + \epsilon}$-decomposable Bernoulli shift
with $a>2$, it is easy to see (cfr.\ the proof of Lemma~\ref{sip-1}) that $\mathsf{Vech}(\mathbf{g}_{t}\mathbf{g}_{t}^{\top })$ is an $\mathcal{L}_{4 + \epsilon/2}$
-decomposable Bernoulli shift, also with $a>2$. The covariance summability
of Bernoulli shifts (see e.g.\ Lemma~D.4 in \citealp{HT2022}) entails that 
\begin{align}
\left\Vert \sum_{\ell =1}^{m}\frac{\ell }{m+1}\left( \bm\Gamma (\ell )+\bm%
\Gamma (\ell )^{\top }\right) \right\Vert & =O\left( \frac{1}{m}\right) ,
\label{vcv-proof-4} \\
\left\Vert \sum_{\ell =m+1}^{\infty }\left( \bm\Gamma (\ell )+\bm\Gamma
(\ell )^{\top }\right) \right\Vert & =O\left( \frac{1}{m}\right) .
\label{vcv-proof-5}
\end{align}

We now bound the rest of the terms in~\eqref{vcv-error}. First, note that 
\begin{align}
\sum_{\ell =1}^{m}\frac{1}{T}\sum_{t=\ell +1}^{T}\mathbf{Z}_{t}\mathbf{Z}%
_{t-\ell }^{\top }=& \,\sum_{\ell =1}^{m}\frac{1}{T}\sum_{t=\ell +1}^{T}%
\mathbf{U}_{t}\mathbf{U}_{t-\ell }^{\top }+\sum_{\ell =1}^{m}\frac{1}{T}%
\sum_{t=\ell +1}^{T}\mathbf{U}_{t}(\mathbf{Z}_{t-\ell }-\mathbf{U}_{t-\ell
})^{\top }  \notag \\
& +\sum_{\ell =1}^{m}\frac{1}{T}\sum_{t=\ell +1}^{T}(\mathbf{Z}_{t}-\mathbf{U%
}_{t})\mathbf{U}_{t-\ell }^{\top }+\sum_{\ell =1}^{m}\frac{1}{T}\sum_{t=\ell
+1}^{T}(\mathbf{Z}_{t}-\mathbf{U}_{t})(\mathbf{Z}_{t-\ell }-\mathbf{U}%
_{t-\ell })^{\top }  \notag \\
=:& \,T_{1}+T_{2}+T_{3}+T_{4}.  \label{vcv-1}
\end{align}
First we study $T_{2}$ in~\eqref{vcv-1}. For simplicity, let $r=d=1$ and
omit $\mathbf{H}$ noting that $\Vert \mathbf{H} \Vert = O_P(1)$ due to~%
\eqref{eq:H:bound}.{\ Further, we may treat $\mathsf{E}(g_{t}^{2})=1$.}
Then, we can write 
\begin{align}
T_{2}& =\frac{1}{T}\sum_{t=1}^{T}\left( g_{t}^{2}-\mathsf{E}%
(g_{t}^{2})\right) \left( \sum_{\ell =1}^{m}\left( \widehat{g}_{t-\ell
}-g_{t-\ell }\right)^{2}+2\sum_{\ell =1}^{m}g_{t-\ell }\left( \widehat{g}%
_{t-\ell }-g_{t-\ell }\right) \right) \mathbb{I}_{\{\ell <t\}}  \notag \\
& \leq \frac{1}{T}\left( \sum_{t=1}^{T}\left( g_{t}^{2}-\mathsf{E}%
(g_{t}^{2})\right)^{2}\right)^{1/2}\left( \sum_{t=1}^{T}\left( \sum_{\ell
=1}^{m}\left( \widehat{g}_{t-\ell }-g_{t-\ell }\right)^{2}\right)
^{2}+4\sum_{t=1}^{T}\left( \sum_{\ell =1}^{m}g_{t-\ell }\left( \widehat{g}%
_{t-\ell }-g_{t-\ell }\right) \right)^{2}\right)^{1/2}  \notag \\
& \leq \frac{1}{T}\left( \sum_{t=1}^{T}\left( g_{t}^{2}-\mathsf{E}%
(g_{t}^{2})\right)^{2}\right)^{1/2}\left( \sum_{t=1}^{T}\left( \sum_{\ell
=1}^{m}\left( \widehat{g}_{t-\ell }-g_{t-\ell }\right)^{2}\right)
^{2}\right)^{1/2}  \notag \\
& \quad +\frac{2}{T}\left( \sum_{t=1}^{T}\left( g_{t}^{2}-\mathsf{E}%
(g_{t}^{2})\right)^{2}\right)^{1/2}\left( \sum_{t=1}^{T}\left( \sum_{\ell
=1}^{m}g_{t-\ell }\left( \widehat{g}_{t-\ell }-g_{t-\ell }\right) \right)
^{2}\right)^{1/2}=:T_{2,1}+T_{2,2}.  \label{vcv-2}
\end{align}
Assumption~\ref{factors}~\ref{assum:factors:one} entails $\sum_{t=1}^{T}(g_{t}^{2}-\mathsf{E}(g_{t}^{2}))^{2}=O_{P}(T)$. Also, by~%
\eqref{eq:bai}, 
\begin{align*}
& \sum_{t=1}^{T}\left( \sum_{\ell =1}^{m}\left( \widehat{g}_{t-\ell
}-g_{t-\ell }\right)^{2}\right)^{2} \\
\leq & \,\sum_{t=1}^{T}\left( \sum_{\ell =1}^{m}\left\vert \frac{1}{T}%
\sum_{s=1}^{T}\widehat{g}_{s}\gamma_{s,t-\ell }\right\vert^{2}\right)
^{2}+\sum_{t=1}^{T}\left( \sum_{\ell =1}^{m}\left\vert \frac{1}{T}%
\sum_{s=1}^{T}\widehat{g}_{s}\zeta_{s,t-\ell }\right\vert^{2}\right)^{2}
\\
& \,+\sum_{t=1}^{T}\left( \sum_{\ell =1}^{m}\left\vert \frac{1}{T}%
\sum_{s=1}^{T}\widehat{g}_{s}\eta_{s,t-\ell }\right\vert^{2}\right)
^{2}+\sum_{t=1}^{T}\left( \sum_{\ell =1}^{m}\left\vert \frac{1}{T}%
\sum_{s=1}^{T}\widehat{g}_{s}\xi_{s,t-\ell }\right\vert^{2}\right)^{2}.
\end{align*}
First, we have 
\begin{align*}
& \sum_{t=1}^{T}\left( \sum_{\ell =1}^{m}\left\vert \frac{1}{T}\sum_{s=1}^{T}%
\widehat{g}_{s}\gamma_{s,t-\ell }\right\vert^{2}\right)^{2} \\
\leq & \,T^{-4}m\sum_{\ell =1}^{m}\sum_{t=1}^{T}\left\vert \sum_{s=1}^{T}%
\widehat{g}_{s}\gamma_{s,t-\ell }\right\vert^{4}\leq T^{-4}m\sum_{\ell
=1}^{m}\sum_{t=1}^{T}\left( \sum_{s=1}^{T}\widehat{g}_{s}^{2}\right)
^{2}\left( \sum_{s=1}^{T}\gamma_{s,t-\ell }^{2}\right)^{2} \\
\leq & \,T^{-2}m\sum_{\ell =1}^{m}\sum_{t=1}^{T}\left( \sum_{s=1}^{T}\gamma
_{s,t-\ell }^{2}\right)^{2}\leq T^{-2}m\sum_{\ell
=1}^{m}\sum_{t=1}^{T}\left( \sum_{s=1}^{T}|\gamma_{s,t-\ell }|\right)
^{4}\leq c_{0}T^{-1}m^{2}
\end{align*}
by Assumption~\ref{idiosyncratic}~\ref{assum:idiosyncratic:two}. Thus, it
holds that 
\begin{equation*}
\frac{1}{T}\left( \sum_{t=1}^{T}\left( g_{t}^{2}-\mathsf{E}%
(g_{t}^{2})\right)^{2}\right)^{1/2}\left( \sum_{t=1}^{T}\left( \sum_{\ell
=1}^{m}\left\vert \frac{1}{T}\sum_{s=1}^{T}\widehat{g}_{s}\gamma_{s,t-\ell
}\right\vert^{2}\right)^{2}\right)^{1/2}=O_{P}\left( \frac{m}{T}\right) .
\end{equation*}%
Using similar arguments, 
\begin{align*}
& \mathsf{E}\left[ \sum_{t=1}^{T}\left( \sum_{\ell =1}^{m}\left\vert \frac{1%
}{T}\sum_{s=1}^{T}\widehat{g}_{s}\zeta_{s,t-\ell }\right\vert^{2}\right)
^{2}\right] \\
\leq & \,\mathsf{E}\left[ T^{-2}m\sum_{\ell =1}^{m}\sum_{t=1}^{T}\left(
\sum_{s=1}^{T}\zeta_{s,t-\ell }^{2}\right)^{2}\right] \leq
T^{-1}m\sum_{\ell =1}^{m}\sum_{t=1}^{T}\sum_{s=1}^{T}\mathsf{E}\left( \zeta
_{s,t-\ell }^{4}\right) \leq c_{0}\frac{m^{2}T}{N^{2}}
\end{align*}
by Assumption~\ref{idiosyncratic}~\ref{assum:idiosyncratic:three}, whereby 
\begin{equation*}
\frac{1}{T}\left( \sum_{t=1}^{T}\left( g_{t}^{2}-\mathsf{E}%
(g_{t}^{2})\right)^{2}\right)^{1/2}\left( \sum_{t=1}^{T}\left( \sum_{\ell
=1}^{m}\left\vert \frac{1}{T}\sum_{s=1}^{T}\widehat{g}_{s}\zeta_{s,t-\ell
}\right\vert^{2}\right)^{2}\right)^{1/2}=O_{P}\left( \frac{m}{N}\right) .
\end{equation*}%
Similarly we get 
\begin{align*}
& \sum_{t=1}^{T}\left( \sum_{\ell =1}^{m}\left\vert \frac{1}{T}\sum_{s=1}^{T}%
\widehat{g}_{s}\eta_{s,t-\ell }\right\vert^{2}\right)^{2} \\
\leq & \,T^{-1}m\sum_{\ell =1}^{m}\sum_{t=1}^{T}\sum_{s=1}^{T}\left\vert
\eta_{s,t-\ell }\right\vert^{4}=T^{-1}m\left( \sum_{\ell
=1}^{m}\sum_{t=1}^{T}\left\vert \frac{1}{N}\sum_{i=1}^{N}\lambda
_{i}e_{i,t-\ell }\right\vert^{4}\right) \left( \sum_{s=1}^{T}\left\vert
g_{s}\right\vert^{4}\right) \\
=& \,O_{P}\left( T\right) T^{-1}m\left( \sum_{\ell
=1}^{m}\sum_{t=1}^{T}\left\vert \frac{1}{N}\sum_{i=1}^{N}\lambda
_{i}e_{i,t-\ell }\right\vert^{4}\right) =O_{P}\left( \frac{m^{2}T}{N^{2}}%
\right) ,
\end{align*}
having used Assumption~\ref{idiosyncratic}~\ref{assum:idiosyncratic:four} in
the final passage. Thus 
\begin{equation*}
\frac{1}{T}\left( \sum_{t=1}^{T}\left( g_{t}^{2}-\mathsf{E}%
(g_{t}^{2})\right)^{2}\right)^{1/2}\left( \sum_{t=1}^{T}\left( \sum_{\ell
=1}^{m}\left\vert \frac{1}{T}\sum_{s=1}^{T}\widehat{g}_{s}\eta_{s,t-\ell
}\right\vert^{2}\right)^{2}\right)^{1/2}=O_{P}\left( \frac{m}{N}\right) .
\end{equation*}%
By the same arguments, we also obtain 
\begin{equation*}
\frac{1}{T}\left( \sum_{t=1}^{T}\left( g_{t}^{2}-\mathsf{E}%
(g_{t}^{2})\right)^{2}\right)^{1/2}\left( \sum_{t=1}^{T}\left( \sum_{\ell
=1}^{m}\left\vert \frac{1}{T}\sum_{s=1}^{T}\widehat{g}_{s}\xi_{s,t-\ell
}\right\vert^{2}\right)^{2}\right)^{1/2}=O_{P}\left( \frac{m}{N}\right) ,
\end{equation*}%
and therefore in~\eqref{vcv-2}, $T_{2,1}=O_{P}(mC_{NT}^{-2})$. Similarly, we
note that 
\begin{align*}
& \sum_{t=1}^{T}\left( \sum_{\ell =1}^{m}g_{t-\ell }\left( \widehat{g}%
_{t-\ell }-g_{t-\ell }\right) \right)^{2} \\
\leq & \,\sum_{t=1}^{T}\left( \sum_{\ell =1}^{m}\left\vert \frac{1}{T}%
g_{t-\ell }\sum_{s=1}^{T}\widehat{g}_{s}\gamma_{s,t-\ell }\right\vert
\right)^{2}+\sum_{t=1}^{T}\left( \sum_{\ell =1}^{m}\left\vert \frac{1}{T}%
g_{t-\ell }\sum_{s=1}^{T}\widehat{g}_{s}\zeta_{s,t-\ell }\right\vert
\right)^{2} \\
& \,+\sum_{t=1}^{T}\left( \sum_{\ell =1}^{m}\left\vert \frac{1}{T}g_{t-\ell
}\sum_{s=1}^{T}\widehat{g}_{s}\eta_{s,t-\ell }\right\vert \right)
^{2}+\sum_{t=1}^{T}\left( \sum_{\ell =1}^{m}\left\vert \frac{1}{T}g_{t-\ell
}\sum_{s=1}^{T}\widehat{g}_{s}\xi_{s,t-\ell }\right\vert \right)^{2}.
\end{align*}
From Assumption~\ref{idiosyncratic}~\ref{assum:idiosyncratic:two}, it holds
that 
\begin{align*}
& \sum_{t=1}^{T}\left( \sum_{\ell =1}^{m}\left\vert \frac{1}{T}g_{t-\ell
}\sum_{s=1}^{T}\widehat{g}_{s}\gamma_{s,t-\ell }\right\vert \right)^{2} \\
\leq & \,mT^{-2}\sum_{\ell =1}^{m}\sum_{t=1}^{T}g_{t-\ell }^{2}\left\vert
\sum_{s=1}^{T}\widehat{g}_{s}\gamma_{s,t-\ell }\right\vert^{2}\leq
mT^{-2}\sum_{\ell =1}^{m}\sum_{t=1}^{T}g_{t-\ell }^{2}\left( \sum_{s=1}^{T}%
\widehat{g}_{s}^{2}\right) \left( \sum_{s=1}^{T}\gamma_{s,t-\ell
}^{2}\right) \\
\leq & \,mT^{-1}\sum_{\ell =1}^{m}\sum_{t=1}^{T}g_{t-\ell }^{2}\left(
\sum_{s=1}^{T}\left\vert \gamma_{s,t-\ell }\right\vert \right)
^{2}=O_{P}(m^{2}),
\end{align*}
with which we obtain 
\begin{equation*}
\frac{1}{T}\left( \sum_{t=1}^{T}\left( g_{t}^{2}-\mathsf{E}%
(g_{t}^{2})\right)^{2}\right)^{1/2}\left( \sum_{t=1}^{T}\left( \sum_{\ell
=1}^{m}\left\vert \frac{1}{T}g_{t-\ell }\sum_{s=1}^{T}\widehat{g}_{s}\gamma
_{s,t-\ell }\right\vert \right)^{2}\right)^{1/2}=O_{P}\left( \frac{m}{%
T^{1/2}}\right) .
\end{equation*}%
Similarly, 
\begin{align*}
& \mathsf{E}\left[ \sum_{t=1}^{T}\left( \sum_{\ell =1}^{m}\left\vert \frac{1%
}{T}g_{t-\ell }\sum_{s=1}^{T}\widehat{g}_{s}\zeta_{s,t-\ell }\right\vert
\right)^{2}\right] \\
\leq & \,mT^{-2}\mathsf{E}\left[ \sum_{\ell =1}^{m}\sum_{t=1}^{T}g_{t-\ell
}^{2}\left\vert \sum_{s=1}^{T}\widehat{g}_{s}\zeta_{s,t-\ell }\right\vert
^{2}\right] \leq mT^{-2}\mathsf{E}\left[ \sum_{\ell
=1}^{m}\sum_{t=1}^{T}g_{t-\ell }^{2}\left( \sum_{s=1}^{T}\widehat{g}%
_{s}^{2}\right) \left( \sum_{s=1}^{T}\zeta_{s,t-\ell }^{2}\right) \right] \\
\leq & \,mT^{-1}\left( \sum_{\ell =1}^{m}\sum_{t=1}^{T}\mathsf{E}\left(
|g_{t-\ell }|^{4}\right) \right)^{1/2}\left( \mathsf{E}\left[ \sum_{\ell
=1}^{m}\sum_{t=1}^{T}\left( \sum_{s=1}^{T}\zeta_{s,t-\ell }^{2}\right)^{2}%
\right] \right)^{1/2} \\
\leq & \,mT^{-1}O\left( \left( mT\right)^{1/2}\right) \left( T\sum_{\ell
=1}^{m}\sum_{t=1}^{T}\sum_{s=1}^{T}\mathsf{E}\left( \left\vert \zeta
_{s,t-\ell }\right\vert^{4}\right) \right)^{1/2}=O\left( \frac{m^{2}T}{N}%
\right) ,
\end{align*}
by Assumption~\ref{idiosyncratic}~\ref{assum:idiosyncratic:two}, so that 
\begin{equation*}
\frac{1}{T}\left( \sum_{t=1}^{T}\left( g_{t}^{2}-\mathsf{E}%
(g_{t}^{2})\right)^{2}\right)^{1/2}\left( \sum_{t=1}^{T}\left( \sum_{\ell
=1}^{m}\left\vert \frac{1}{T}g_{t-\ell }\sum_{s=1}^{T}\widehat{g}_{s}\zeta
_{s,t-\ell }\right\vert \right)^{2}\right)^{1/2}=O_{P}\left( \frac{m}{%
N^{1/2}}\right) .
\end{equation*}%
Analogously, 
\begin{align*}
& \sum_{t=1}^{T}\left( \sum_{\ell =1}^{m}\left\vert \frac{1}{T}g_{t-\ell
}\sum_{s=1}^{T}\widehat{g}_{s}\eta_{s,t-\ell }\right\vert \right)^{2} \\
\leq & \,mT^{-1}\sum_{\ell =1}^{m}\sum_{t=1}^{T}g_{t-\ell
}^{2}\sum_{s=1}^{T}\left\vert g_{s}\frac{1}{N}\sum_{i=1}^{N}\lambda
_{i}e_{i,t-\ell }\right\vert^{2}=mT^{-1}\left(
\sum_{s=1}^{T}g_{s}^{2}\right) \sum_{\ell =1}^{m}\sum_{t=1}^{T}g_{t-\ell
}^{2}\left\vert \frac{1}{N}\sum_{i=1}^{N}\lambda_{i}e_{i,t-\ell
}\right\vert^{2} \\
\leq & \,mT^{-1}O_{P}(T)\left( \sum_{\ell =1}^{m}\sum_{t=1}^{T}g_{t-\ell
}^{4}\right)^{1/2}\left( \sum_{\ell =1}^{m}\sum_{t=1}^{T}\left\vert \frac{1%
}{N}\sum_{i=1}^{N}\lambda_{i}e_{i,t-\ell }\right\vert^{4}\right)^{1/2} \\
=& \,mT^{-1}O_{P}\left( T\cdot \left( mT\right)^{1/2}\cdot \frac{\left(
mT\right)^{1/2}}{N}\right) =O_{P}\left( \frac{m^{2}T}{N}\right) .
\end{align*}
Hence, 
\begin{equation*}
\frac{1}{T}\left( \sum_{t=1}^{T}\left( g_{t}^{2}-\mathsf{E}%
(g_{t}^{2})\right)^{2}\right)^{1/2}\left( \sum_{t=1}^{T}\left( \sum_{\ell
=1}^{m}\left\vert \frac{1}{T}g_{t-\ell }\sum_{s=1}^{T}\widehat{g}_{s}\eta
_{s,t-\ell }\right\vert \right)^{2}\right)^{1/2}=O_{P}\left( \frac{m}{%
N^{1/2}}\right) .
\end{equation*}%
The same passages, in essence, yield 
\begin{equation*}
\frac{1}{T}\left( \sum_{t=1}^{T}\left( g_{t}^{2}-\mathsf{E}%
(g_{t}^{2})\right)^{2}\right)^{1/2}\left( \sum_{t=1}^{T}\left( \sum_{\ell
=1}^{m}\left\vert \frac{1}{T}g_{t-\ell }\sum_{s=1}^{T}\widehat{g}_{s}\xi
_{s,t-\ell }\right\vert \right)^{2}\right)^{1/2}=O_{P}\left( \frac{m}{%
N^{1/2}}\right) ,
\end{equation*}%
so that $T_{2,2}=O_{P}(mC_{NT}^{-1})$ in~\eqref{vcv-2}. Therefore, we
finally have 
\begin{equation*}
T_{2}=O_{P}\left( \frac{m}{C_{NT}}\right) .
\end{equation*}%
Following the analogous arguments, $T_{3}$ and $T_{4}$ are similarly
bounded, from which we conclude that 
\begin{equation}
\left\Vert \sum_{\ell =1}^{m}\frac{1}{T}\sum_{t=\ell +1}^{T}\mathbf{Z}_{t}%
\mathbf{Z}_{t-\ell }^{\top }-\sum_{\ell =1}^{m}\frac{1}{T}\sum_{t=\ell
+1}^{T}\mathbf{U}_{t}\mathbf{U}_{t-\ell }^{\top }\right\Vert =O_{P}\left( 
\frac{m}{C_{NT}}\right) , \label{vcv-proof-1}
\end{equation}%
and repeating essentially the same passages, it can be shown that 
\begin{equation*}
\left\Vert \frac{1}{T}\sum_{t=1}^{T}\mathbf{Z}_{t}\mathbf{Z}_{t}^{\top }-%
\frac{1}{T}\sum_{t=1}^{T}\mathbf{U}_{t}\mathbf{U}_{t}^{\top }\right\Vert
=O_{P}\left( \frac{1}{C_{NT}}\right) .
\end{equation*}%
Next, note that by definition, we can write 
\begin{equation*}
\mathbf{U}_{t}\mathbf{U}_{t-\ell }^{\top }=\mathbf{L}_{r}(\mathbf{H}^{\top
}\otimes \mathbf{H}^{\top })\mathsf{Vec}\left( \mathbf{g}_{t}\mathbf{g}%
_{t}^{\top }\right) \mathsf{Vec}\left( \mathbf{g}_{t - \ell}\mathbf{g}_{t - \ell}^{\top
}\right)^{\top}(\mathbf{H}\otimes \mathbf{H})\mathbf{L}_{r}^{\top },
\end{equation*}%
and a similar representation holds for $\widetilde{\mathbf{U}}_{t}\widetilde{%
\mathbf{U}}_{t-\ell }^{\top}$ with $\mathbf{H}_{0}$ replacing $\mathbf{H}$.
Then, by Lemma~\ref{lem:hi}, it can be verified that 
\begin{equation}
\left\Vert \sum_{\ell =1}^{m}\frac{1}{T}\sum_{t=\ell +1}^{T}\mathbf{U}_{t}%
\mathbf{U}_{t-\ell }^{\top }-\sum_{\ell =1}^{m}\frac{1}{T}\sum_{t=\ell
+1}^{T}\widetilde{\mathbf{U}}_{t}\widetilde{\mathbf{U}}_{t-\ell }^{\top
}\right\Vert =O_{P}\left( \frac{m}{C_{NT}}\right) . \label{vcv-proof-1-1}
\end{equation}%
We now consider bounding 
\begin{equation*}
\sum_{\ell =1}^{m}\frac{1}{T}\sum_{t=\ell +1}^{T}\widetilde{\mathbf{U}}_{t}%
\widetilde{\mathbf{U}}_{t-\ell }^{\top }-\sum_{\ell =1}^{m}\frac{1}{T}%
\sum_{t=\ell +1}^{T}\mathsf{E}\left( \widetilde{\mathbf{U}}_{t}\widetilde{%
\mathbf{U}}_{t-\ell }^{\top }\right)
\end{equation*}%
and again, we set $r=d=1$ for simplicity. Under Assumption~\ref{factors}~\ref%
{assum:factors:one}, it is easy to see that for all $j\geq 0$, the sequence $\mathbf{S}_{t,\ell }=\widetilde{\mathbf U}_{t}\widetilde{\mathbf U}_{t-\ell }^{\top }-\mathsf{E}(%
\widetilde{\mathbf U}_{t}\widetilde{\mathbf U}_{t-\ell }^{\top })$ is an $\mathcal{L}_{2}$-decomposable Bernoulli shift with $a>2$. Hence 
\begin{equation*}
\mathsf{E}\left( \left\Vert \frac{1}{T}\sum_{\ell =1}^{m}\sum_{t=\ell
+1}^{T}\mathbf S_{t,\ell }\right\Vert^{2}\right) \leq T^{-2}m\sum_{\ell =1}^{m}%
\mathsf{E}\left( \left\Vert \sum_{t=\ell +1}^{T}\mathbf S_{t,\ell }\right\Vert
^{2}\right) \leq \frac{c_{0}m^{2}}{T},
\end{equation*}%
where the last passage follows from Proposition~4 in \citet{berkeshormann}.
Hence we have 
\begin{equation}
\left\Vert \sum_{\ell =1}^{m}\frac{1}{T}\sum_{t=\ell +1}^{T}\widetilde{%
\mathbf{U}}_{t}\widetilde{\mathbf{U}}_{t-\ell }^{\top }-\sum_{\ell =1}^{m}%
\frac{1}{T}\sum_{t=\ell +1}^{T}\mathsf{E}\left( \widetilde{\mathbf{U}}_{t}%
\widetilde{\mathbf{U}}_{t-\ell }^{\top }\right) \right\Vert =O_{P}\left( 
\frac{m}{T^{1/2}}\right) . \label{vcv-proof-2}
\end{equation}%
Thus, putting together \eqref{vcv-proof-1}, \eqref{vcv-proof-1-1} and~%
\eqref{vcv-proof-2}, it follows that 
\begin{equation}
\left\Vert \sum_{\ell =1}^{m}\frac{1}{T}\sum_{t=\ell +1}^{T}\mathbf{Z}_{t}%
\mathbf{Z}_{t-\ell }^{\top }-\sum_{\ell =1}^{m}\frac{1}{T}\sum_{t=\ell
+1}^{T}\mathsf{E}\left( \widetilde{\mathbf{U}}_{t}\widetilde{\mathbf{U}}%
_{t-\ell }^{\top }\right) \right\Vert =O_{P}\left( \frac{m}{C_{NT}}\right) \label{vcv-proof-3}
\end{equation}%
and similarly, 
\begin{equation}
\left\Vert \frac{1}{T}\sum_{t=1}^{T}\mathbf{Z}_{t}\mathbf{Z}_{t}^{\top }-%
\frac{1}{T}\sum_{t=1}^{T}\mathsf{E}\left( \widetilde{\mathbf{U}}_{t}%
\widetilde{\mathbf{U}}_{t}^{\top }\right) \right\Vert =O_{P}\left( \frac{1}{%
C_{NT}}\right) . \label{vcv-proof-3-bis}
\end{equation}%
Finally, using~\eqref{vcv-proof-4}, \eqref{vcv-proof-5}, \eqref{vcv-proof-3}
and~\eqref{vcv-proof-3-bis} in \eqref{vcv-error}, it yields 
\begin{equation*}
\left\Vert \widehat{\mathbf{V}}-\mathbf{V}\right\Vert =O_{P}\left( \frac{m}{%
C_{NT}}\right) +O\left( \frac{1}{m}\right) =O_{P}\left( \frac{1}{\log
(T/\gamma )}\right) ,
\end{equation*}%
from the conditions made in~\eqref{b-mosum-restrict} on $m$. 

\subsection{Proof of Theorem~\protect\ref{thm:consistency}}

WLOG, we may regard $\mathbf{V}=\mathbf{I}_{r}$, which does not
alter the arguments as $\mathbf{V}$ is a positive definite matrix under
Assumption~\ref{factors}~\ref{assum:factors:three} with bounded eigenvalues. Also for simplicity, we write $\mathbf{M}(k)=%
\mathbf{M}_{N,T,\gamma }(k)$ and $\mathcal{T}(k)=\mathcal{T}_{N,T,\gamma
}(k) $. Decompose $\mathbf{M}(k)$ as 
\begin{align*}
\mathbf{M}(k)& =\frac{1}{\sqrt{2\gamma }}\left[ \sum_{t=k+1}^{k+\gamma }%
\mathsf{Vech}\left( \widehat{\mathbf{g}}_{t}\widehat{\mathbf{g}}_{t}^{\top }-%
\mathbf{H}^{\top }\mathsf{E}(\mathbf{g}_{t}\mathbf{g}_{t}^{\top })\mathbf{H}%
\right) -\sum_{t=k-\gamma +1}^{k}\mathsf{Vech}\left( \widehat{\mathbf{g}}_{t}%
\widehat{\mathbf{g}}_{t}^{\top }-\mathbf{H}^{\top }\mathsf{E}(\mathbf{g}_{t}%
\mathbf{g}_{t}^{\top })\mathbf{H}\right) \right] \\
& \qquad +\frac{1}{\sqrt{2\gamma }}\left[ \sum_{t=k+1}^{k+\gamma }\mathsf{%
Vech}\left( \mathbf{H}^{\top }\mathsf{E}(\mathbf{g}_{t}\mathbf{g}_{t}^{\top
})\mathbf{H}\right) -\sum_{t=k-\gamma +1}^{k}\mathsf{Vech}\left( \mathbf{H}%
^{\top }\mathsf{E}(\mathbf{g}_{t}\mathbf{g}_{t}^{\top })\mathbf{H}\right) %
\right] \\
& =:\mathbf{N}(k)+\mathbf{S}(k).
\end{align*}
Then, we can write 
\begin{equation}
\frac{1}{2}\left( \left\Vert \mathbf{S}(k)\right\Vert^{2}+\left\Vert 
\mathbf{N}(k)\right\Vert^{2}\right) \leq (\mathcal{T}(k))^{2}=\left\Vert 
\mathbf{S}(k)+\mathbf{N}(k)\right\Vert^{2}\leq 2\left( \left\Vert \mathbf{S}%
(k)\right\Vert^{2}+\left\Vert \mathbf{N}(k)\right\Vert^{2}\right) . \label{eq:tsq}
\end{equation}%
From Lemma~\ref{mosum-alt-1}, 
\begin{equation}
\max_{\gamma \leq k\leq T-\gamma }\Vert \mathbf{N}(k)\Vert =O_{P}\left( 
\sqrt{\log (T/\gamma )}\right) . \label{eq:e:size}
\end{equation}%
By definition of $\bm\delta_{j}$, we have 
\begin{equation}
\mathbf{S}(k)=\left\{ 
\begin{array}{ll}
\frac{\gamma -|k-k_{j}|}{\sqrt{2\gamma }}\;\mathsf{Vech}\left( \mathbf{H}%
^{\top }\bm\delta_{j}\mathbf{H}\right) & \text{if \ }k_{j}-\gamma +1\leq
k\leq k_{j}+\gamma - 1, \\ 
\mathbf{0} & \text{if \ }\min_{1\leq j\leq R}|k-k_{j}|\geq \gamma .%
\end{array}%
\right. \label{eq:s:size}
\end{equation}%
Further, thanks to Lemma~\ref{lem:hi}, there exists some event $\mathcal{H}%
_{N,T}$ satisfying $\mathsf{P}(\mathcal{H}_{N,T})\rightarrow 1$ as $\min
(N,T)\rightarrow \infty $ such that on $\mathcal{H}_{N,T}$, 
\begin{align}
\left\Vert \mathsf{Vech}\left( \mathbf{H}^{\top }\bm\delta_{j}\mathbf{H}%
\right) \right\Vert & \geq \frac{1}{2}\left\Vert \mathbf{H}^{\top }\bm\delta
_{j}\mathbf{H}\right\Vert \geq \frac{1}{2}\Lambda_{\min }(\mathbf{H}^{\top }%
\mathbf{H})\Vert \bm\delta_{j}\Vert  \notag \\
& \geq \frac{1}{2}\left( 1-\left\Vert \mathbf{H}-\mathbf{H}_{0}\right\Vert
\right) \Vert \bm\delta_{j}\Vert \geq \frac{1}{4}d_j,  \label{eq:s:lb}
\end{align}
and similarly 
\begin{equation}
\left\Vert \mathsf{Vech}\left( \mathbf{H}^{\top }\bm\delta_{j}\mathbf{H}%
\right) \right\Vert \leq \Lambda_{\max }(\mathbf{H}^{\top }\mathbf{H})\Vert %
\bm\delta_{j}\Vert \leq \frac{3}{2}d_j. \label{eq:s:ub}
\end{equation}

\begin{proof}[Proof of Theorem \ref{thm:consistency} \ref{thm:consistency:one}]

Consider for $j=1,\ldots ,R$, 

\begin{equation*}
\mathcal{S}_{T,j}= \left\{ \mathcal{T}(k_{j})\geq \max \left(
\max_{k:\,|k-k_{j}|>(1-\eta )\gamma }\mathcal{T}(k),\,\widetilde{D}%
_{T,\gamma }(\alpha) \cdot \omega_{T}^{(1)}\right) \right\} \cap \mathcal{H}_{N,T}
\end{equation*}%
and $\mathcal{S}_{T}=\bigcap_{1\leq j\leq R}\mathcal{S}_{T,j}$. Then for any 
$\alpha ,\eta \in (0,1)$, we have 
\begin{align*}
(\mathcal{T}(k_{j}))^{2}& \geq \frac{1}{16}d_j^{2}\gamma +O_{P}(\log
(T/\gamma ))=\frac{1}{16}d_j^{2}\gamma (1+o_{P}(1)), \\
\max_{\substack{ k:\,|k-k_{j}|>(1-\eta )\gamma  \\ 1\leq j\leq R}}(\mathcal{T%
}(k))^{2}& \leq \eta^{2}(\mathcal{T}(k_{j}))^{2}+O_{P}(\log (T/\gamma
))=\eta^{2}(\mathcal{T}(k_{j}))^{2}(1+o_{P}(1))
\end{align*}
under Assumption~\ref{kirch}, by~\eqref{eq:tsq}, \eqref{eq:e:size}, %
\eqref{eq:s:size} and~\eqref{eq:s:lb}, where the $O_{P}$-bounds hold
uniformly over $j$ and $k$. Combined with that $\widetilde{D}_{T,\gamma}(\alpha) \asymp \sqrt{\log (T/\gamma )}$ for any $\alpha \in (0,1)$, we have $\mathsf{P}(\mathcal{S}_{T})\rightarrow 1$ as $\min (N,T)\rightarrow \infty$. Also defining 
\begin{multline*}
\widetilde{\mathcal{S}}_{T,j}=\bigcap_{0\leq q\leq \lfloor 2/\eta \rfloor -2}%
\left[ \left\{ \mathcal{T}(k_{j}+q\eta \gamma /2)\geq \max_{k_{j}+(q+1)\eta
\gamma /2\leq k\leq k_{j}+\gamma }\mathcal{T}(k)\right\} \right. \\
\left. \bigcap \left\{ \mathcal{T}(k_{j}-q\eta \gamma /2)\geq
\max_{k_{j}-\gamma \leq k\leq k_{j}-(q+1)\eta \gamma /2}\mathcal{T}%
(k)\right\} \right] \cap \mathcal{H}_{N,T},
\end{multline*}%
by the analogous arguments, we have $\mathsf{P}(\widetilde{\mathcal{S}}%
_{T,j})\rightarrow 1$ and hence $\mathsf{P}(\widetilde{\mathcal{S}}%
_{T})\rightarrow 1$ where $\widetilde{\mathcal{S}}_{T}=\cap_{j=1}^{R}%
\widetilde{\mathcal{S}}_{T,j}$. On $\mathcal{S}_{T}\cap \widetilde{\mathcal{S%
}}_{T}$, we detect exactly one change point estimator within the radius of $\eta \gamma /2$ for each change point according to the rule~\eqref{eq:eta}.
Further, due to~\eqref{eq:e:size} and~\eqref{eq:s:size}, 
\begin{equation*}
\mathsf{P}\left\{ \max_{\substack{ k:\,|k-k_{j}|\geq \gamma  \\ 1\leq j\leq
R }}\mathcal{T}(k)>\widetilde{D}_{T,\gamma }(\alpha )\cdot \omega
_{T}^{(1)}\right\} \rightarrow 0 \text{ \ as \ } \min(N, T) \rightarrow \infty,
\end{equation*}%
which guarantees that no estimator is detected outside the radius of $\gamma 
$ from each change point. Altogether, the above arguments show that 
\begin{equation*}
\mathsf{P}\left( \widehat{R}=R;\,\max_{1\leq j\leq R}|\widehat{k}%
_{j}-k_{j}|\leq \eta \gamma /2\right) \rightarrow 1 \text{ \ as \ } \min(N, T) \rightarrow \infty.
\end{equation*}
\end{proof}

\begin{proof}[Proof of Theorem \ref{thm:consistency} \ref{thm:consistency:two}]
For each $j$, recall that $|\widehat{k}_{j}-k_{j}|\leq \gamma$, on $\mathcal{S}_{T}\cap \widetilde{\mathcal{S}}_{T}$. WLOG, suppose that $\widehat{k}_{j}\leq k_{j}$ and define $\widetilde{\mathcal{T}}_{j}(k)=(\mathcal{T}(k))^{2}-(\mathcal{T}(k_{j}))^{2}$. Then, recalling $\omega_{T}$
defined in Lemma~\ref{mosum-alt-2}, let us consider 
\begin{multline*}
\left\{ Cd_j^{2}(\widehat{k}_{j}-k_{j})\leq -\omega_{T}^{2}\right\}
\subset \left\{ \max_{k_{j}-\gamma +1\leq k\leq k_{j}-Cd_j^{2}\omega
_{T}^{2}}\widetilde{\mathcal{T}}_{j}(k)\geq \max_{k_{j}-Cd_j^{2}\omega
_{T}^{2}+1\leq k\leq k_{j}+\gamma }\widetilde{\mathcal{T}}_{j}(k)\right\} \\
\subset \left\{ \max_{k_{j}-\gamma +1\leq k\leq k_{j}-Cd_j^{2}\omega
_{T}^{2}}\widetilde{\mathcal{T}}_{j}(k)\geq 0\right\}
\end{multline*}%
for some fixed $C\in (0,\infty )$. We can decompose $\widetilde{\mathcal{T}}%
_{j}(k)$ as 
\begin{align*}
\widetilde{\mathcal{T}}_{j}(k)
&= \left( \mathbf{N}(k)-\mathbf{N}(k_{j})+\mathbf{S}(k)- \mathbf{S}(k_{j})\right)^\top \left( \mathbf{N}(k)+\mathbf{N}(k_{j})+\mathbf{S}%
(k)+\mathbf{S}(k_{j})\right) \\
&=-\left( \mathbf{S}(k_{j})-\mathbf{S}(k)\right)^{\top}\left( \mathbf{S}%
(k_{j})+\mathbf{S}(k)\right) +\left( \mathbf{N}(k)-\mathbf{N}(k_{j})\right)
\left( \mathbf{S}(k)+\mathbf{S}(k_{j})\right) 
\\
&+\left( \mathbf{N}(k)+\mathbf{N%
}(k_{j})\right) \left( \mathbf{S}(k)-\mathbf{S}(k_{j})\right) +\left( \mathbf{N}(k)-\mathbf{N}(k_{j})\right) \left( \mathbf{N}(k_{j})+%
\mathbf{N}(k)\right) 
\\
&=:
\widetilde{\mathcal{T}}_{j,1}(k)+\widetilde{\mathcal{T}}_{j,2}(k)+\widetilde{\mathcal{T}}_{j,3}(k)+\widetilde{\mathcal{T}}_{j,4}(k).
\end{align*}
From~\eqref{eq:s:size}, \eqref{eq:s:lb} and~\eqref{eq:s:ub}, we have $\widetilde{\mathcal{T}}_{j,1}(k)<0$ and 
\begin{align}
& \left\vert \widetilde{\mathcal{T}}_{j,1}(k)\right\vert =\frac{(2\gamma -|k-k_{j}|)|k-k_{j}|}{%
2\gamma }\left\Vert \mathsf{Vech}(\mathbf{H}^{\top }\bm\delta_{j}\mathbf{H}%
)\right\Vert^{2}\geq \frac{1}{2}\left\Vert \mathsf{Vech}(\mathbf{H}^{\top }%
\bm\delta_{j}\mathbf{H})\right\Vert^{2}|k-k_{j}|,  \notag \\
& \left\Vert \mathbf{S}(k_{j})-\mathbf{S}(k)\right\Vert =\frac{|k-k_{j}|}{%
\sqrt{2\gamma }}\left\Vert \mathsf{Vech}(\mathbf{H}^{\top }\bm\delta_{j}%
\mathbf{H})\right\Vert ,  \notag \\
& \left\Vert \mathbf{S}(k_{j})+\mathbf{S}(k)\right\Vert \leq \sqrt{2\gamma }%
\left\Vert \mathsf{Vech}(\mathbf{H}^{\top }\bm\delta_{j}\mathbf{H}%
)\right\Vert .  \label{eq:thm:consist:ub}
\end{align}
Then, using the arguments from the proof of Theorem~3.2 in \cite{kirch}, 
\begin{align*}
& \mathsf{P}\left( \max_{k_{j}-\gamma +1\leq k\leq k_{j}-Cd_j^{2}\omega
_{T}^{2}}\widetilde{T}_{j}(k)\geq 0,\,\mathcal{S}_{T}\cap \widetilde{%
\mathcal{S}}_{T}\right) \\
\leq & \,\mathsf{P}\left( \max_{k_{j}-\gamma +1\leq k\leq
k_{j}-Cd_j^{2}\omega_{T}^{2}}\left\vert \frac{\widetilde{T}_{j,2}(k)}{\widetilde{T}_{j,1}(k)}+%
\frac{\widetilde{T}_{j,3}(k)}{\widetilde{T}_{j,1}(k)}+\frac{\widetilde{T}_{j,4}(k)}{\widetilde{T}_{j,1}(k)}\right\vert \geq
1,\,\mathcal{S}_{T}\cap \widetilde{\mathcal{S}}_{T}\right) \\
\le & \, 2\mathsf{P}\left( \max_{k_{j}-\gamma +1\leq k\leq
k_{j}-Cd_j^{2}\omega_{T}^{2}}\frac{8\sqrt{2\gamma }\Vert \mathbf{N}(k)-%
\mathbf{N}(k_{j})\Vert }{d_j|k-k_{j}|}\geq \frac{1}{3}\right) \\
& + 2\mathsf{P}\left( \max_{k_{j}-\gamma +1\leq k\leq k_{j}-Cd_j^{2}\omega
_{T}^{2}}\frac{4\Vert \mathbf{N}(k)+\mathbf{N}(k_{j})\Vert }{d_j\sqrt{%
2\gamma }}\geq \frac{1}{3}\right) .
\end{align*}

By~\eqref{eq:e:size} and Assumption~\ref{kirch}, 
\begin{equation*}
\max_{k_{j}-\gamma +1\leq k\leq k_{j}-Cd_j^{2}\omega_{T}^{2}}\frac{\Vert 
\mathbf{N}(k)+\mathbf{N}(k_{j})\Vert }{d_j\sqrt{2\gamma }}\leq
\max_{\gamma \leq k\leq T-\gamma }\frac{2\Vert \mathbf{N}(k)\Vert }{d_j%
\sqrt{2\gamma }}=O_{P}\left( \frac{\sqrt{\log (T/\gamma )}}{d_j\sqrt{%
2\gamma }}\right) =o_{P}(1).
\end{equation*}%
Also, we have
\begin{multline*}
\sqrt{2\gamma }\left\Vert \mathbf{N}(k)-\mathbf{N}(k_{j})\right\Vert \leq
2\left\Vert \sum_{t=k+1}^{k_{j}}\mathsf{Vech}\left( \widehat{\mathbf{g}}_{t}%
\widehat{\mathbf{g}}_{t}^{\top }-\mathbf{H}^{\top }\mathsf{E}(\mathbf{g}_{t}%
\mathbf{g}_{t}^{\top })\mathbf{H}\right) \right\Vert \\
+\left\Vert \sum_{t=k-\gamma +1}^{k_{j}-\gamma }\mathsf{Vech}\left( \widehat{\mathbf{g}}_{t}\widehat{\mathbf{g}}_{t}^{\top }-\mathbf{H}^{\top }\mathsf{E}(%
\mathbf{g}_{t}\mathbf{g}_{t}^{\top })\mathbf{H}\right) \right\Vert
+\left\Vert \sum_{t=k+\gamma +1}^{k_{j}+\gamma }\mathsf{Vech}\left( \widehat{\mathbf{g}}_{t}\widehat{\mathbf{g}}_{t}^{\top }-\mathbf{H}^{\top }\mathsf{E}(%
\mathbf{g}_{t}\mathbf{g}_{t}^{\top })\mathbf{H}\right) \right\Vert.
\end{multline*}
Recalling the definition of $\mathcal{M}_{T}^{(\ell )}$ from Lemma~\ref%
{mosum-alt-2}, we have 
\begin{align*}
& \mathsf{P}\left( \max_{k_{j}-\gamma +1\leq k\leq k_{j}-Cd_j^{2}\omega
_{T}^{2}}\frac{8\sqrt{2\gamma }\Vert \mathbf{N}(k)-\mathbf{N}(k_{j})\Vert }{%
d_j|k-k_{j}|}\geq \frac{1}{3}\right) \\
=& \,\mathsf{P}\left( \max_{\ell \in \{0,\pm 1\}}\max_{k_{j}-\gamma +1\leq
k\leq k_{j}-Cd_j^{2}\omega_{T}^{2}}\frac{\sqrt{Cd_j^{-2}\omega_{T}}}{%
|k-k_{j}|}\left\Vert \sum_{t=k+\ell \gamma +1}^{k_{j}+\ell \gamma }\mathsf{%
Vech}\left( \widehat{\mathbf{g}}_{t}\widehat{\mathbf{g}}_{t}^{\top }-\mathbf{%
H}^{\top }\mathsf{E}(\mathbf{g}_{t}\mathbf{g}_{t}^{\top })\mathbf{H}\right)
\right\Vert \geq \frac{\sqrt{C}\omega_{T}}{96}\right) \\
\leq & \,\mathsf{P}\left( \cap_{\ell \in \{0,\pm 1\}}\mathcal{M}_{T}^{(\ell
)}\right) +o(1)=o(1),
\end{align*}
for large enough $C$. Altogether, we have the RHS of~%
\eqref{eq:thm:consist:ub} bounded as $o(1)$. Analogous arguments apply to
the case where $\widehat{k}_{j}>k_{j}$. Finally, setting $\omega
_{T}^{(2)}=C\omega_{T}^{2}$ concludes the proof.
\end{proof}

\clearpage

\section{Additional simulation results}
\label{app:sim}

\subsection{Additional results obtained under~\ref{m:two}}

We additionally report the histograms of the change point estimators obtained by MOSUM-diagonal, BSCOV \citep{li2023detection} and BDH \citep{baiduan} on realisations generated under the scenario~\ref{m:two} in Section~\ref{simulations} with $N \in \{200, 500\}$, see Figures~\ref{fig:m3:p200depFALSE}--\ref{fig:m3:p500depTRUE}.

\begin{figure}[h!t!]
\centering
\includegraphics[width = 1\linewidth]{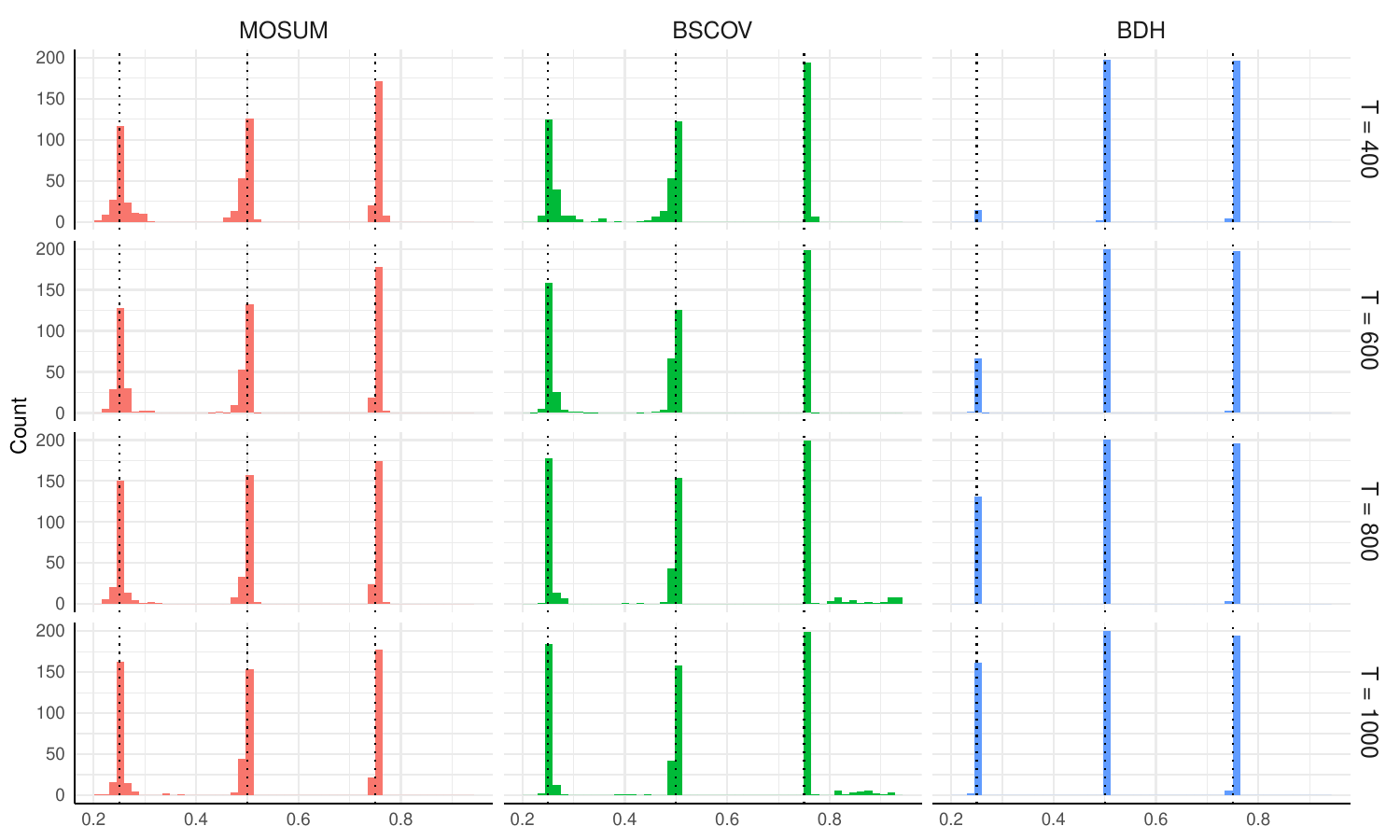}
\caption{\ref{m:two} Histogram of the change point estimators returned by MOSUM-diagonal, BSCOV and BDH when $N = 200$, $(\rho_f, \rho_e) = (0, 0)$ and varying $T \in \{400, 600, 800, 1000\}$ (top to bottom). The scaled locations of the true change points, $k_j/T$, at $(1/4, 1/2, 3/4)$ are marked by vertical dotted lines.}
\label{fig:m3:p200depFALSE}
\end{figure}

\begin{figure}[h!t!]
\centering
\includegraphics[width = 1\linewidth]{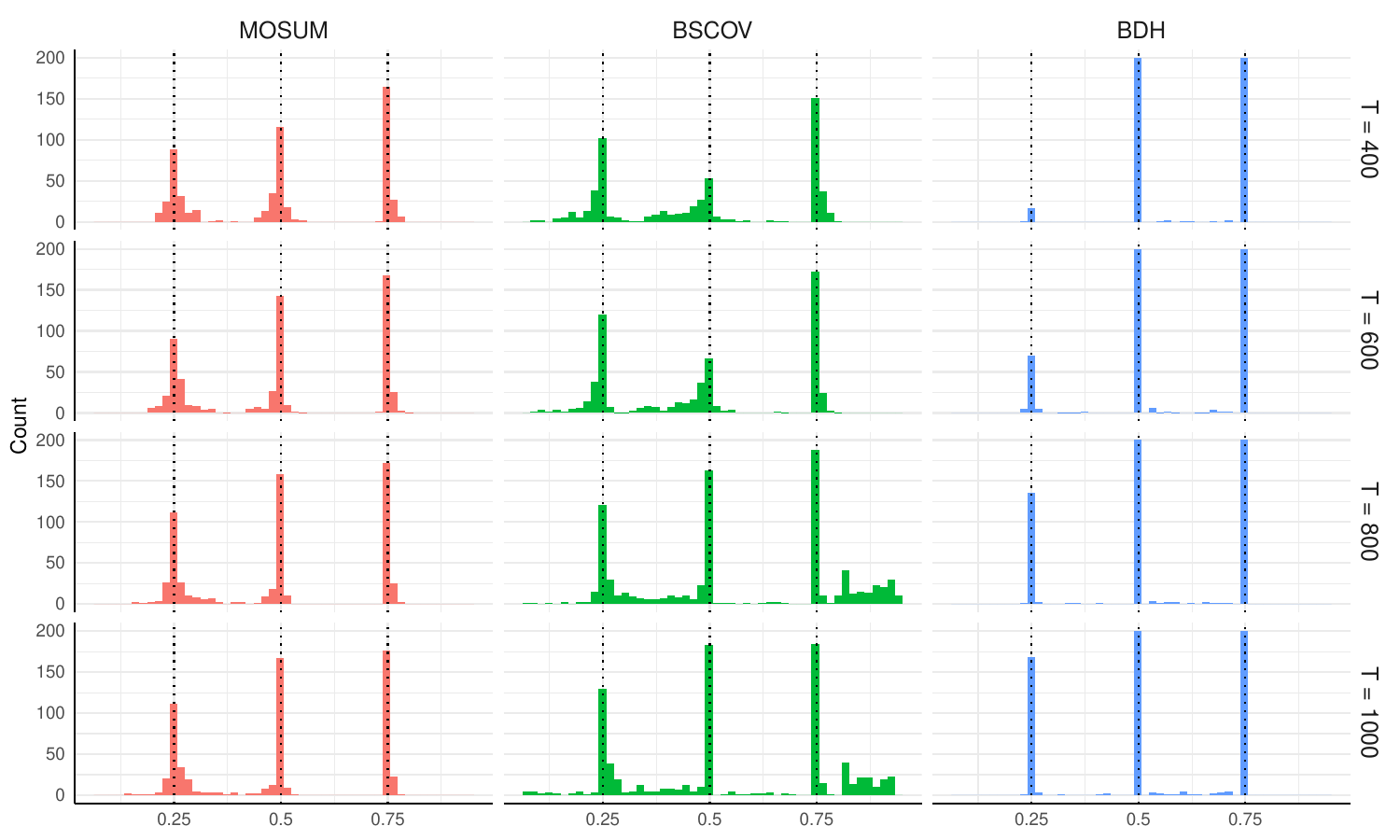}
\caption{\ref{m:two} Histogram of the change point estimators returned by MOSUM-diagonal, BSCOV and BDH when $N = 200$, $(\rho_f, \rho_e) = (0.7, 0.3)$ and varying $T \in \{400, 600, 800, 1000\}$ (top to bottom). The scaled locations of the true change points, $k_j/T$, at $(1/4, 1/2, 3/4)$ are marked by vertical dotted lines.}
\label{fig:m3:p200depTRUE}
\end{figure}

\begin{figure}[h!t!]
\centering
\includegraphics[width = 1\linewidth]{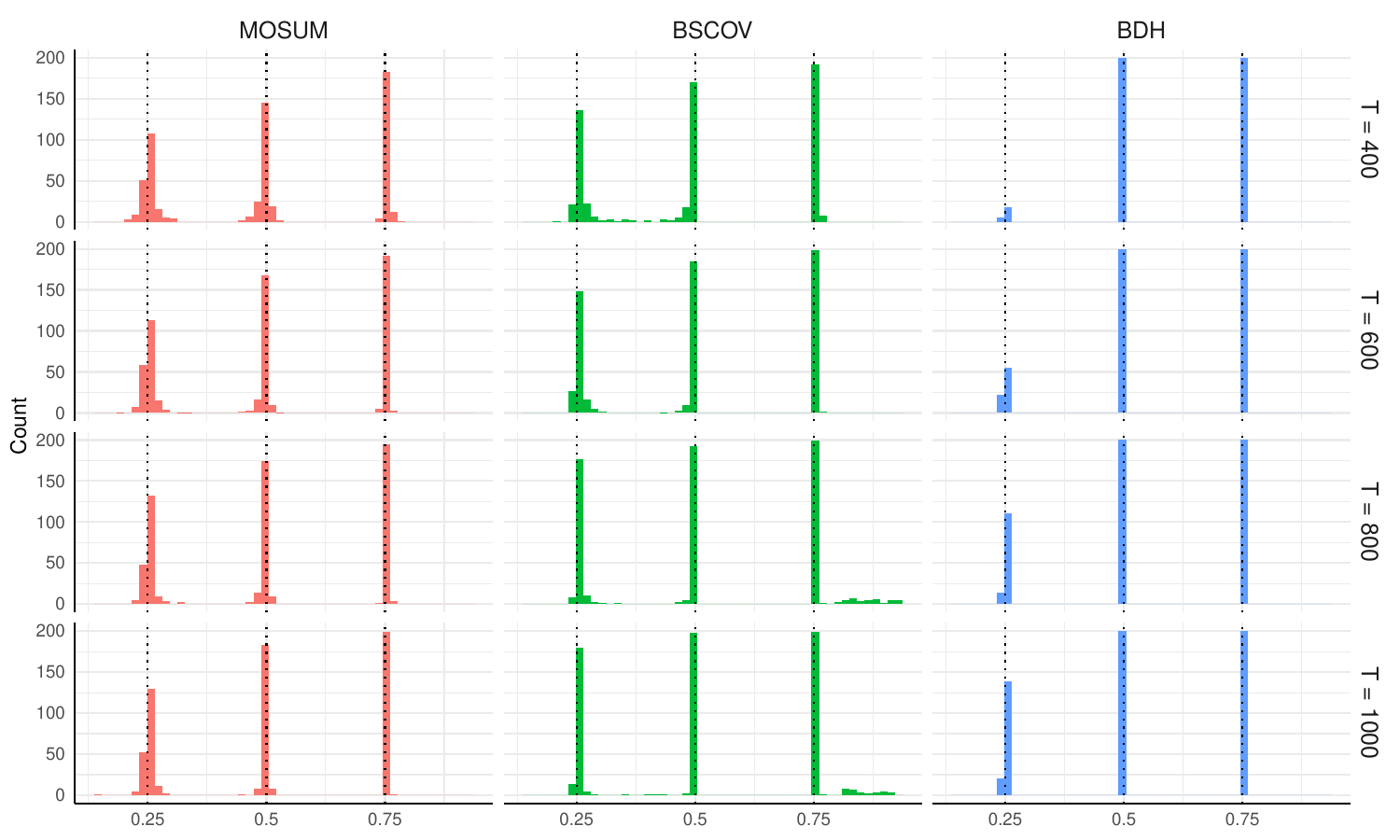}
\caption{\ref{m:two} Histogram of the change point estimators returned by MOSUM-diagonal, BSCOV and BDH when $N = 500$, $(\rho_f, \rho_e) = (0, 0)$ and varying $T \in \{400, 600, 800, 1000\}$ (top to bottom). The scaled locations of the true change points, $k_j/T$, at $(1/4, 1/2, 3/4)$ are marked by vertical dotted lines.}
\label{fig:m3:p500depFALSE}
\end{figure}

\begin{figure}[h!t!]
\centering
\includegraphics[width = 1\linewidth]{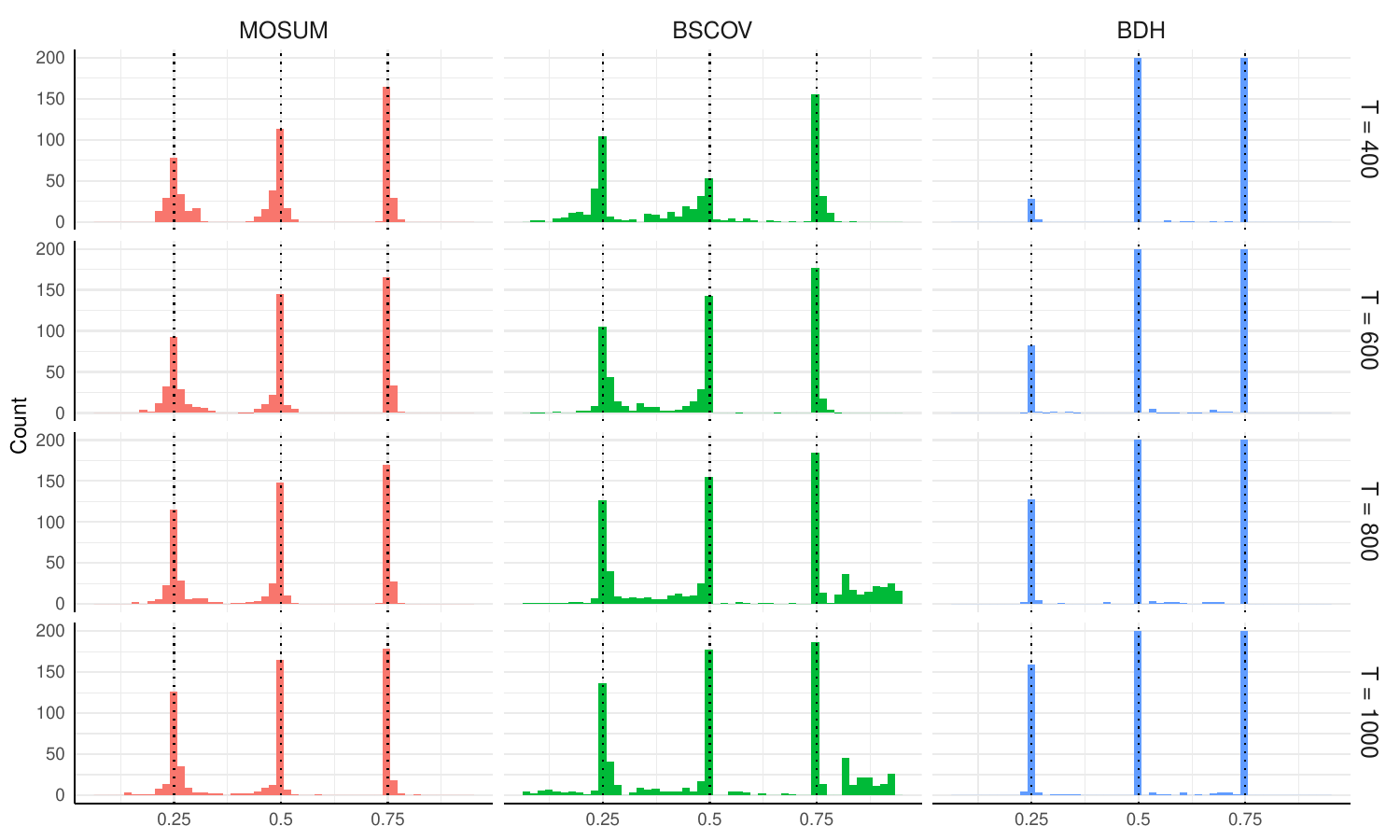}
\caption{\ref{m:two} Histogram of the change point estimators returned by MOSUM-diagonal, BSCOV and BDH when $N = 500$, $(\rho_f, \rho_e) = (0.7, 0.3)$ and varying $T \in \{400, 600, 800, 1000\}$ (top to bottom). The scaled locations of the true change points, $k_j/T$, at $(1/4, 1/2, 3/4)$ are marked by vertical dotted lines.}
\label{fig:m3:p500depTRUE}
\end{figure}

\clearpage 

\subsection{Choice of $\omega^{(1)}_T$}
\label{sec:kappa}

As discussed in Section~\ref{sec:tuning}, we set the threshold as $D_{T, \gamma} = \widetilde{D}_{T, \gamma}(\alpha) \cdot \omega^{(1)}_T$ with $\omega^{(1)}_T = \log^\kappa(T/\gamma)$ for some $\kappa \ge 0$.
In this section, we demonstrate that the detection performance proposed MOSUM procedure is less sensitive to the choice of $\kappa$ within a reasonable range, see Figures~\ref{fig:kappa:p100:depFALSE}--\ref{fig:kappa:p500:depTRUE} which plot the histograms of the change point estimators detected by MOSUM-diagonal, with varying $\kappa \in \{0, 0.1, 0.2, 0.3\}$, over $200$ realisations generated under~\ref{m:two} in Section~\ref{simulations}.
We complement these results with those obtained in the no change point scenario of~\ref{m:zero}, see Table~\ref{tab:kappa}, where it shows that $\kappa = 0.2$ is a choice that balances between good detection performance as well as in keeping the false positives at bay when the data contain no change point, particularly when serial dependence is present in the data.

\begin{figure}[h!t!]
\centering
\includegraphics[width = 1\textwidth]{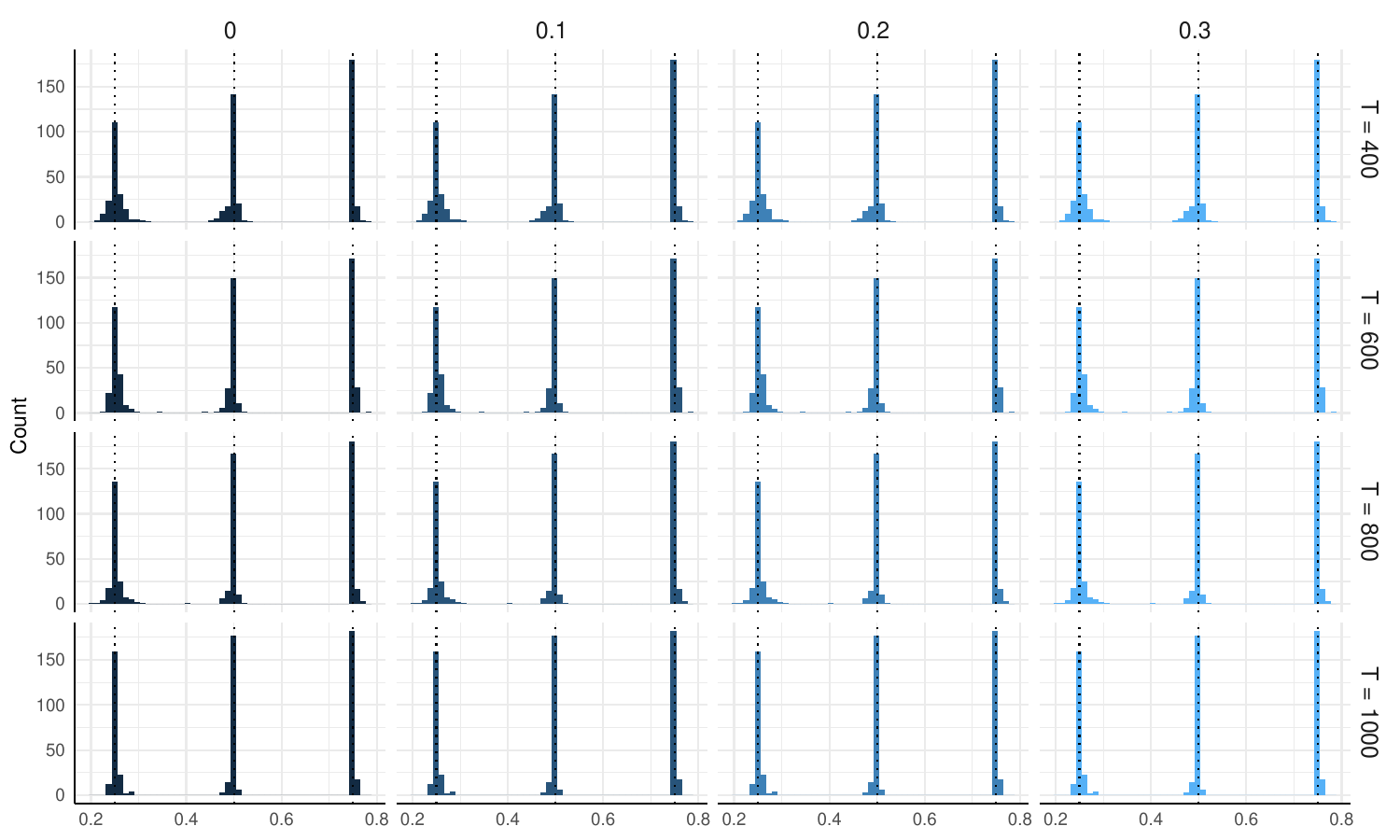}
\caption{\ref{m:two} Histogram of the change point estimators returned by MOSUM-diagonal with $\kappa \in \{0, 0.1, 0.2, 0.3\}$ (left to right) when $N = 100$, $(\rho_f, \rho_e) = (0, 0)$ and varying $T \in \{400, 600, 800, 1000\}$ (top to bottom). The scaled locations of the true change points, $k_j/T$, at $(1/4, 1/2, 3/4)$ are marked by vertical dotted lines.}
\label{fig:kappa:p100:depFALSE}
\end{figure} 

\begin{figure}[h!t!]
\centering
\includegraphics[width = 1\textwidth]{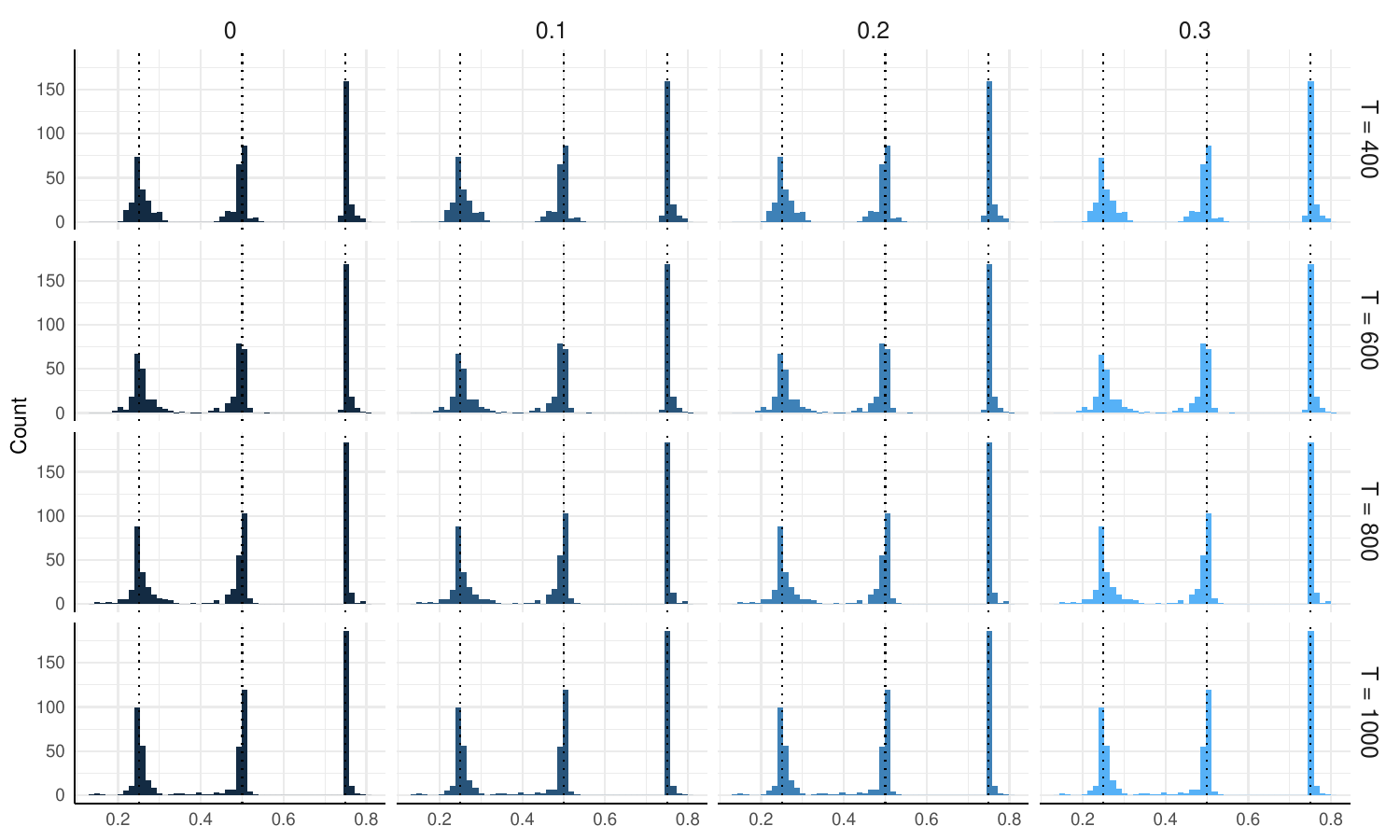}
\caption{\ref{m:two} Histogram of the change point estimators returned by MOSUM-diagonal with $\kappa \in \{0, 0.1, 0.2, 0.3\}$ (left to right) when $N = 100$, $(\rho_f, \rho_e) = (0.7, 0.3)$ and varying $T \in \{400, 600, 800, 1000\}$ (top to bottom). The scaled locations of the true change points, $k_j/T$, at $(1/4, 1/2, 3/4)$ are marked by vertical dotted lines.}
\label{fig:kappa:p100:depTRUE}
\end{figure} 

\begin{figure}[h!t!]
\centering
\includegraphics[width = 1\textwidth]{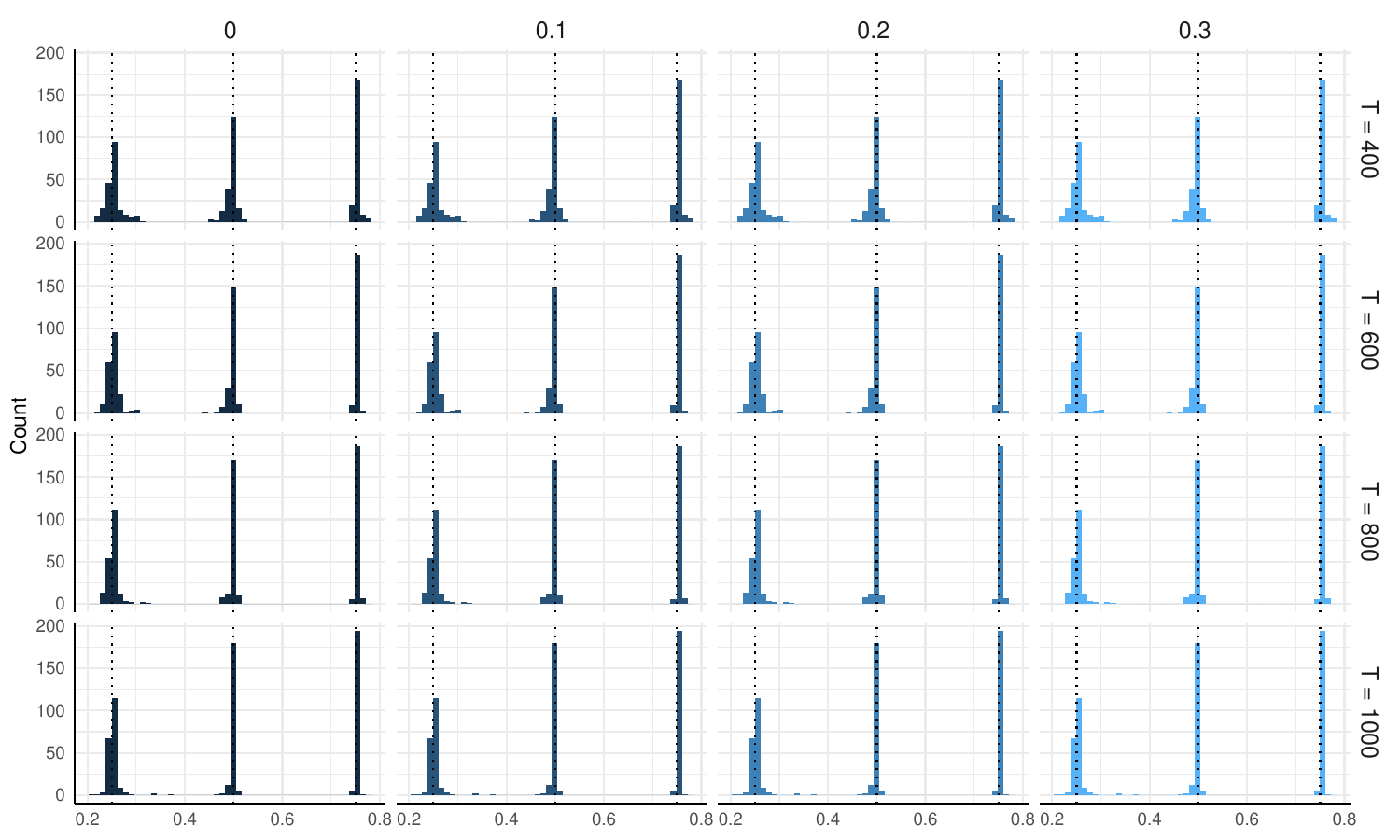}
\caption{\ref{m:two} Histogram of the change point estimators returned by MOSUM-diagonal with $\kappa \in \{0, 0.1, 0.2, 0.3\}$ (left to right) when $N = 200$, $(\rho_f, \rho_e) = (0, 0)$ and varying $T \in \{400, 600, 800, 1000\}$ (top to bottom). The scaled locations of the true change points, $k_j/T$, at $(1/4, 1/2, 3/4)$ are marked by vertical dotted lines.}
\label{fig:kappa:p200:depFALSE}
\end{figure} 

\begin{figure}[h!t!]
\centering
\includegraphics[width = 1\textwidth]{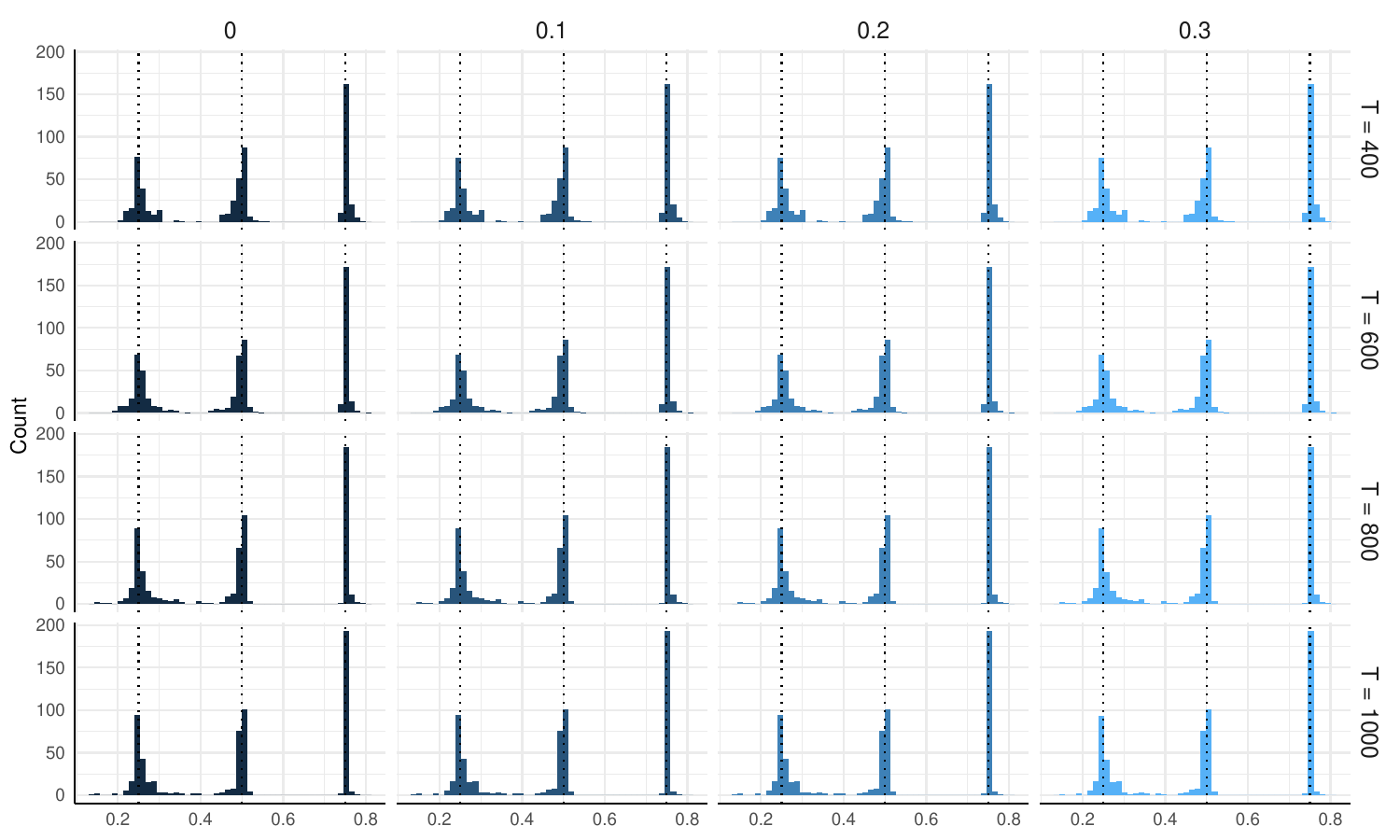}
\caption{\ref{m:two} Histogram of the change point estimators returned by MOSUM-diagonal with $\kappa \in \{0, 0.1, 0.2, 0.3\}$ (left to right) when $N = 200$, $(\rho_f, \rho_e) = (0.7, 0.3)$ and varying $T \in \{400, 600, 800, 1000\}$ (top to bottom). The scaled locations of the true change points, $k_j/T$, at $(1/4, 1/2, 3/4)$ are marked by vertical dotted lines.}
\label{fig:kappa:p200:depTRUE}
\end{figure} 

\begin{figure}[h!t!]
\centering
\includegraphics[width = 1\textwidth]{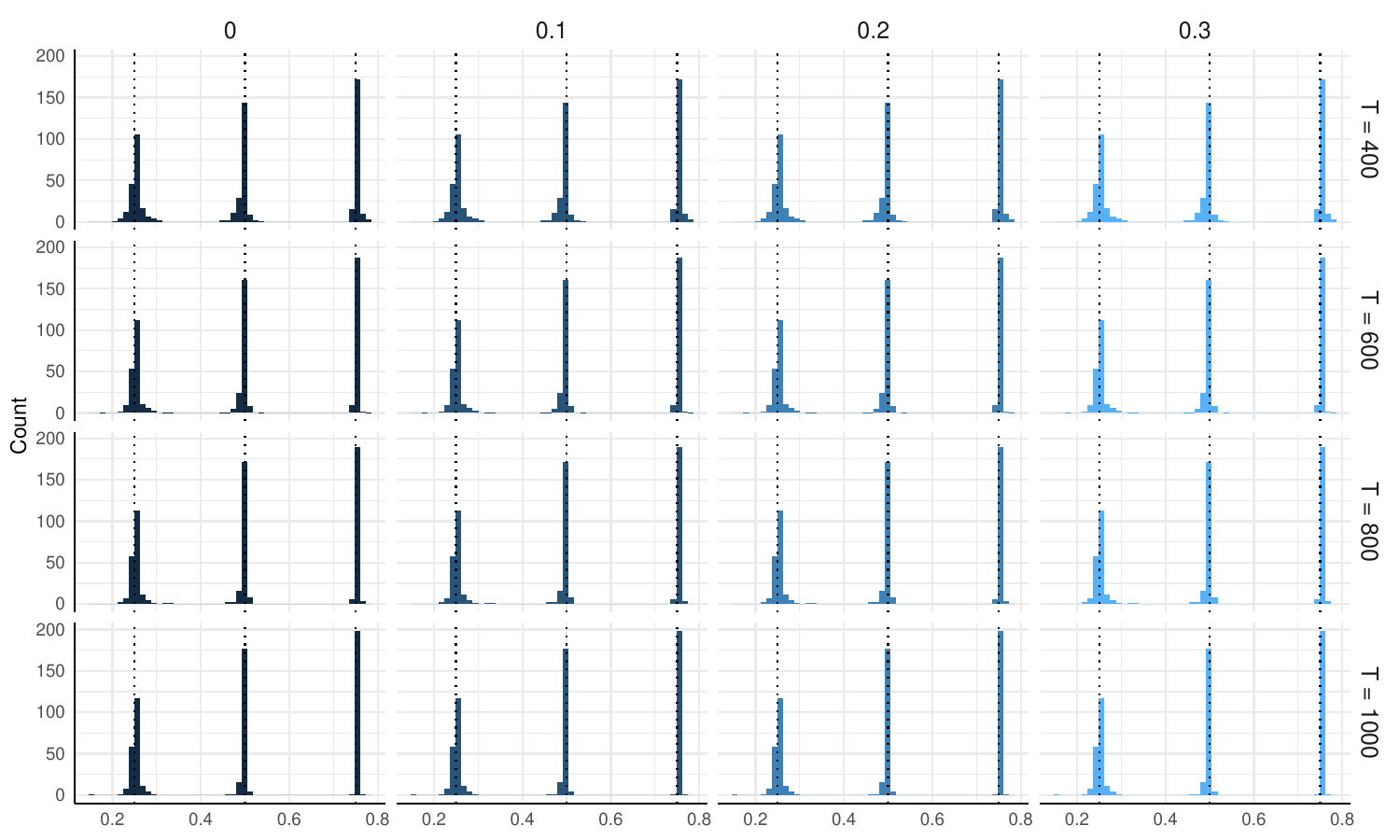}
\caption{\ref{m:two} Histogram of the change point estimators returned by MOSUM-diagonal with $\kappa \in \{0, 0.1, 0.2, 0.3\}$ (left to right) when $N = 500$, $(\rho_f, \rho_e) = (0, 0)$ and varying $T \in \{400, 600, 800, 1000\}$ (top to bottom). The scaled locations of the true change points, $k_j/T$, at $(1/4, 1/2, 3/4)$ are marked by vertical dotted lines.}
\label{fig:kappa:p500:depFALSE}
\end{figure} 

\begin{figure}[h!t!]
\centering
\includegraphics[width = 1\textwidth]{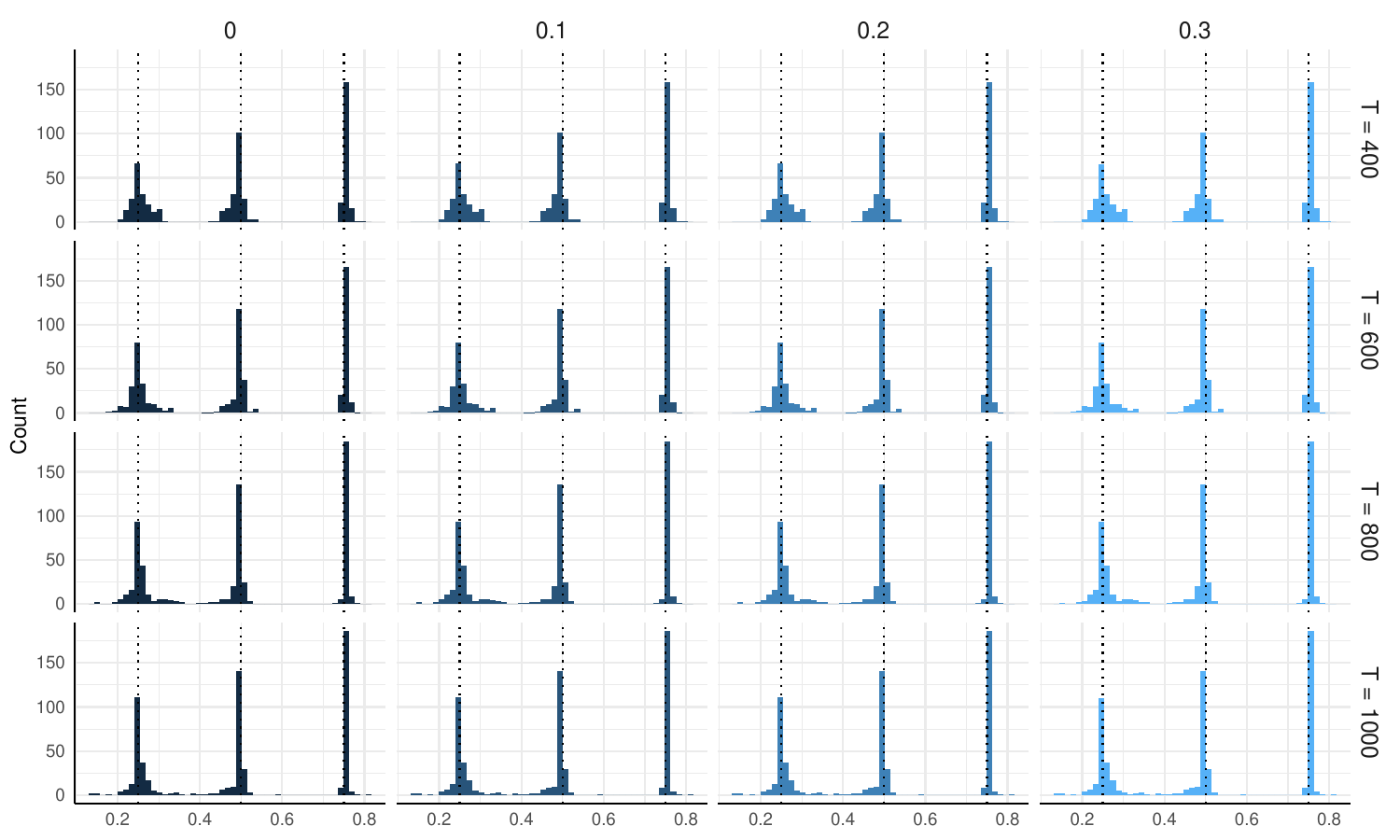}
\caption{\ref{m:two} Histogram of the change point estimators returned by MOSUM-diagonal with $\kappa \in \{0, 0.1, 0.2, 0.3\}$ (left to right) when $N = 500$, $(\rho_f, \rho_e) = (0.7, 0.3)$ and varying $T \in \{400, 600, 800, 1000\}$ (top to bottom). The scaled locations of the true change points, $k_j/T$, at $(1/4, 1/2, 3/4)$ are marked by vertical dotted lines.}
\label{fig:kappa:p500:depTRUE}
\end{figure} 

\clearpage

\begin{table}[h!t!p!]
\caption{\ref{m:zero} with $R = 0$: Distribution of $\widehat R - R$ returned by MOSUM-diagonal over $200$ realisations with varying $\kappa \in \{0, 0.1, 0.2, 0.3\}$.}
\label{tab:kappa}
\centering
{\footnotesize
\begin{tabular}{rrrcccccc}
\toprule 
&	&	&	\multicolumn{3}{c}{$(\rho_f, \rho_e) = (0, 0)$}  &			\multicolumn{3}{c}{$(\rho_f, \rho_e) = (0.7, 0.3)$} 			\\	\cmidrule(lr){4-6} \cmidrule(lr){7-9}	
&	&	&	\multicolumn{3}{c}{$\wh{R} - R$}  &			\multicolumn{3}{c}{$\wh{R} - R$} 			\\		
$n$ &	$p$ &	$\kappa$ &	$0$ &	$1$ &	$\ge 2$ &	$0$ &	$1$ &	$\ge 2$ 	\\	\cmidrule(lr){1-3} \cmidrule(lr){4-6} \cmidrule(lr){7-9}	
$400$ &	$100$ &	$0$ &	0.955 &	0.045 &	0 &	0.75 &	0.22 &	0.03	\\		
&	&	$0.1$ &	0.985 &	0.015 &	0 &	0.84 &	0.135 &	0.025	\\		
&	&	$0.2$ &	0.985 &	0.015 &	0 &	0.895 &	0.08 &	0.025	\\		
&	&	$0.3$ &	0.985 &	0.015 &	0 &	0.955 &	0.03 &	0.015	\\		
&	$200$ &	$0$ &	0.96 &	0.04 &	0 &	0.745 &	0.22 &	0.035	\\		
&	&	$0.1$ &	0.99 &	0.01 &	0 &	0.79 &	0.19 &	0.02	\\		
&	&	$0.2$ &	0.995 &	0.005 &	0 &	0.88 &	0.11 &	0.01	\\		
&	&	$0.3$ &	1 &	0 &	0 &	0.94 &	0.055 &	0.005	\\		
&	$500$ &	$0$ &	0.95 &	0.05 &	0 &	0.725 &	0.25 &	0.025	\\		
&	&	$0.1$ &	0.985 &	0.015 &	0 &	0.845 &	0.14 &	0.015	\\		
&	&	$0.2$ &	0.99 &	0.01 &	0 &	0.88 &	0.11 &	0.01	\\		
&	&	$0.3$ &	0.99 &	0.01 &	0 &	0.925 &	0.07 &	0.005	\\	\cmidrule(lr){1-3} \cmidrule(lr){4-6} \cmidrule(lr){7-9}	
$600$ &	$100$ &	$0$ &	0.94 &	0.055 &	0.005 &	0.635 &	0.275 &	0.09	\\		
&	&	$0.1$ &	0.97 &	0.03 &	0 &	0.755 &	0.195 &	0.05	\\		
&	&	$0.2$ &	0.99 &	0.01 &	0 &	0.855 &	0.125 &	0.02	\\		
&	&	$0.3$ &	0.995 &	0.005 &	0 &	0.92 &	0.065 &	0.015	\\		
&	$200$ &	$0$ &	0.97 &	0.02 &	0.01 &	0.62 &	0.26 &	0.12	\\		
&	&	$0.1$ &	0.975 &	0.015 &	0.01 &	0.755 &	0.18 &	0.065	\\		
&	&	$0.2$ &	0.99 &	0.01 &	0 &	0.845 &	0.14 &	0.015	\\		
&	&	$0.3$ &	0.995 &	0.005 &	0 &	0.92 &	0.07 &	0.01	\\		
&	$500$ &	$0$ &	0.94 &	0.055 &	0.005 &	0.625 &	0.265 &	0.11	\\		
&	&	$0.1$ &	0.97 &	0.025 &	0.005 &	0.755 &	0.195 &	0.05	\\		
&	&	$0.2$ &	0.985 &	0.015 &	0 &	0.855 &	0.11 &	0.035	\\		
&	&	$0.3$ &	0.99 &	0.01 &	0 &	0.925 &	0.06 &	0.015	\\	\cmidrule(lr){1-3} \cmidrule(lr){4-6} \cmidrule(lr){7-9}	
$800$ &	$100$ &	$0$ &	0.94 &	0.06 &	0 &	0.66 &	0.3 &	0.04	\\		
&	&	$0.1$ &	0.97 &	0.03 &	0 &	0.795 &	0.185 &	0.02	\\		
&	&	$0.2$ &	1 &	0 &	0 &	0.905 &	0.085 &	0.01	\\		
&	&	$0.3$ &	1 &	0 &	0 &	0.975 &	0.015 &	0.01	\\		
&	$200$ &	$0$ &	0.95 &	0.05 &	0 &	0.645 &	0.31 &	0.045	\\		
&	&	$0.1$ &	0.975 &	0.025 &	0 &	0.82 &	0.17 &	0.01	\\		
&	&	$0.2$ &	1 &	0 &	0 &	0.94 &	0.05 &	0.01	\\		
&	&	$0.3$ &	1 &	0 &	0 &	0.99 &	0.01 &	0	\\		
&	$500$ &	$0$ &	0.96 &	0.04 &	0 &	0.63 &	0.315 &	0.055	\\		
&	&	$0.1$ &	1 &	0 &	0 &	0.805 &	0.175 &	0.02	\\		
&	&	$0.2$ &	1 &	0 &	0 &	0.92 &	0.075 &	0.005	\\		
&	&	$0.3$ &	1 &	0 &	0 &	0.97 &	0.03 &	0	\\		
\cmidrule(lr){1-3} \cmidrule(lr){4-6} \cmidrule(lr){7-9}	
$1000$ &	$100$ &	$0$ &	0.935 &	0.06 &	0.005 &	0.65 &	0.26 &	0.09	\\		
&	&	$0.1$ &	0.98 &	0.02 &	0 &	0.805 &	0.16 &	0.035	\\		
&	&	$0.2$ &	1 &	0 &	0 &	0.895 &	0.1 &	0.005	\\		
&	&	$0.3$ &	1 &	0 &	0 &	0.97 &	0.03 &	0	\\		
&	$200$ &	$0$ &	0.96 &	0.04 &	0 &	0.685 &	0.23 &	0.085	\\		
&	&	$0.1$ &	0.995 &	0.005 &	0 &	0.83 &	0.13 &	0.04	\\		
&	&	$0.2$ &	1 &	0 &	0 &	0.9 &	0.09 &	0.01	\\		
&	&	$0.3$ &	1 &	0 &	0 &	0.96 &	0.04 &	0	\\		
&	$500$ &	$0$ &	0.97 &	0.03 &	0 &	0.67 &	0.23 &	0.1	\\		
&	&	$0.1$ &	0.985 &	0.015 &	0 &	0.815 &	0.15 &	0.035	\\		
&	&	$0.2$ &	1 &	0 &	0 &	0.91 &	0.085 &	0.005	\\		
&	&	$0.3$ &	1 &	0 &	0 &	0.975 &	0.025 &	0	\\	\bottomrule		
\end{tabular}
}
\end{table}

\end{document}